%% file: theNRconjecture.tex
\documentclass[11pt]{article}
\usepackage[margin=1in]{geometry}
\usepackage{amsmath,amsthm,amssymb}
\usepackage{graphicx}
\usepackage{tikz}
\usetikzlibrary{decorations.pathmorphing}
\usetikzlibrary{matrix}

\usepackage{todonotes}
\usepackage[pagebackref=true]{hyperref}
\usepackage{hyperref}

\usepackage[T1]{fontenc}
\usepackage{lmodern}

\usepackage[theorems,skins,breakable]{tcolorbox}

\tcbset{
  errorstyle/.style={fonttitle=\bfseries\upshape, colback=red!50!black!30,
    colframe=black!10, oversize, breakable, left=0mm, right=0mm, top=0mm, bottom=0mm,
    middle=0mm},
  amstyle/.style={fonttitle=\bfseries\upshape, colback=yellow!10!white,
    colframe=red!10!black!10, oversize, breakable, left=0mm, right=0mm, top=0mm,
    bottom=0mm, middle=0mm},
  gcstyle/.style={fonttitle=\bfseries\upshape, colback=gray!10!white,
    colframe=red!10!black!10, oversize, breakable, left=0mm, right=0mm, top=0mm,
    bottom=0mm, middle=0mm},
  draftstyle/.style={fonttitle=\bfseries\upshape, colback=orange!10!white,
    colframe=black!0, oversize, breakable, left=0mm, right=0mm, top=0mm, bottom=0mm,
    middle=0mm},
}

\newtheorem{theorem}{Theorem}
\newtheorem*{theorem*}{Theorem}
\newtheorem{corollary}[theorem]{Corollary}

\newtheorem*{conjecture*}{Conjecture}
\newtheorem{lemma}[theorem]{Lemma}

\newtheorem*{claim*}{Claim}
\newtheorem{proposition}[theorem]{Proposition}
\newtheorem{observation}[theorem]{Observation}

\newtheorem{definition}[theorem]{Definition}

\theoremstyle{definition}

\theoremstyle{remark}
\newtheorem*{remark}{Remark}

\newcommand{\mech}{\text{\sc Mech}}
\newcommand{\opt}{\text{\sc Opt}}

\newcommand{\x}{\mathbf{x}}
\newcommand{\R}{\mathbb{R}}

\newcounter{note}[section]

\newcommand{\ZP}{\mathbb{Z}_{>0}}
\newcommand{\fl}[2]{\lfloor #1 \rfloor_{#2}}
\newcommand{\ceil}[2]{\lceil #1 \rceil_{#2}}
\newcommand{\apsi}{\underline\psi}

\title{A proof of the Nisan-Ronen conjecture}

\author{George Christodoulou\thanks{School of Informatics,
     Aristotle University of Thessaloniki, Greece.  Email:
     \texttt{gichristo@csd.auth.gr} } \and Elias Koutsoupias \thanks{Department of Computer Science, University
     of Oxford, UK. Email: \texttt{elias@cs.ox.ac.uk}} \and Annam{\'a}ria
   Kov{\'a}cs\thanks{Department of Informatics, Goethe University,
     Frankfurt M., Germany. Email: \texttt{panni@cs.uni-frankfurt.de}}} 

\date{}

\begin{document}

\maketitle

\begin{abstract}
  Noam Nisan and Amir Ronen conjectured that the best approximation ratio of deterministic truthful mechanisms for makespan-minimization for $n$ unrelated machines is $n$. This work validates the conjecture.
  
\end{abstract}

\input{intro}

\input{related}

\input{preliminaries}

\input{main-argument}

\input{nice-multistar}

\input{box}

\bibliographystyle{plain} 
\bibliography{biblio}

\end{document}

%% file: intro.tex
\section{Introduction}

This work resolves the oldest open problem in {\em algorithmic mechanism design}, the Nisan-Ronen conjecture. The seminal work of Nisan and Ronen~\cite{NR01} set the foundations of the field of {\em algorithmic mechanism design} by probing the computational and information-theoretic aspects of {\em mechanism design}. Mechanism design, a celebrated branch of game theory and microeconomics, studies the design of algorithms (called mechanisms) in environments where the input is privately held and provided by selfish participants. A mechanism for an optimization problem, on top of the traditional algorithmic goal (that assumes knowledge of the input), bears the extra burden of providing incentives to the participants to report their true input.

One of the main thrusts in this research area is to demarcate the limitations imposed by truthfulness on algorithms. To what extent are mechanisms less powerful than traditional algorithms? This question was crystallized by the widely influential Nisan-Ronen conjecture for the scheduling problem, which asserts that no deterministic mechanism for $n$ machines can have approximation ratio better than $n$.

The objective of the scheduling problem is to minimize the makespan of allocating $m$ tasks to $n$ unrelated machines, where each machine $i$ needs $t_{ij}$ units of time to process task $j$. The problem combines various interesting properties. First, it belongs to the most challenging and unexplored area of {\em multi-dimensional} mechanism design, as the private information (the $t_{ij}$'s) is multi-dimensional. In contrast, {\em related} machines scheduling belongs to {\em single-dimensional} mechanism design, which is well-understood, and for which the power of truthful mechanisms does not substantially differ from the best non-truthful algorithms: not only truthful mechanisms can compute exact optimal solutions (if one disregards computational issues~\cite{AT01}), but a truthful PTAS exists~\cite{CK13}. Second, the objective of the scheduling problem has a min-max objective, which from the mechanism design point of view is much more challenging, as opposed to the min-sum objective achievable by the famous VCG~\cite{Vic61,Cla71,Gro73} mechanism. VCG is truthful and can be applied to the scheduling problem, but it achieves a very poor approximation ratio, equal to the number of machines $n$~\cite{NR01}.

Is there a better mechanism for scheduling than VCG? Nisan and Ronen~\cite{NR01} conjectured that the answer should be negative, but for the past two decades, the question has remained open. In this work we validate the Nisan-Ronen conjecture.

\begin{theorem} \label{thm:nrtheorem}
  There is no deterministic truthful mechanism with approximation ratio better than $n$ for the problem of scheduling $n$ unrelated machines.
\end{theorem}

This result resolves the Nisan-Ronen conjecture. Over the years, various research attempts with limited success have been made to improve the original lower bound of $2$ by Nisan and Ronen. For example, the bound was improved to $2.41$ in~\cite{ChrKouVid09}, and later to $2.61$ in~\cite{KV07}, which held as the best bound for over a decade. More recently, the lower bound was improved to $2.75$ by Giannakopoulos, Hamerl, and Po\c{c}as~\cite{giannakopoulos2020}, and then to $3$ by Dobzinski and Shaulker~\cite{DS20}. These improved bounds represented progress in the field, but they left a huge gap between the lower and upper bounds. The first non-constant lower bound for the truthful scheduling problem was given in~\cite{CKK21b}, which showed a lower bound of $\Omega(\sqrt{n})$.

%% file: related.tex
\subsection{Further Related work}

A line of research studies randomized mechanisms with slightly
improved approximation guarantees. Note that there are two notions of
truthfulness for randomized mechanisms; a mechanism is {\em
  universally truthful} if it is defined as a probability distribution
over deterministic truthful mechanisms, and it is {\em
  truthful-in-expectation}, if in expectation no player can benefit by
lying. In \cite{NR01}, a universally truthful mechanism was proposed
for the case of two machines, and was later extended to the case of
$n$ machines by Mu'alem and Schapira~\cite{MualemS18} with an
approximation guarantee of $0.875n$, which was later improved to
$0.837n$ by \cite{LuYu08}. Lu and Yu~\cite{LuY08a} showed a
truthful-in-expectation mechanism with an approximation guarantee of
$(n+5)/2$.
Mu'alem and Schapira~\cite{MualemS18}, showed a lower bound of $2-1/n$, for both types of randomized truthful mechanisms. In~\cite{CKK10}, the lower bound was extended to fractional mechanisms, where each task can be fractionally allocated to multiple machines. They also showed a fractional mechanism with a guarantee of $(n+1)/2$.
Even for the special case of two machines the upper and lower bounds
are not tight~\cite{LuY08a,ChenDZ15}.

In the Bayesian setting, Daskalakis and Weinberg \cite{DaskalakisW15} showed a mechanism that is at most a factor of 2 from the {\em optimal truthful mechanism}, but not with respect to optimal makespan. The work of \cite{ChawlaHMS13} provided bounds of prior-independent mechanisms (where the input comes from a probability distribution unknown to the mechanism). Giannakopoulos and Kyropoulou~\cite{GiannakopoulosK17} showed that the VCG mechanism achieves an approximation ratio of $O( \log n/\log \log n )$ under some distributional and symmetry assumptions. Finally, in a recent work, \cite{CKK21} showed positive results for {\em graph} settings where each task can only be allocated to at most two machines. The construction of the lower bound of this work uses multi-cliques and multi-stars.

There is also work on variants of the scheduling problem and also of
special cases, for which significant results have been
obtained. Ashlagi, Dobzinski, and Lavi~\cite{ADL09}, resolved a restricted version
of the Nisan-Ronen conjecture, for the special but natural class of
{\em anonymous} mechanisms. Lavi and Swamy~\cite{LaviS09} studied a
restricted input domain which however retains the multi-dimensional
flavour of the setting. They considered inputs with only two possible
values ``low'' and ``high'' and they showed an elegant deterministic
mechanism with an approximation factor of 2. They also showed that
even for this setting achieving the optimal makespan is not possible
under truthfulness, and provided a lower bound of
$11/10$. Yu~\cite{Yu09} extended the results for a range of values,
and~\cite{Auletta0P15} studied multi-dimensional
domains where the private information of the machines is a single bit.

%% file: preliminaries.tex
\section{Preliminaries}
\label{sec:preliminaries}

In the unrelated machines scheduling problem, a set $M$ of $m$ tasks
needs to be scheduled on a set $N$ of $n$ machines. The processing
time or value that each machine $i$ needs to process task $j$ is
$t_{ij}$, and the completion time of machine $i$ for a subset $S$ of
tasks is equal to the sum of the individual task values
$t_i(S)=\sum_{j\in S}t_{ij}$. The objective is to find an allocation
of tasks to machines that minimize the \emph{makespan}, that is, the
maximum completion time.
 
\subsection{Mechanism design setting}

We assume that each machine $i\in N$ is controlled by a selfish agent (player) that is unwilling to process tasks and the values $t_{ij}$ is private information known only to them (also called the {\em type} of agent $i$).  The set $\mathcal{T}_i$ of possible types of agent $i$ consists of all vectors $b_i=(b_{i1},\ldots, b_{im})\in \mathbb{R}_+^{m}.$ Let also $\mathcal{T} = \times_{i\in N}\mathcal{T}_i$ denote the space of type profiles.

A mechanism defines for each player $i$ a set $\mathcal{B}_i$ of available strategies the player can choose from. We consider \emph{direct revelation} mechanisms, i.e., $\mathcal{B}_i=\mathcal{T}_i$ for all $i,$ meaning that the players strategies are to simply report their types to the mechanism. Each player $i$ provides a \emph{bid} $b_i\in \mathcal{T}_i$, which not necessarily matches the true type $t_i$, if this serves their interests. A mechanism $(A,\mathcal P)$ consists of two parts:
\paragraph{An allocation algorithm:} The allocation algorithm $A$ allocates the tasks to machines based on the players' inputs (bid vector) $b=(b_1,\ldots ,b_n)$. Let $\mathcal{A}$ be the set of all possible partitions of $m$ tasks to $n$ machines. The allocation function $A:\mathcal{T}\rightarrow \mathcal{A}$ partitions the tasks into the $n$ machines; we denote by $A_i(b)$ the subset of tasks assigned to machine $i$ for bid vector $b=(b_1,\ldots,b_n)$.
\paragraph{A payment scheme:} The payment scheme $\mathcal P=(\mathcal P_1,\ldots,\mathcal P_n)$ determines the payments, which also depends on the bid vector $b$. The functions $\mathcal P_1,\ldots,\mathcal P_n$ stand for the payments that the mechanism hands to each agent, i.e., $\mathcal P_i:\mathcal{T}\rightarrow \R$.

\medskip

The {\em utility} $u_i$ of a player $i$ is the payment that they get minus the {\em actual} time that they need to process the set of tasks assigned to them, $u_i(b)=\mathcal P_i(b)-t_i(A_i(b))$, where $t_i(A_i(b))=\sum_{j\in A_i(b)} t_{i,j}$. We are interested in \emph{truthful} mechanisms. A mechanism is truthful, if for every player, reporting his true type is a \emph{dominant strategy}. Formally,
$$u_i(t_i,b_{-i})\geq u_i(t'_i,b_{-i}),\qquad \forall i\in [n],\;\;
t_i,t'_i\in \mathcal{T}_i, \;\; b_{-i}\in \mathcal{T}_{-i},$$ where notation $\mathcal{T}_{-i}$ denotes all parts of $\mathcal{T}$ except its $i$-th part.

The quality of a mechanism for a given type $t$ is measured by the makespan $\mech(t)$ achieved by its allocation algorithm $A$, $ \mech(t) = \max_{i\in N}t_{i}(A_i(t))$, which is compared to the optimal makespan $\opt(t)=\min_{A\in \mathcal{A}}\max_{i\in N}t_i(A_i)$.

It is well known that only a subset of algorithms can be allocation algorithms of truthful mechanisms. In particular, no mechanism's allocation algorithm is optimal for every $t$, so it is natural to focus on the approximation ratio of the mechanism's allocation algorithm. A mechanism is \emph{$c$-approximate}, if its allocation algorithm is $c$-approximate, that is, if $c\geq\frac{\mech(t)}{\opt(t)}\;$ for all possible inputs $t$.

In this work, we do not place any requirements on the time to compute the mechanism's allocation $A$ and payments $\mathcal P.$ In other words, the lower bound is information-theoretic and does not make use of any computational assumptions.

\subsection{Weak monotonicity}

A mechanism consists of two algorithms, the allocation algorithm and the payment algorithm. However, we are only interested in the performance of the allocation algorithm. Therefore it is natural to ask whether it is possible to characterize the class of allocation algorithms that are part of a truthful mechanism with no reference to the payment algorithm. Indeed the following definition provides such a characterization.

\begin{definition} \label{def:wmon} An allocation algorithm $A$ is called {\em weakly monotone (WMON)} if it satisfies the following property: for every $i\in N$ and two inputs $t=(t_i,t_{-i})$ and $t'=(t'_i,t_{-i})$, the associated allocations $A$ and $A'$ satisfy $$t_i(A_i)-t_i(A'_i)\leq t'_i(A_i)-t'_i(A'_i).$$
  An equivalent condition, using indicator variables $a_{ij}(t)\in\{0,1\}$ of whether task $j$ is allocated to player $i$, is
  \begin{align*}
    \sum_{j\in M} (a_{ij}(t')-a_{ij}(t))(t_{ij}'-t_{ij})\leq 0.
  \end{align*}
\end{definition}

It is well known that the allocation function of every truthful mechanism is weakly monotone~\cite{BCR+06}. Although it is not needed in establishing a lower bound on the approximation ratio, it turns out that this is also a sufficient condition for truthfulness in convex domains~\cite{SY05} (which is the case for the scheduling domain).

A useful tool in our proof relies on the following immediate consequence of weak monotonicity (see \cite{CKK20} for a simple proof). Intuitively, it states that when you fix the values of all players for a subset of tasks, then the {\em restriction} of the allocation to the rest of the tasks is still weakly monotone.

\begin{lemma}\label{lem:restriction} Let $A$ be a weakly monotone allocation, and let
  $(S,T)$ a partition of $M$. When we fix the values of the tasks of $T$, the restriction of the allocation $A$ on $S$ is also weakly monotone. %
\end{lemma}

The following implication of weak monotonicity, that was first used in \cite{NR01}, is a standard tool for showing lower bounds for truthful mechanisms (see for example \cite{NR01,ChrKouVid09,MualemS18,ADL09,giannakopoulos2020,DS20}).

\begin{lemma}\label{lemma:tool}
  Consider a truthful mechanism $(A,\mathcal P)$ and its allocation for a bid vector $t$. Let $S$ be a subset of the tasks allocated to player $i$ and let $S'$ be a subset of the tasks allocated to the other players. Consider any bid profile $t'=(t'_i,t_{-i})$ that is obtained from $t$ by decreasing the values of player $i$ in $S$, i.e., $t'_{ij}< t_{ij}, j\in S$, and increasing the values of player $i$ in $S'$, i.e., $t'_{ij}> t_{ij}, j\in S'$. Then the allocation of player $i$ for $t$ and $t'$ agree for all tasks in $S\cup S'$.
\end{lemma}

Notice that this lemma guarantees that only the set $S$ of tasks allocated to player $i$ remains the same. This does not preclude changing the allocation of the other players for the tasks in $S'$, unless there are only two players. This is a major obstacle in multi-player settings, that we avoid in this work by focusing on {\em graph settings} with only two players for each task.

%% file: main-argument.tex
\section{The main argument}
\label{sec:main-argument}

We consider multi-graph instances, where edges correspond to tasks and nodes to machines~\cite{CKK21}.
For an edge $e$, we use the notation $e=\{i,j\}$ to denote its vertices $i$ and $j$, although they do not determine $e$ uniquely. An edge $e=\{i,j\}$ corresponds to a task that has extremely high values for nodes other than $i$ and $j$, which guarantees that any algorithm with approximation ratio at most $n$ must allocate it to either machine $i$ or machine $j$ (see Figure~\ref{fig:m-graphs} for an illustration).

The argument deals with multi-cliques with very high \emph{multiplicity}\footnote{The multiplicity of a multi-graph is defined to be the minimum multiplicity among its edges.}, in which \emph{every edge has an endpoint with value 0} (see Figure~\ref{fig:thm-nice-multi-star} for an example). The goal is to carefully select a subgraph of this multi-clique and change the values of some of its edges to obtain a lower bound on the approximation ratio. The fact that one of the two values of every edge is 0 is very convenient: a lower bound on the approximation ratio of the subgraph is a lower bound on the approximation ratio of the whole multi-clique as well, since the other edges do not affect the cost of the optimal allocation.

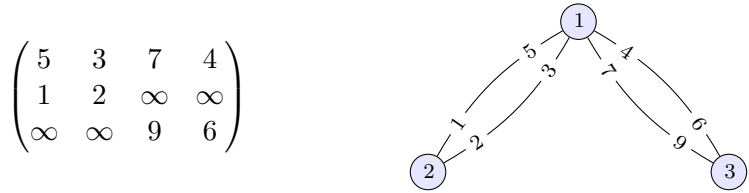
\begin{figure}[h]
  \centering
  \begin{tikzpicture}[scale=2, allow upside down,
      vertex/.style = {circle, draw, text centered, text width=0.4em, text
        height=0.5em, fill=blue!10}]

 \draw (-2,1.5) node {$\left(\begin{matrix}
  5 & 3 & 7 & 4 \\
  1 & 2 & \infty & \infty \\
  \infty & \infty & 9 & 6 \\
\end{matrix}\right)$};

\scriptsize

      \foreach \k/\x/\y in {1/1/2, 3/2/1, 2/0/1}
        \node[vertex] (\k) at (\x,\y) {$\small \k$};

\draw (2)  to [bend left=15] node[sloped,fill=white,pos=0.2] {$1$} node[sloped,fill=white,pos=0.8] {$5$}  (1);
\draw (1)  to [bend left=15] node[sloped,fill=white,pos=0.2] {$4$} node[sloped,fill=white,pos=0.8] {$6$}  (3);
\draw (2)  to [bend right=15] node[sloped,fill=white,pos=0.2] {$2$} node[sloped,fill=white,pos=0.8] {$3$}  (1);
\draw (1)  to [bend right=15] node[sloped,fill=white,pos=0.2] {$7$} node[sloped,fill=white,pos=0.8] {$9$}  (3);

    \end{tikzpicture}
    \caption{A multi-star instance with three players and four tasks in matrix form (left) and in graph form (right). The symbol $\infty$ denotes values that are extremely high compared to the other values. This instance is a multi-star, in which player 1 is the root and players 2 and 3 are the leaves.}
\label{fig:m-graphs}
\end{figure}

In this and the next section, which deal with multi-graphs, for an instance $v$, we use the notation $v_i^e$ to denote the value of node $i$ for an edge $e=\{i,j\}$. Section~\ref{sec:box} deals only with multi-stars and we use a different, more convenient, notation. In most of the argument, we fix the values of the multi-clique and we focus on the boundary functions.

\begin{definition}[Boundary function] \label{dfn:boundary function} Fix a mechanism and consider a multi-clique with values $v$. For an edge $e=\{i,j\}$, the boundary function or critical value function $\psi_{i,j}^e(z)$ is the threshold value for the allocation of $e$ to node $i$. More precisely, if we keep all the other values fixed and change the value of $e$ for node $j$ to $z$, then $e$ is allocated to $i$ if $v_i^e<\psi_{i,j}^e(z)$ and to $j$ if $v_i^e>\psi_{i,j}^e(z)$.
\end{definition}

Note that a boundary function $\psi_{i,j}^e(\cdot)$ may depend on the other values of $v$. A more precise notation would be $\psi_{i,j}^e[v](z)$ or $\psi_{i,j}^e[v_{-e}](z)$, but in this and in the next section we are dealing with a fixed $v$, so we drop it from the notation.

Boundary functions are real functions that are monotone but not necessarily continuous. An annoying, albeit minor, issue is that boundary functions are difficult to handle at points of discontinuity, so we want to avoid such points. The high level idea for achieving this is obvious: since the functions are monotone, they have only countably many points of discontinuity by Froda's theorem; by taking random instances in a continuous range, we have no points of discontinuity almost surely. This is formalized in Lemma~\ref{lemma:continuity}.
The following definition makes the requirement precise.
\begin{definition}[Continuity requirement]
\label{def:continuity}
  Fix a mechanism and consider an instance $v$. We say that $v$ has a discontinuity if there exists some edge $e=\{i,j\}$ so that $v_j^e\neq 0$ and $\psi_{i,j}^e(\cdot)$ is discontinuous at $v_j^e$. We say that an instance satisfies the continuity requirement if \emph{no rational translation of $v$ has a discontinuity}.
\end{definition}
The instances of the argument will be selected randomly so that they satisfy the continuity requirement almost surely.

The aim of the proof of Theorem~\ref{thm:nrtheorem} is to show --- by the probabilistic method --- that there exists a multi-clique of sufficiently high multiplicity that contains \emph{a star with approximation ratio at least $n$}, when we keep the values of all other edges fixed. Actually the argument aims to show that the bound on the approximation ratio for the star is arbitrarily close to $n-1$. The extra +1 in the approximation ratio comes, almost for free, by adding a loop to the root of the star, or equivalently an additional edge between the root and another node $j$ with very high value for $j$.

To show that there exists a star $S$ with approximation ratio $n-1$, we \emph{roughly} aim to show that there exists a star with some root $i$, with the following properties: 
\begin{enumerate}
\item every edge $e=(i,j)$ of $S$ has value 0 for $i$ and {\em the same} value $z$ for the leaves, for some $z>0$.
\item the sum of the values of the boundary functions over all edges $\sum_{e\in S} \psi_{i,j}^e(z)$ is at least $(n-1)z$.
\item the mechanism allocates all edges to the root, when we change its values to $\psi_{i,j}^e(z)$ for all $j\neq i$.
\end{enumerate}

It follows immediately that such a star has approximation ratio $n-1$: the mechanism allocates all tasks to the root with makespan $\sum_{e\in S} \psi_{i,j}^e(z)\geq (n-1)z$, while a better allocation is to allocate all tasks to the leaves with makespan $z$.

A star that satisfies the second property will be called \emph{nice} and the third property \emph{box}. Actually, for technical reasons we need to work with approximate notions of niceness and box-ness.

\begin{definition}[Nice star and multi-star] \label{def:nice star} For a given $\varepsilon>0$ and an instance $v$,
  a star $S$ with root $i$ and leaves all the remaining $(n-1)$ nodes is called $\varepsilon$-nice, or simply nice, if there exists $z>0$ such that, first, every edge $e=\{i,j\}$ of $S$ has value $v_i^e=0$ for root $i$ and $v_j^e\in (z, (1+\varepsilon)z)$ for leaf $j$, and second,
\begin{equation} \label{eq:sum-of-psi}
  \sum_{e\in S} \psi_{i,j}^e(v_j^e)\geq (1-3\varepsilon)(n-1)z.
\end{equation}  
  A multi-star is called nice if all of its stars with $n-1$ leaves are nice, with the same $z$.
\end{definition}

By letting $\varepsilon$ tend to 0, $1-3\varepsilon$ can be arbitrarily close to 1 and $v_i^e$ can be arbitrarily close to $z$. The factor 3 in the expression $1-3\varepsilon$ is for convenience when such a property is established in Section~\ref{sec:proof-nice-multi-star}.

\begin{definition}[Box]
\label{def:box}
For a given $\delta>0$ and instance $v$, a star $S$ with root $i$ is called $\delta$-box, or simply box, if every edge $e$ of $S$ has value 0 for $i$ and the mechanism allocates all edges to the root, when we change the root values to $\psi_{i,j}^e(v_j^e)-\delta$ for every leaf $j$ of $S$.
\end{definition}

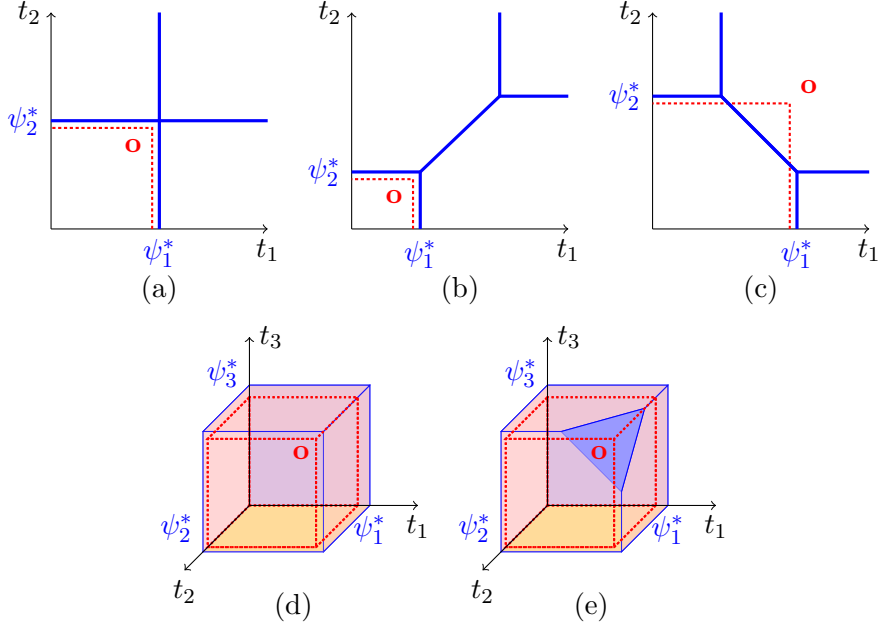
\begin{figure}[h]
\newcommand{\Depth}{2}
\newcommand{\Height}{2}
\newcommand{\Width}{2}
\newcommand{\Fraction}{0.5}
\newcommand{\deltaX}{0.2}
\newcommand{\deltaY}{0.2}
\newcommand{\psiStar}{\psi^*}

  \centering

  \begin{tikzpicture}[scale=0.48]
    \draw[->] (0,0) -- (6,0) node[anchor=north] {$t_{1}$};
    \draw[->] (0,0) -- (0,6) node[anchor=east] {$t_{2}$};
    \draw[very thick, blue]  (3,0) -- (3,6) ; \draw[very thick, blue] (0,3) node[anchor=east] {$\psiStar_{2}$} -- (6,3) ;

    \draw[very thick, red, densely dotted, thick] (3-\deltaY,0)  -- (3-\deltaY,3-\deltaY) node[anchor=north east] {\bf o} -- (0,3-\deltaY) ;

    \draw (1.5,5) node[anchor=center] {}; %

\draw[blue] (3,0) node[anchor=north] {$\psiStar_1$};
    \draw (3,-1) node[anchor=north] {(a)};
  \end{tikzpicture}
  \begin{tikzpicture}[scale=0.48]

    \draw[->] (0,0) -- (6,0) node[anchor=north] {$t_{1}$};
    \draw[->] (0,0) -- (0,6) node[anchor=east] {$t_{2}$};
    \draw[very thick, blue] (4.1,6) -- (4.1, 3.68) -- (1.9,1.58) -- (0,1.58) node[anchor=east] {$\psiStar_{2}$};

    \draw[very thick, red,densely dotted, thick] (1.9-\deltaY,0)  -- (1.9-\deltaY,1.58-\deltaY) node[anchor=north east] {\bf o} -- (0,1.58-\deltaY) ;

    \draw[very thick, blue] (1.9,0) node[anchor=north] {$\psiStar_1$}  -- (1.9, 1.58); \draw[very thick, blue] (4.1,3.68) -- (6, 3.68) ;

    \draw (1.5,5) node[anchor=center] {}; %

    \draw (3,-1) node[anchor=north] {(b)};
  \end{tikzpicture}
  \begin{tikzpicture}[scale=0.48]

    \draw[->] (0,0) -- (6,0) node[anchor=north] {$t_{1}$};
    \draw[->] (0,0) -- (0,6) node[anchor=east] {$t_{2}$};
    \draw[very thick, blue] (1.9,6) -- (1.9,3.68) -- (4,1.58) -- (6,1.58) ; \draw[very thick, blue] (0, 3.68) node[anchor=east] {$\psiStar_{2}$} -- (1.9,3.68) ;

    \draw[very thick, red,densely dotted, thick] (4-\deltaY,0)  -- (4-\deltaY,3.68-\deltaY) node[anchor=south west] {\bf o} -- (0,3.68-\deltaY) ;

    \draw[very thick, blue] (4,0) node[anchor=north] {$\psiStar_1$} -- (4, 1.58) -- (1.9, 3.68);

    \draw (1.5,5) node[anchor=center] {}; %

    \draw (3,-1) node[anchor=north] {(c)};
  \end{tikzpicture}%

\begin{tikzpicture}[scale=0.8]
\coordinate (O) at (0,0,0);
\coordinate (A) at (0,\Width,0);
\coordinate (B) at (0,\Width,\Height);
\coordinate (C) at (0,0,\Height);
\coordinate (D) at (\Depth,0,0);
\coordinate (E) at (\Depth,\Width,0);
\coordinate (F) at (\Depth,\Width,\Height);
\coordinate (G) at (\Depth,0,\Height);

\coordinate (O1) at (0,0,0);
\coordinate (A1) at (0,\Width-\deltaX,0);
\coordinate (B1) at (0,\Width-\deltaX,\Height-\deltaX);
\coordinate (C1) at (0,0,\Height-\deltaX);
\coordinate (D1) at (\Depth-\deltaX,0,0);
\coordinate (E1) at (\Depth-\deltaX,\Width-\deltaX,0);
\coordinate (F1) at (\Depth-\deltaX,\Width-\deltaX,\Height-\deltaX);
\coordinate (G1) at (\Depth-\deltaX,0,\Height-\deltaX);

\coordinate (BF) at (\Fraction*\Depth,\Width,\Height);
\coordinate (EF) at (\Depth,\Width,\Fraction*\Height);
\coordinate (GF) at (\Depth,\Fraction*\Width,\Height);

\draw[blue,fill=yellow!80] (O) -- (C) -- (G) -- (D) -- cycle;%
\draw[blue,fill=blue!30] (O) -- (A) -- (E) -- (D) -- cycle;%
\draw[blue,fill=red!10] (O) -- (A) -- (B) -- (C) -- cycle;%
\draw[blue,fill=red!20,opacity=0.8] (D) -- (E) -- (F) -- (G) -- cycle;%
 \draw[blue,fill=red!20,opacity=0.6] (C) -- (B) -- (F) -- (G) -- cycle;%
 \draw[blue,fill=red!20,opacity=0.8] (A) -- (B) -- (F) -- (E) -- cycle;%

\draw[red,densely dotted, thick] (O1) -- (C1) -- (G1) -- (D1) -- cycle;%
\draw[red,densely dotted, thick] (O1) -- (A1) -- (E1) -- (D1) -- cycle;%
\draw[red,densely dotted, thick] (O1) -- (A1) -- (B1) -- (C1) -- cycle;%
\draw[red,densely dotted, thick] (D1) -- (E1) -- (F1) -- (G1) -- cycle;%
 \draw[red,densely dotted, thick] (C1) -- (B1) -- (F1) -- (G1) -- cycle;%
 \draw[red,densely dotted, thick] (A1) -- (B1) -- (F1) -- (E1) -- cycle;%

\coordinate (AA) at (0,1.4*\Width,0);
\coordinate (CC) at (0,0,1.4*\Height);
\coordinate (DD) at (1.4*\Depth,0,0);

\draw[->] (O) -- (AA);
\draw[->] (O) -- (CC);
\draw[->] (O) -- (DD);

\draw (2.8,0,0) node[anchor=north] {$t_1$};
\draw (0,2.8,0) node[anchor=west] {$t_3$};
\draw (0,0,2.8) node[anchor=north] {$t_2$};

\draw[blue] (2,0,0) node[anchor=north] {$\psiStar_1$};
\draw[blue] (0,2.2,0) node[anchor=east] {$\psiStar_3$};
\draw[blue] (0,0,2) node[anchor=south east] {$\psiStar_2$};

    \draw[red] (1.9,1.9,1.9) node[anchor=north east] {\bf o};

    \draw (1.5,-0.5, 2) node[anchor=north] {(d)};
\end{tikzpicture}
\begin{tikzpicture}[scale=0.8]
\coordinate (O) at (0,0,0);
\coordinate (A) at (0,\Width,0);
\coordinate (B) at (0,\Width,\Height);
\coordinate (C) at (0,0,\Height);
\coordinate (D) at (\Depth,0,0);
\coordinate (E) at (\Depth,\Width,0);
\coordinate (F) at (\Depth,\Width,\Height);
\coordinate (G) at (\Depth,0,\Height);

\coordinate (O1) at (0,0,0);
\coordinate (A1) at (0,\Width-\deltaX,0);
\coordinate (B1) at (0,\Width-\deltaX,\Height-\deltaX);
\coordinate (C1) at (0,0,\Height-\deltaX);
\coordinate (D1) at (\Depth-\deltaX,0,0);
\coordinate (E1) at (\Depth-\deltaX,\Width-\deltaX,0);
\coordinate (F1) at (\Depth-\deltaX,\Width-\deltaX,\Height-\deltaX);
\coordinate (G1) at (\Depth-\deltaX,0,\Height-\deltaX);

\coordinate (BF) at (\Fraction*\Depth,\Width,\Height);
\coordinate (EF) at (\Depth,\Width,\Fraction*\Height);
\coordinate (GF) at (\Depth,\Fraction*\Width,\Height);

\draw[blue,fill=yellow!80] (O) -- (C) -- (G) -- (D) -- cycle;%
\draw[blue,fill=blue!30] (O) -- (A) -- (E) -- (D) -- cycle;%
\draw[blue,fill=red!10] (O) -- (A) -- (B) -- (C) -- cycle;%
\draw[blue,fill=red!20,opacity=0.8] (D) -- (E) -- (EF) -- (GF) -- (G) -- cycle;%
\draw[blue,fill=red!20,opacity=0.6] (C) -- (B) -- (BF) -- (GF) -- (G) -- cycle;%
\draw[blue,fill=red!20,opacity=0.8] (A) -- (B) -- (BF) -- (EF) -- (E) -- cycle;%
\draw[blue,fill=blue!40,opacity=1] (BF) -- (EF) -- (GF);

\draw[red,densely dotted, thick] (O1) -- (C1) -- (G1) -- (D1) -- cycle;%
\draw[red,densely dotted, thick] (O1) -- (A1) -- (E1) -- (D1) -- cycle;%
\draw[red,densely dotted, thick] (O1) -- (A1) -- (B1) -- (C1) -- cycle;%
\draw[red,densely dotted, thick] (D1) -- (E1) -- (F1) -- (G1) -- cycle;%
 \draw[red,densely dotted, thick] (C1) -- (B1) -- (F1) -- (G1) -- cycle;%
 \draw[red,densely dotted, thick] (A1) -- (B1) -- (F1) -- (E1) -- cycle;%

\coordinate (AA) at (0,1.4*\Width,0);
\coordinate (CC) at (0,0,1.4*\Height);
\coordinate (DD) at (1.4*\Depth,0,0);

\draw[->] (O) -- (AA);
\draw[->] (O) -- (CC);
\draw[->] (O) -- (DD);

\draw (2.8,0,0) node[anchor=north] {$t_1$};
\draw (0,2.8,0) node[anchor=west] {$t_3$};
\draw (0,0,2.8) node[anchor=north] {$t_2$};

\draw[blue] (2,0,0) node[anchor=north] {$\psiStar_1$};
\draw[blue] (0,2.2,0) node[anchor=east] {$\psiStar_3$};
\draw[blue] (0,0,2) node[anchor=south east] {$\psiStar_2$};

    \draw[red] (1.9,1.9,1.9 ) node[anchor=north east] {\bf o};

    \draw (1.5,-0.5, 2) node[anchor=north] {(e)};

\end{tikzpicture}

\caption{\small (Box). Allocation partitions of the root values for a star of 2 leaves (a)-(c) and 3 leaves (d)-(e); in the latter case, only part of the allocation partition is shown.
  This figure is about a star with root $i$ and leaves $j\in \{1,2,3\}$. If we denote the edges of the star by $e_j$, the figure uses the shorthand: $t_j=v_i^{e_j}$ for the values of the root, and $\psi_j^*=\psi_{i,j}^{e_j}(v_j^{e_j})$ for the boundary values.
  Note that monotonicity restricts the shapes and boundaries of the allocation areas, as discussed in detail in Sections~\ref{sec:region-r_p} and~\ref{sec:facts-about-two}.
  The dotted red lines correspond to values $\psiStar_j-\delta$ of the box definition. Cases (a), (b), and (d) are boxes, as the corner $o$ is inside the region where the root gets all the tasks. On the other hand, cases (c) and (e) are not boxes, since the corner point $o$ lies outside this region.}
  \label{fig:box}
\end{figure}

See Figure~\ref{fig:box} for an illustration. Note that while nice stars have $n-1$ leaves, box stars may have fewer leaves. This will be useful to facilitate induction in Section~\ref{sec:box}, where this property is established.

Now that we have the definitions of nice multi-stars and box stars, we can state the two main propositions that almost immediately establish the main result.

First, the following theorem, proved in Section~\ref{sec:proof-nice-multi-star}, says that there exist nice multi-stars of arbitrarily high multiplicity (see Figure~\ref{fig:thm-nice-multi-star}).

\begin{theorem}[Nice Multi-Star] \label{thm:nice multi-star}
  For every mechanism with bounded approximation ratio and every $q$, there exists a multi-clique that satisfies the continuity requirement and contains a nice multi-star with multiplicity $q$.
\end{theorem}

\begin{figure}[h]
\centering
\begin{tikzpicture}[scale=2, allow upside down,
      vertex/.style = {circle, draw, text centered, text width=0.4em, text
        height=0.5em, fill=blue!10}]

      \scriptsize

      \foreach \k/\x/\y in {0/1/2, 3/2/1, 2/1/0, 1/0/1}
        \node[vertex] (\k) at (\x,\y) {$\small \k$};

      \foreach \from/\to in {0/1,0/2,0/3,1/0,1/2,1/3,2/0,2/1,2/3,3/0,3/1,3/2}
        \draw (\from)  to [bend left=15] node[sloped,fill=white,pos=0.2] {$z$} node[sloped,fill=white,pos=0.8] {$0$}  (\to);
\draw (1,-0.5) node {(a)};
  \end{tikzpicture}
\hspace{10mm}
  \begin{tikzpicture}[scale=2, allow upside down,
      vertex/.style = {circle, draw, text centered, text width=0.4em, text
        height=0.5em, fill=blue!10}]

      \scriptsize

      \foreach \k/\x/\y in {0/1/2, 3/2/1, 2/1/0, 1/0/1}
        \node[vertex] (\k) at (\x,\y) {$\small \k$};

      \foreach \from/\to in {0/1,0/2,0/3}
      \draw (\from)  to [bend left=15] node[sloped,fill=white,pos=0.2] {$0$} node[sloped,fill=white,pos=0.8] {$z$}  (\to);
      \foreach \from/\to in {0/1,0/2,0/3}
      \draw (\from)  to [bend right=15] node[sloped,fill=white,pos=0.2] {$0$} node[sloped,fill=white,pos=0.8] {$z$}  (\to);
    \draw (1,-0.5) node {(b)};
    \end{tikzpicture}

    \caption{This is an illustration of the components that appear in the statement of Theorem~9. In (a), a multi-clique with four nodes and multiplicity 2 is depicted. For simplification, we use $z$ to denote non-zero values, which are not necessarily the same for all edges. It should be noted that the actual multiplicity needed is much higher. In (b), a multi-star which is a subgraph of the multi-clique is depicted. Again the actual multiplicity needed is very high. If this is a {\em nice} multi-star, the value $z$ is approximately the same for all leaves. }
 \label{fig:thm-nice-multi-star}
\end{figure}
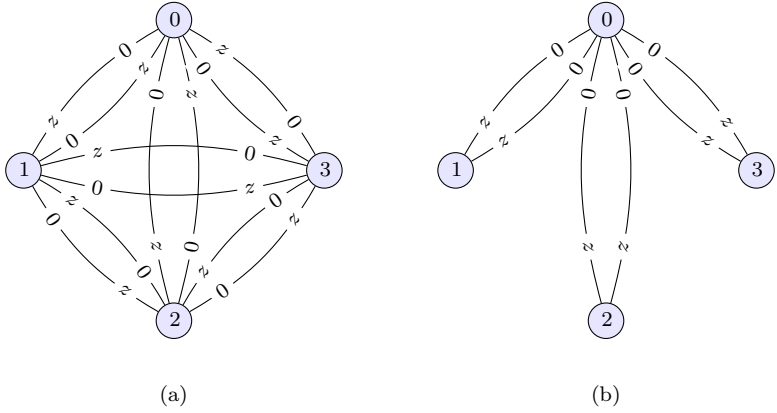

Second, the following theorem says that nice multi-stars with sufficiently high multiplicity --- guaranteed to exist by the above theorem --- contain a spanning box i.e., a box star of $n-1$ leaves (see Figure~\ref{fig:thm-box}).

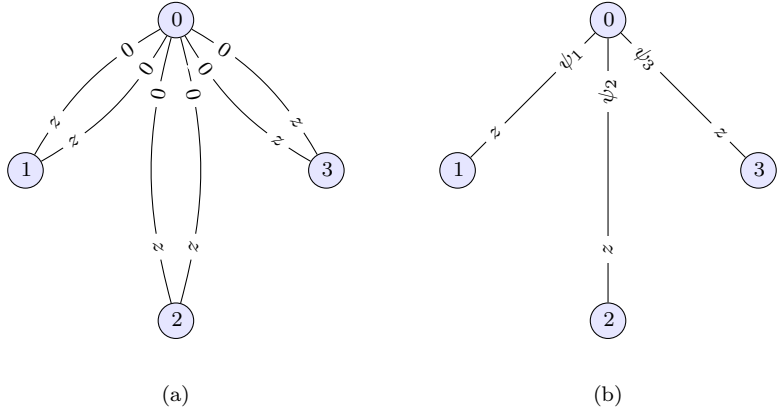
\begin{figure}[h]
\centering
  \begin{tikzpicture}[scale=2, allow upside down,
      vertex/.style = {circle, draw, text centered, text width=0.4em, text
        height=0.5em, fill=blue!10}]

      \scriptsize

      \foreach \k/\x/\y in {0/1/2, 3/2/1, 2/1/0, 1/0/1}
        \node[vertex] (\k) at (\x,\y) {$\small \k$};

      \foreach \from/\to in {0/1,0/2,0/3}
      \draw (\from)  to [bend left=15] node[sloped,fill=white,pos=0.2] {$0$} node[sloped,fill=white,pos=0.8] {$z$}  (\to);
      \foreach \from/\to in {0/1,0/2,0/3}
      \draw (\from)  to [bend right=15] node[sloped,fill=white,pos=0.2] {$0$} node[sloped,fill=white,pos=0.8] {$z$}  (\to);
    \draw (1,-0.5) node {(a)};
    \end{tikzpicture}
\hspace{10mm}
  \begin{tikzpicture}[scale=2, allow upside down,
      vertex/.style = {circle, draw, text centered, text width=0.4em, text
        height=0.5em, fill=blue!10}]

      \scriptsize

      \foreach \k/\x/\y in {0/1/2, 3/2/1, 2/1/0, 1/0/1}
        \node[vertex] (\k) at (\x,\y) {$\small \k$};

      \foreach \from/\to in {0/1,0/2}
      \draw (\to)  to  node[sloped,fill=white,pos=0.2] {$z$} node[sloped,fill=white,pos=0.8] {$\psi_{\to}$}  (\from);
    
      \draw (0)  to  node[sloped,fill=white,pos=0.2] {$\psi_{3}$} node[sloped,fill=white,pos=0.8] {$z$}  (3);

\draw (1,-0.5) node {(b)};
    \end{tikzpicture}
    \caption{The figure depicts a multi-star with high multiplicity
      ($a$) and the final spanning star $(b)$ selected by
      Theorem~\ref{thm:box}. If the multi-star satisfies the niceness
      property, then the final star also satisfies the niceness property (hence roughly
      $\psi_1+\psi_2+\psi_3\geq 3z$) and also the box-ness property
      (see Figure~\ref{fig:box}(d)) which guarantees that these edges
      must be allocated to the root. Note that this final star is a
      subgraph of the original multi-clique, the contribution of the
      remaining edges (which do not appear in the figure) to the
      optimal makespan is $0$.}
\label{fig:thm-box}    
\end{figure}

\begin{theorem}[Box] \label{thm:box} Fix $\delta, \varepsilon>0$ and a mechanism with approximation ratio at most $n$. Consider an instance that satisfies the continuity requirement and contains a multi-star, of sufficiently high multiplicity, in which all values of the root $i$ are 0 and all values of the leaves are in $(z, (1+\varepsilon)z)$. Then the multi-star contains a star with $n-1$ leaves, which is a $\delta$-box.
\end{theorem}

The proof of the Box Theorem is given in Section~\ref{sec:box}. To keep the notation as clean as possible, in the proof we assume without loss of generality that the root of the multi-star is node 0, the leaves are nodes 1 to $n-1$, and we adapt the notation accordingly. %

The proof of the main result (Theorem~\ref{thm:nrtheorem}) follows immediately from the above two theorems. We consider multi-cliques that contain loops with value 0 for every node and observe that the addition of loops does not affect the proofs of the above theorems. Use Theorem~\ref{thm:nice multi-star} to find a multi-clique that satisfies the continuity requirements and contains a nice multi-star with sufficiently high multiplicity. Use Theorem~\ref{thm:box} to find a nice box inside it.  

\begin{lemma}
  A $\delta$-nice box with a loop in the root, in which all values of the root are 0 and all values of the leaves are in $(z, (1+\varepsilon)z)$, has approximation ratio $n$, as $\delta$ and $\varepsilon$ tend to 0.
\end{lemma}
\begin{proof}
  Take a nice box and consider the instance when we change the values of the root $i$ to $\psi_{i,j}^e(v_j^e)-\delta$ for all $j\neq i$. By the box-ness property all the edges are allocated to the root. Change now the value of the loop to $z$ and decrease the values of the root to $\psi_{i,j}^e(v_j^e)-2\delta$. By using the fact that the task that corresponds to the loop must still be allocated to root, even when we increase its value to $z$, and by also applying monotonicity (Lemma~\ref{lemma:tool}), the allocation of the edges remains the same. The makespan of the mechanism is
  \begin{align*}
    z+\sum_{j\neq i} (\psi_{i,j}^e(v_j^e)-2\delta)\geq z+(1-3\varepsilon)(n-1)z-2(n-1)\delta,
  \end{align*}
  while the optimal makespan is at most $(1+\varepsilon)z$, (when the root gets the loop and the leaves get all the remaining edges). The ratio tends to $n$ as $\delta$ and $\varepsilon$ tend to 0.
\end{proof}

The lemma, together with the Nice-Multi-Star and the Box theorems, establishes the main result. Actually it establishes something stronger. The approximation ratio of truthful scheduling is $n$ even when the domain is the set of multi-graphs.

\begin{corollary}
  The approximation ratio of truthful scheduling of multi-graphs is $n$.
\end{corollary}

%% file: nice-multistar.tex
\section{Proof of the Nice Multi-Star theorem (Theorem~\ref{thm:nice multi-star})}
\label{sec:proof-nice-multi-star}

\newcommand{\AVMclique}{multi-clique of dipoles}
\newcommand{\AVdipole}{dipole}

The aim of this section is to prove Theorem~\ref{thm:nice multi-star}. To keep the argument clean, we first collect and prove some useful statements. These include %
an immediate application of Young's inequality on boundary functions of one-dimensional mechanisms, and two easy properties of mechanisms with two players and one or two tasks.

\subsection{Boundary functions and integrals}
\label{sec:bound-funct-integr}

For an edge $e=\{i,j\}$, we fix the values of every other edge and consider the boundary functions $\psi_{i,j}^e(\cdot)$ and $\psi_{j,i}^e(\cdot)$. By the definition of the boundary functions (Definition~\ref{dfn:boundary function}), if the values of an edge $e$ are $v_i^e$ and $v_j^e$, $e$ is allocated to node $i$ when $v_i^e<\psi_{i,j}^e(v_j^e)$ and to node $j$ when $v_j^e<\psi_{j,i}^e(v_i^e)$, where $v_i^e$ and $v_j^e$ are the values of nodes $i$ and $j$ on edge $e$, respectively. Note that the boundary functions do not determine the allocation when we have equality.

These functions are nondecreasing and, roughly speaking, inverse of
each other. More precisely, $\psi_{j,i}^e$ is a pseudo-inverse of
$\psi_{i,j}^e$, that is,
\begin{align}
  \sup\{x\geq 0 \colon \psi_{i,j}^e(x)<y\} \,\leq\, \psi_{j,i}^e(y) \,\leq\, \inf\{x\geq 0 \colon \psi_{i,j}^e(x)>y\}.
\end{align}
By convention the supremum on the left is 0 when $y=0$. 

The following proposition is a special case of Young's Inequality\footnote{The original Young's Inequality states that if $f$ is nonnegative, continuous, and strictly increasing, then $\int_0^a f(x)\, dx+\int_0^b f^{-1}(x)\, dx\geq ab$.}. The original proof required continuous and strictly increasing functions, but subsequent proofs relaxed these conditions to nondecreasing pseudo-inverse functions~\cite{Cunningham1971}.

\begin{proposition}[Young's inequality] \label{prop:young's ineq}
  If $\psi_{i,j}^e$ and $\psi_{j,i}^e$ are the boundary functions of an edge, then for every $a\geq 0$:
  \begin{align}
    \int_0^a \psi_{i,j}^e(x)\, dx+\int_0^a \psi_{j,i}^e(x) \, dx&\geq a^2.
  \end{align}
\end{proposition}

\subsection{Independence of boundary functions of parallel edges}
\label{sec:indep-bound-funct}

\begin{proposition} \label{prop:2x2-independence}
  For every mechanism with bounded approximation ratio and every multi-graph in which one value of every edge is 0, the boundary function $\psi_{i,j}^e(\cdot)$ of an edge $e=\{i,j\}$ is independent of the values of the other edges between $i$ and $j$, except perhaps at its points of discontinuity.
\end{proposition}
\begin{proof}
  Let $e'\neq e$ be an edge between $i$ and $j$. It suffices to show that $\psi_{i,j}^e(\cdot)$ is independent of the values of $e'$ when we fix the values of the remaining edges. It is known~\cite{DS08} and it also follows directly from Theorem~\ref{thm:sibling-independence}  that a $2\times 2$ mechanism has unbounded approximation ratio unless it is a (relaxed) task-independent mechanism. This proves the proposition, since the contribution of the remaining edges to the optimal makespan is 0.
\end{proof}

\begin{proposition} \label{prop:bounded slope}
  For every mechanism with approximation ratio less than $n$ and every multi-graph in which one value of every edge is 0, the boundary function $\psi_{i,j}^e(\cdot)$ of an edge $e=\{i,j\}$ satisfies $\psi_{i,j}^e(x)<n\, x$.
\end{proposition}
\begin{proof}
  If for some $x$, $\psi_{i,j}^e(x)\geq n\, x$, by setting the values of $e$ for nodes $j$ and $i$ to $x$ and $\psi_{i,j}^e(x)-\epsilon x$, respectively, for some $\epsilon>0$, the approximation ratio is at least $n-\epsilon$. The proposition follows by letting $\epsilon\rightarrow 0$.
\end{proof}

\subsection{The construction}
\label{sec:construction}

Fix some $\varepsilon$ so that $1/\varepsilon$ is a positive integer
and let $\varepsilon\ZP=\{k \varepsilon \colon k\in \ZP\}$ denote the
integer multiples of $\varepsilon$. We will also use the notation
$\fl{x}{\varepsilon}=\varepsilon\lfloor x/\varepsilon \rfloor$ and
$\ceil{x}{\varepsilon}=\varepsilon\lceil x/\varepsilon \rceil$.

Let $D_{\varepsilon}$ denote the subset $\varepsilon\ZP$ with values in $(0,1]$, i.e.,  $D_{\varepsilon}=\{\varepsilon,2\varepsilon,\ldots,1-\varepsilon,1\}$, and define $W_{\varepsilon}=\{(0, z)\colon z\in D_{\varepsilon}\} \cup \{(z, 0)\colon z\in D_{\varepsilon}\}$.

In this section, we will say that a multi-graph has \emph{multiplicity} $k$ when \emph{every} edge of the multi-graph has multiplicity $k$. The instances of the argument are selected as follows:
\begin{definition}[Random multi-clique]
   Take a multi-clique of sufficiently large multiplicity $q'$. For each edge $e$ of this multi-clique, select its values in two steps: in the first step, uniformly and independently select a value $(0,z)$ or $(z,0)$ from $W_{\varepsilon}$ and, in the second step, change $z$ to a random value uniformly distributed in $(z, (1+\varepsilon)z)$.
  We will refer to the result as a \emph{random multi-clique}.
\end{definition}
Once we select the values, \emph{we fix them once and for all in this section}. Since every edge of a random multi-clique has a value equal to 0, the optimal makespan for these multi-graphs is 0.

The only reason of having the second step in selecting the values of a random multi-clique is to satisfy the continuity requirements, which is established in the next lemma. This second step plays no other role in this section and it can be mostly ignored, in the sense that the argument would be more straightforward if the values of an edge $e$ were simply in $W_{\varepsilon}$.

\begin{lemma} \label{lemma:continuity}
  A random multi-clique satisfies the continuity requirement almost surely. 
\end{lemma}
\begin{proof}
  For a fixed instance $v$, Froda's theorem applied to the monotone boundary function $\psi_{i,j}^e(\cdot)$ guarantees that its set of discontinuity points is countable. Since the set of edges is finite, and the rational translations of $v$ are countable, the set of discontinuity points for $v$ and its rational translations is also countable. The lemma follows because, every value of a random multi-clique is selected uniformly in $(z, (1+\varepsilon)z)$, for some $z, \varepsilon>0$, which is an uncountable set.
\end{proof}

The \emph{objective in this section} is to show that for every truthful mechanism with bounded approximation ratio and for any given positive integer $q$, if the multiplicity $q'$ of the random multi-clique is sufficiently high, it contains a nice multi-star of multiplicity $q$. To do this, we work with \emph{a finite set of functions} with domain $D_{\varepsilon}$ that approximate the boundary functions with multiples of $\varepsilon$: 
\begin{equation}
  \label{eq:1}
  \apsi_{i,j}^e(z)=\fl{\psi_{i,j}^e(z)}{\varepsilon},  %
\end{equation}
for every $z\in D_{\varepsilon}$.

The important property is that the set of \emph{approximate boundary functions} is finite. Their domain $D_{\varepsilon}$ consists of $1/\varepsilon$ values and their range consists of $n/\varepsilon$ values $\{0,\varepsilon,2\varepsilon,\ldots,n-\varepsilon\}$, which follows from the fact that a mechanism with approximation ratio strictly less than $n$ must satisfy $\psi_{i,j}^e(z)<n z$ (Proposition~\ref{prop:bounded slope}). We will use this fact to show that a large random multi-clique contains many edges with identical approximate boundary functions. Actually, we will need something stronger, because the edges of nice multi-stars must also agree on their values. To achieve this, we consider sets of edges with identical approximate boundary functions that cover the whole spectrum of values. The precise requirements are captured by the following definition.

\begin{definition}
  A set $A$ of $2/\varepsilon$ edges between nodes $i$ and $j$ is called a full-range-dipole or simply \AVdipole\ if
\begin{itemize}
\item for every value $z$ in $D_{\varepsilon}$, there is a pair of edges $e,e'\in A$  with values such that $\fl{v_i^e}{\varepsilon}=z, v_j^e=0$, $v_i^{e'}=0,$ and $\fl{v_j^{e'}}{\varepsilon}=z$, and 
\item all edges $e\in A$ have the same $\apsi_{i,j}^e(\cdot)$, which we simply denote by $\apsi_{i,j}(\cdot)$. Similarly all edges have the same $\apsi_{j,i}^e(\cdot)$.
\end{itemize}
\end{definition}

Dipoles are useful because any clique of them --- that is, a multi-graph with a dipole between every pair of nodes --- contains a nice star as the following lemma shows. Actually the lemma shows the \emph{stronger statement} that bound~\eqref{eq:sum-of-psi} in the definition of nice stars (Definition~\ref{def:nice star}) holds when we replace $\psi_{i,j}^e(v_j^e)$ by its lower bound $\apsi_{i,j}^e(z)$. It is this stronger statement that will be useful to establish the main result of the section. 

\begin{lemma} \label{lemma:existence of i and z}
  In every clique of \AVdipole{}s there exist $i\in [n]$ and $z\in D_{\varepsilon}$ such that the following holds for the star $S$ that has root $i$, leaves the remaining $n-1$ nodes, and edges with value $0$ for $i$ and in $(z, (1+\varepsilon)z)$ for every $j\neq i$: 
  \begin{equation}
    \label{eq:sum-of-apsi}
    \sum_{j\neq i} \apsi_{i,j}(z) \geq (1-3\varepsilon)(n-1)z.
  \end{equation}
\end{lemma}
\begin{proof}
  First, for every edge $e$ between $i$ and $j$, since $\psi_{i,j}^e(x)$ and $\psi_{j,i}^e(x)$ are pseudo-inverse functions, Young's inequality (Proposition~\ref{prop:young's ineq}) gives
  \begin{align*}
    \int_0^1 \psi_{i,j}^e(x) +\psi_{j,i}^e(x)\, dx \geq 1.
  \end{align*}
  Second, by the monotonicity property of the function $\psi_{i,j}^e(\cdot)$ and the definition of the approximate boundary functions, we get
  \begin{align*}
    \int_0^1 \psi_{i,j}^e(x)\, dx \leq \int_0^1 \psi_{i,j}^e(\ceil{x}{\varepsilon})\, dx
    \leq \int_0^1 \apsi_{i,j}(\ceil{x}{\varepsilon})+\varepsilon\, dx
    = \varepsilon + \sum_{z\in D_{\varepsilon}} \varepsilon \apsi_{i,j}(z).
  \end{align*}
  Putting both facts together, we get that for every edge $e$ between $i$ and $j$:
  \begin{align*}
    \sum_{z\in D_{\varepsilon}} (\apsi_{i,j}(z) +\apsi_{j,i}(z)) \geq \frac{1}{\varepsilon} - 2.
  \end{align*}
  By summing for all pairs $\{i,j\}$, we get
  \begin{align*}
    \sum_{z\in D_{\varepsilon}} \sum_{\{i,j\}} (\apsi_{i,j}(z) +\apsi_{j,i}(z)) \geq \binom{n}{2} \left(\frac{1}{\varepsilon} - 2\right), \\
    \sum_{z\in D_{\varepsilon}} \sum_i \sum_{j\neq i} \apsi_{i,j}(z)  \geq \binom{n}{2} \left(\frac{1}{\varepsilon} - 2\right).
  \end{align*}
  Therefore, by a variant of the mean value theorem, there exists $z\in D_{\varepsilon}$ and $i\in [n]$ such that \eqref{eq:sum-of-apsi} holds. To see this, assume that there is no such $z$ and $i$ and therefore,
  \begin{align*}
    \sum_{z\in D_{\varepsilon}} \sum_i \sum_{j\neq i} \apsi_{i,j}(z)
         &<  \sum_{z\in D_{\varepsilon}} \sum_i (n-1)(1-3\varepsilon) z \\
         &= n(n-1)(1-3\varepsilon)(\varepsilon+2\varepsilon+\cdots+1) \\
         &=n(n-1)(1-3\varepsilon)\frac{1+\varepsilon}{2\varepsilon} \\
         &\leq \binom{n}{2}\left(\frac{1}{\varepsilon}-2\right),
  \end{align*}
a contradiction.

The edges that have value 0 for $i$ and $z$ for every $j\neq i$ form a star that satisfies the properties of the lemma.
\end{proof}

By definition $\apsi_{i,j}(z)$ is a lower bound of $\psi_{i,j}^e(z)$. Furthermore, if $z=\fl{v_j^e}{\varepsilon}$, we have $\apsi_{i,j}(z)\leq \psi_{i,j}^e(z)\leq \psi_{i,j}^e(v_j^e)$. Therefore the star guaranteed by the above lemma is a nice star. However, we need a nice multi-star with appropriately high multiplicity. There is a straightforward strategy to achieve it: take a random multi-clique with sufficiently high multiplicity; between every pair of nodes, it will contain many dipoles with the same approximate boundary functions with nonzero probability, due to the fact that the number of values and the number of approximate boundary functions are bounded; by the above lemma every multi-clique formed from these dipoles will have a nice star with the same $i$ and $z$.

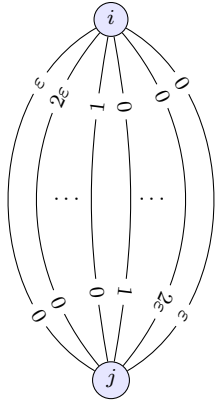
\begin{figure}
  \centering\scriptsize

  \begin{tikzpicture}[scale=2.4, allow upside down,
      vertex/.style = {circle, draw, text centered, text width=0.4em, text height=0.4em, fill=blue!10}]

      \foreach \k/\a/\x/\y in {j/2/0/0, i/2/0/2%
      }
        \node[vertex] (\k\a) at (\x,\y) {$\small \k$};

        \tikzset{decoration={snake,amplitude=.8mm}}

        \draw (i2) to [bend left=10] node[sloped,fill=white,pos=0.2] {$0$} node[sloped,fill=white,pos=0.8] {$1$} node[right] {$\cdots$} (j2);
        \draw (j2) to [bend left=10] node[sloped,fill=white,pos=0.2] {$0$}  node[sloped,fill=white,pos=0.8] {$1$} node[left] {$\cdots$} (i2);
        \draw (i2) to [bend left=40] node[sloped,fill=white,pos=0.2] {$0$} node[sloped,fill=white,pos=0.8] {$2\varepsilon$} (j2);
        \draw (j2) to [bend left=40] node[sloped,fill=white,pos=0.2] {$0$} node[sloped,fill=white,pos=0.8] {$2\varepsilon$} (i2);
        \draw (i2) to [bend left=60] node[sloped,fill=white,pos=0.2] {$0$} node[sloped,fill=white,pos=0.8] {$\varepsilon$} (j2);
        \draw (j2) to [bend left=60] node[sloped,fill=white,pos=0.2] {$0$} node[sloped,fill=white,pos=0.8] {$\varepsilon$} (i2);
        
  \end{tikzpicture}
  \caption{A dipole. All pairs ($\apsi_{i,j}^e(\cdot)$, $\apsi_{j,i}^e(\cdot)$) are identical and the set of values $\{\fl{v_i^e}{\varepsilon}, \fl{v_j^e}{\varepsilon}\}$ is equal to $W_{\varepsilon}=\{(0, z)\colon z\in D_{\varepsilon}\} \cup \{(z, 0)\colon z\in D_{\varepsilon}\}$. The figure shows $\fl{v_i^e}{\varepsilon}$ instead of $v_i^e$.
  }
  \label{fig:dipole} 
\end{figure}

\begin{lemma} \label{lemma:existence_of_all-values-multi-clique}
  Define a \AVMclique\  of multiplicity $q$ to be a multi-graph that contains $q$ disjoint \AVdipole{}s between every pair of nodes $\{i,j\}$ with the same approximate boundary functions $\apsi_{i,j}(\cdot)$ and $\apsi_{j,i}(\cdot)$. For any given $q$, there is $q'$ such that a random multi-clique of multiplicity $q'$ contains a \AVMclique\  of multiplicity $q$ with strictly positive probability.
\end{lemma}
\begin{proof}
  Consider a random multi-clique and focus on the edges between two nodes $i$ and $j$. We  argue that for sufficiently high $q'$, the probability that there are $q$ \AVdipole{}s between $i$ and $j$ is close to 1. Actually it suffices this probability to be more than $1-1/\binom{n}{2}$ to get the lemma by the union bound.

  Each one of the $q'$ edges between $i$ and $j$ is associated with a pair of values in $W_{\varepsilon}$ and a pair of approximate boundary functions. There are two crucial properties here. First, since the domain and the range of the approximate boundary functions take $2/\varepsilon$ and $(n/\varepsilon)^2$ values, respectively, there are at most $K=(n/\varepsilon)^{4/\varepsilon}$ possible pairs of approximate boundary functions. Second, the boundary functions --- and therefore the approximate boundary functions --- are independent of the values of the edges between $i$ and $j$ (Proposition~\ref{prop:2x2-independence}).

  It follows from the first property that there exists a set of at least $q'/K$ edges with identical approximate boundary functions. Fix such a set, partition it into subsets of $2/\varepsilon$ edges, and for each part consider the event that it forms an \AVdipole. The probability for each such event is equal to the probability that all the values in a part are distinct, which is $p=(2/\varepsilon)! (\varepsilon/2)^{2/\varepsilon}$. To see this, note that the probability that the $2/\varepsilon$ edges have values $(0,\varepsilon),(\varepsilon,0),\ldots,(0,1),(1,0)$ is $(\varepsilon/2)^{2/\varepsilon}$ and that there are $(2/\varepsilon)!$ permutations of these pairs of values. Taking into account that the values on the edges are selected independently, the events for each part are independent. It follows immediately that if $q'/K$ is sufficiently high, there are $q$ \AVdipole{}s with probability more than $1-1/\binom{n}{2}$.

  Although it is not needed in the proof, we give a bound of $q'$ in relation to $q$. The situation is captured by a random experiment in which there are $(q'/K)/(2/\varepsilon)$ independent Bernoulli trails, each with probability $p$ of success, and the question is to determine the probability that there are at least $q$ successes. It follows\footnote{Proof: For $X\in B(m,p)$, we have $\sigma_X^2=mp(1-p)\leq E[X]=m p$, so substituting them to Cantelli's inequality $P[X>E[X]-\lambda \sigma_X]\geq 1-1/(1+\lambda^2)$, we get $P[X>m p-\lambda\sqrt{m p}]>1-1/(1+\lambda^2)$. When $m=(2k+\lambda^2)/p$, it can be easily verified that $m p-\lambda\sqrt{m p}\geq k$, so we get the result.} from Cantelli's inequality that a binomially distributed random variable $X\sim B((2k+\lambda^2)/p,p)$ satisfies $P[X>k]\geq 1- 1/(1+\lambda^2)$. In our case, $k=q$, $\lambda^2=\binom{n}{2}$, and $p=(2/\varepsilon)! (\varepsilon/2)^{2/\varepsilon}$. Actually, in the rest of the proof, we will need $q$ to be much larger than $\binom{n}{2}$, so we can simplify $2k+\lambda^2$ to $3k$ to obtain the bound $(q'/K)/(2/\varepsilon)\geq 3q/p$. We conclude that $q' \geq 6 K q/(p\varepsilon)$ suffices for the statement of the lemma.
\end{proof}

The proof of Theorem~\ref{thm:nice multi-star} follows directly from Lemma~\ref{lemma:existence of i and z} and Lemma~\ref{lemma:existence_of_all-values-multi-clique}.

\begin{proof}[Proof of Theorem~\ref{thm:nice multi-star}]
  Lemma~\ref{lemma:existence_of_all-values-multi-clique} guarantees the existence of a \AVMclique\  of multiplicity $q$. By Lemma~\ref{lemma:existence of i and z}, there exist $i$ and $z$ such that every clique formed by its \AVdipole{}s satisfy~\eqref{eq:sum-of-apsi}. Since parallel \AVdipole{}s have the same approximate boundary functions~\footnote{but not necessarily the same boundary functions; this is the reason for using the stronger statement with approximate boundary functions in \eqref{eq:sum-of-apsi}.}, the same $i$ and $z$ work for every clique of \AVdipole{}s. 

  Therefore there exists a multi-star of multiplicity $q$ such that all its stars satisfy~\eqref{eq:sum-of-apsi}. Since $\psi_{i,j}^e(v_j^e)\geq \psi_{i,j}^e(\fl{v_j^e}{\varepsilon}) \geq \apsi_{i,j}^e(\fl{v_j^e}{\varepsilon})$, this is a nice multi-star.
\end{proof}

%% file: box.tex
\newcommand{\slicedoff}{chopped off}
\newcommand{\Slicedoff}{Chopped off}
\newcommand{\slicedoffx}{strictly chopped off}

\section{Proof of the Box Theorem (Theorem~\ref{thm:box})
}
\label{sec:box}

In this section we prove the Box Theorem~\ref{thm:box}, which states
that all multi-stars with sufficiently high multiplicity contain a
box.  Roughly, a box is a star $S$ for which when we fix the values of
its leaves, and we set the values of the root to all other tasks to
$0$, the allocation region $R_S$ of the root for obtaining all the tasks in
$S$ is rectangular (see Figure~\ref{fig:box} for an illustration,
Definition~\ref{def:box} for the precise definition, and Section~\ref{sec:region-r_p} for a precise definition and properties of the allocation region $R_S$).

Actually the proof establishes that the probability that a random star is a box tends to 1, as the multiplicity tends to infinity.

The Nice Multi-Star theorem of the previous section (Section~\ref{sec:proof-nice-multi-star}) shows that there exist nice multi-stars of any desired multiplicity and values in $(z, (1+\varepsilon)z)$, for some $z\in[\varepsilon,1]$, for a given small parameter $\varepsilon>0$. To keep the notation simple, the main theorem of this section is about multi-stars with values in $(\xi,1)$, where $\xi$ can be any strictly positive value. Naturally the argument applies to any interval of values not only to $(\xi, 1)$, and in particular to the interval $(z, (1+\varepsilon)z)$, so the two theorems can be combined to obtain the main result of this work. 

The proof of the Box Theorem is by induction on the number of leaves $k$. The case of one leaf is trivial --- all stars are boxes. However, we don't use this as the base case of the induction, because proving the theorem for $k=2$ requires different handling than the general case, and it is harder than the inductive step, in many aspects. The difference between the case of two leaves and more than two leaves is due to simple geometric facts that will become clear when we discuss the idea of the inductive step.  

Roughly speaking the inductive step goes as follows: suppose that we have established that most stars of multi-stars with $k-1$ leaves are boxes. Then we can find a star $S$ of $k$ leaves --- actually many such stars --- in which all its sub-stars of $k-1$ leaves are boxes. Now $S$ may be a box itself, in which case there is nothing to prove, or a box with a single corner cut off by a diagonal cut (see Figures~\ref{fig:box} and~\ref{fig:3dshift}). Let's call this shape ``\slicedoff\ box''.

The heart of the argument, which is based on the characterization of $2\times 2$ (2 players, 2 tasks) mechanisms, is to establish that either the \slicedoff\ box is actually a box, or we can obtain a box when one of the tasks in $S$ is replaced by a sibling task (i.e., a task of the same leaf). To show this, we take a task $p$ of the star $S$ and some sibling $p'$ of $p$ and consider the $(p,p')$-slice mechanism, that is, the mechanism for these two tasks when we fix the values of every other task. By the characterization of $2\times 2$ mechanisms, the slice mechanism is either a (relaxed) affine minimizer or a (relaxed) task independent mechanism.

If the slice mechanism is task independent the proof is relatively straightforward. The fact that tasks $p$ and $p'$ are independent implies by geometry of the allocation region that if we replace task $p$ by $p',$ %
we obtain a \emph{box} that includes $p'$ and the remaining $k-1$ tasks. 
The only real complication for this case would arise at discontinuity points, but we have excluded them by the continuity requirement.

The other case is when the $(p,p')$-slice mechanism is an affine minimizer and the allocation boundaries are linear functions. The key idea is to exploit this linearity. We keep only the linearity property for $p$ and completely ignore task $p'$ and focus on the original box. Linearity is used to show that the allocation boundaries move rectilinearly, when we change the leaf value of task $p$. In other words, the whole upper envelope of the \slicedoff\ box moves rectilinearly. But then if it is moved sufficiently close to the side, the \slicedoff\ piece will reach the boundary and create a \slicedoff\ box on it (see the 3-dimensional case in Figure~\ref{fig:3dshift}). This would contradict the inductive hypothesis that the sub-stars are boxes or equivalently that the sides are lower dimensional boxes, except for the fact that we have changed the leaf value of task $p$. To actually reach a contradiction, we use a stronger inductive hypothesis in which the side (sub-star) is a box for many values of the leaf (see the definition of \slicedoff\ box, below).

Notice however that the above argument fails for $k=2$, because a \slicedoff\ 1-dimensional interval is still an interval, that is, a box. Thus the case of $k=2$ needs to be handled differently. For this, we consider a star with two leaves and two edges per leaf and show that at least one of its four stars is a box. This is done by an argument similar to the argument for the inductive step, only that linearity is exploited for moving boundaries in two distinct directions. The fact that a multi-star with two leaves and two edges per leaf contains a box, can be combined by a known extremal graph theory result --- related to the Zarankiewicz problem --- to show that almost all stars are boxes.

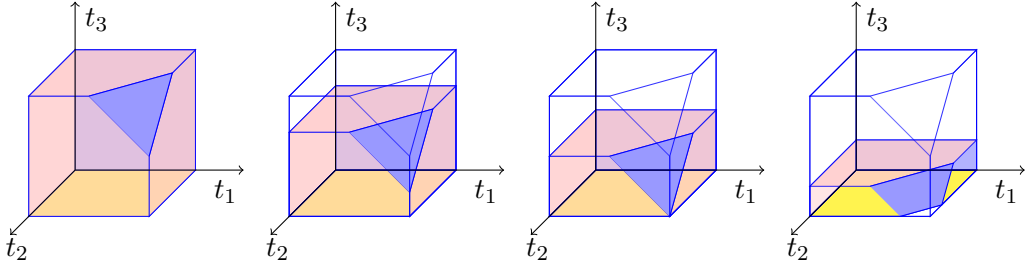
\begin{figure}
  \centering
  
\newcommand{\Depth}{2}
\newcommand{\Height}{2}
\newcommand{\Width}{2}
\newcommand{\Fraction}{0.5}

\begin{tikzpicture}[scale=0.8]
\coordinate (O) at (0,0,0);
\coordinate (A) at (0,\Width,0);
\coordinate (B) at (0,\Width,\Height);
\coordinate (C) at (0,0,\Height);
\coordinate (D) at (\Depth,0,0);
\coordinate (E) at (\Depth,\Width,0);
\coordinate (F) at (\Depth,\Width,\Height);
\coordinate (G) at (\Depth,0,\Height);

\coordinate (BF) at (\Fraction*\Depth,\Width,\Height);
\coordinate (EF) at (\Depth,\Width,\Fraction*\Height);
\coordinate (GF) at (\Depth,\Fraction*\Width,\Height);

\draw[blue,fill=yellow!80] (O) -- (C) -- (G) -- (D) -- cycle;%
\draw[blue,fill=blue!30] (O) -- (A) -- (E) -- (D) -- cycle;%
\draw[blue,fill=red!10] (O) -- (A) -- (B) -- (C) -- cycle;%
\draw[blue,fill=red!20,opacity=0.8] (D) -- (E) -- (EF) -- (GF) -- (G) -- cycle;%
\draw[blue,fill=red!20,opacity=0.6] (C) -- (B) -- (BF) -- (GF) -- (G) -- cycle;%
\draw[blue,fill=red!20,opacity=0.8] (A) -- (B) -- (BF) -- (EF) -- (E) -- cycle;%

\draw[blue,fill=blue!40,opacity=1] (BF) -- (EF) -- (GF);

\coordinate (AA) at (0,1.4*\Width,0);
\coordinate (CC) at (0,0,1.4*\Height);
\coordinate (DD) at (1.4*\Depth,0,0);

\draw[->] (O) -- (AA);
\draw[->] (O) -- (CC);
\draw[->] (O) -- (DD);

\draw (2.5,0,0) node[anchor=north] {$t_1$};
\draw (0,2.5,0) node[anchor=west] {$t_3$};
\draw (0,0,2.5) node[anchor=north] {$t_2$};

\end{tikzpicture}
\begin{tikzpicture}[scale=0.8]
\coordinate (O) at (0,0,0);
\coordinate (A) at (0,\Width,0);
\coordinate (B) at (0,\Width,\Height);
\coordinate (C) at (0,0,\Height);
\coordinate (D) at (\Depth,0,0);
\coordinate (E) at (\Depth,\Width,0);
\coordinate (F) at (\Depth,\Width,\Height);
\coordinate (G) at (\Depth,0,\Height);

\coordinate (BF) at (\Fraction*\Depth,\Width,\Height);
\coordinate (EF) at (\Depth,\Width,\Fraction*\Height);
\coordinate (GF) at (\Depth,\Fraction*\Width,\Height);

\coordinate (Aa) at (0,\Width-0.6,0);
\coordinate (Ba) at (0,\Width-0.6,\Height);
\coordinate (Ea) at (\Depth,\Width-0.6,0);
\coordinate (BFa) at (\Fraction*\Depth,\Width-0.6,\Height);
\coordinate (EFa) at (\Depth,\Width-0.6,\Fraction*\Height);
\coordinate (GFa) at (\Depth,\Fraction*\Width-0.6,\Height);

\draw[blue,fill=yellow!80] (O) -- (C) -- (G) -- (D) -- cycle;%
\draw[blue,fill=blue!30] (O) -- (Aa) -- (Ea) -- (D) -- cycle;%
\draw[blue,fill=red!10] (O) -- (Aa) -- (Ba) -- (C) -- cycle;%
\draw[blue,fill=red!20,opacity=0.8] (D) -- (Ea) -- (EFa) -- (GFa) -- (G) -- cycle;%
\draw[blue,fill=red!20,opacity=0.6] (C) -- (Ba) -- (BFa) -- (GFa) -- (G) -- cycle;%
\draw[blue,fill=red!20,opacity=0.8] (Aa) -- (Ba) -- (BFa) -- (EFa) -- (Ea) -- cycle;%

\draw[blue,fill=blue!40,opacity=1] (BFa) -- (EFa) -- (GFa);

\draw[blue] (O) -- (C) -- (G) -- (D) -- cycle;%
\draw[blue] (O) -- (A) -- (E) -- (D) -- cycle;%
\draw[blue] (O) -- (A) -- (B) -- (C) -- cycle;%
\draw[blue,opacity=0.8] (D) -- (E) -- (EF) -- (GF) -- (G) -- cycle;%
\draw[blue,opacity=0.6] (C) -- (B) -- (BF) -- (GF) -- (G) -- cycle;%
\draw[blue,opacity=0.8] (A) -- (B) -- (BF) -- (EF) -- (E) -- cycle;%

\coordinate (AA) at (0,1.4*\Width,0);
\coordinate (CC) at (0,0,1.4*\Height);
\coordinate (DD) at (1.4*\Depth,0,0);

\draw[->] (O) -- (AA);
\draw[->] (O) -- (CC);
\draw[->] (O) -- (DD);

\draw (2.5,0,0) node[anchor=north] {$t_1$};
\draw (0,2.5,0) node[anchor=west] {$t_3$};
\draw (0,0,2.5) node[anchor=north] {$t_2$};

\end{tikzpicture}
\begin{tikzpicture}[scale=0.8]
\coordinate (O) at (0,0,0);
\coordinate (A) at (0,\Width,0);
\coordinate (B) at (0,\Width,\Height);
\coordinate (C) at (0,0,\Height);
\coordinate (D) at (\Depth,0,0);
\coordinate (E) at (\Depth,\Width,0);
\coordinate (F) at (\Depth,\Width,\Height);
\coordinate (G) at (\Depth,0,\Height);

\coordinate (BF) at (\Fraction*\Depth,\Width,\Height);
\coordinate (EF) at (\Depth,\Width,\Fraction*\Height);
\coordinate (GF) at (\Depth,\Fraction*\Width,\Height);

\coordinate (Aa) at (0,\Width-1,0);
\coordinate (Ba) at (0,\Width-1,\Height);
\coordinate (Ea) at (\Depth,\Width-1,0);
\coordinate (BFa) at (\Fraction*\Depth,\Width-1,\Height);
\coordinate (EFa) at (\Depth,\Width-1,\Fraction*\Height);
\coordinate (GFa) at (\Depth,\Fraction*\Width-1,\Height);

\draw[blue,fill=yellow!80] (O) -- (C) -- (G) -- (D) -- cycle;%
\draw[blue,fill=blue!30] (O) -- (Aa) -- (Ea) -- (D) -- cycle;%
\draw[blue,fill=red!10] (O) -- (Aa) -- (Ba) -- (C) -- cycle;%
\draw[blue,fill=red!20,opacity=0.8] (D) -- (Ea) -- (EFa) -- (GFa) -- (G) -- cycle;%
\draw[blue,fill=red!20,opacity=0.6] (C) -- (Ba) -- (BFa) -- (GFa) -- (G) -- cycle;%
\draw[blue,fill=red!20,opacity=0.8] (Aa) -- (Ba) -- (BFa) -- (EFa) -- (Ea) -- cycle;%

\draw[blue,fill=blue!40,opacity=1] (BFa) -- (EFa) -- (GFa);

\draw[blue] (O) -- (C) -- (G) -- (D) -- cycle;%
\draw[blue] (O) -- (A) -- (E) -- (D) -- cycle;%
\draw[blue] (O) -- (A) -- (B) -- (C) -- cycle;%
\draw[blue,opacity=0.8] (D) -- (E) -- (EF) -- (GF) -- (G) -- cycle;%
\draw[blue,opacity=0.6] (C) -- (B) -- (BF) -- (GF) -- (G) -- cycle;%
\draw[blue,opacity=0.8] (A) -- (B) -- (BF) -- (EF) -- (E) -- cycle;%

\coordinate (AA) at (0,1.4*\Width,0);
\coordinate (CC) at (0,0,1.4*\Height);
\coordinate (DD) at (1.4*\Depth,0,0);

\draw[->] (O) -- (AA);
\draw[->] (O) -- (CC);
\draw[->] (O) -- (DD);

\draw (2.5,0,0) node[anchor=north] {$t_1$};
\draw (0,2.5,0) node[anchor=west] {$t_3$};
\draw (0,0,2.5) node[anchor=north] {$t_2$};

\end{tikzpicture}
\begin{tikzpicture}[scale=0.8]
\coordinate (O) at (0,0,0);
\coordinate (A) at (0,\Width,0);
\coordinate (B) at (0,\Width,\Height);
\coordinate (C) at (0,0,\Height);
\coordinate (D) at (\Depth,0,0);
\coordinate (E) at (\Depth,\Width,0);
\coordinate (F) at (\Depth,\Width,\Height);
\coordinate (G) at (\Depth,0,\Height);

\coordinate (BF) at (\Fraction*\Depth,\Width,\Height);
\coordinate (EF) at (\Depth,\Width,\Fraction*\Height);
\coordinate (GF) at (\Depth,\Fraction*\Width,\Height);

\coordinate (T0) at (1.5,1.5,2);

\coordinate (BGF) at (1.5,1.5,2);
\coordinate (EGF) at (2,1.5,1.5);

\coordinate (Aa) at (0,\Width-1.5,0);
\coordinate (Ba) at (0,\Width-1.5,\Height);
\coordinate (Ea) at (\Depth,\Width-1.5,0);
\coordinate (BFa) at (\Fraction*\Depth,\Width-1.5,\Height);
\coordinate (EFa) at (\Depth,\Width-1.5,\Fraction*\Height);
\coordinate (GFa) at (\Depth,\Fraction*\Width-1.5,\Height);

\coordinate (BGFa) at (1.5,0,2);
\coordinate (EGFa) at (2,0,1.5);

\draw[blue,fill=yellow!80] (O) -- (C) -- (BGFa) -- (EGFa) -- (D) -- cycle;%
\draw[blue,fill=blue!30] (O) -- (Aa) -- (Ea) -- (D) -- cycle;%
\draw[blue,fill=red!10] (O) -- (Aa) -- (Ba) -- (C) -- cycle;%
\draw[blue,fill=red!20,opacity=0.8] (Aa) -- (Ba) -- (BFa) -- (EFa) -- (Ea) -- cycle;%

\draw[blue,fill=blue!40,opacity=1] (BFa) -- (EFa) -- (EGFa) -- (BGFa);

\draw[blue] (O) -- (C) -- (G) -- (D) -- cycle;%
\draw[blue] (O) -- (A) -- (E) -- (D) -- cycle;%
\draw[blue] (O) -- (A) -- (B) -- (C) -- cycle;%
\draw[blue,opacity=0.8] (D) -- (E) -- (EF) -- (GF) -- (G) -- cycle;%
\draw[blue,opacity=0.6] (C) -- (B) -- (BF) -- (GF) -- (G) -- cycle;%
\draw[blue,opacity=0.8] (A) -- (B) -- (BF) -- (EF) -- (E) -- cycle;%

\coordinate (AA) at (0,1.4*\Width,0);
\coordinate (CC) at (0,0,1.4*\Height);
\coordinate (DD) at (1.4*\Depth,0,0);

\draw[->] (O) -- (AA);
\draw[->] (O) -- (CC);
\draw[->] (O) -- (DD);

\draw (2.5,0,0) node[anchor=north] {$t_1$};
\draw (0,2.5,0) node[anchor=west] {$t_3$};
\draw (0,0,2.5) node[anchor=north] {$t_2$};

\end{tikzpicture}

\caption{The figure is a high-level illustration of the main argument of the proof of Lemma~\ref{lemma:linearBoundary} for $k=3$. The left figure corresponds to a \slicedoffx\ box, while the other three figures illustrate the shift of the boundaries when $s_3$ is decreased, which leads to a contradiction, as the last figure is not a box for tasks $1,2$.}
\label{fig:3dshift}
\end{figure}

The rest of this section makes the above outline precise. In particular, in Subsection~\ref{sec:def-notation} we set up the notation and provide definitions, a formal statement of the main result of this section (Thm~\ref{thm:box-restated}), and useful facts and lemmas. In Section~\ref{sec:box_main_argument} we present the main argument. In Section~\ref{sec:inductionbase} we provide the proof of the base case, in Section~\ref{sec:induction_step} the proof of the induction step, and finally in Section~\ref{sub:existence} we wrap up to show existence of a box.

\subsection{Notation, definitions and preliminaries}
\label{sec:def-notation}

Most of the following definitions depend on the mechanism at hand, therefore it will be convenient to \emph{fix an arbitrary truthful mechanism throughout this section}. Since we are dealing with multi-stars and not general multi-graphs, it will be convenient to change the notation.
 
We consider multi-stars with $n$ nodes: the root is node 0 and the leaves are nodes $1,\ldots, n-1$. We name the edges $1,2,\ldots, m$, where $m=\ell(n-1)$. Edges correspond to tasks; two parallel edges are called \emph{siblings}. The set of tasks for leaf $i$ is denoted by $C_i$, and we use the term leaf for both $i$ and $C_i$.  The multiplicity of every edge is $\ell$, i.e., $|C_i|=\ell$, for every $i\in [n-1]$.

The values of an edge $j\in [m]$
between node 0 and leaf $i$ has two nonnegative values $(v_0^j,v_i^j)$, which represent the processing time for the two players; for simplicity we denote them by $(t_j, s_j)$. All the instances in this section satisfy $s_j\in [0,B)$ for some arbitrarily high value $B$, although in most of the argument $s_j$ takes values in $[0,1]$. The set of values for all edges $T=(t,s)=(t_j, s_j)_{j\in [m]}$ is called an instance.

For a given instance $T=(t,s)$, we denote the boundary function
$\psi_{0,i}^j(v_i^j)$ of the root for edge $j\in [m]$ by
$\psi_j(s_j)$. Recall that the interpretation is that, having fixed the values of
all other edges, \emph{edge $j$ is given to the root if $t_j<
  \psi_j(s_j)$, and it is given to the leaf when
  $t_j>\psi_j(s_j)$}. Since the argument sometimes considers more than
one instance, it will be useful to extend the notation to
$\psi_j(t_{-j}, s_{-j}, s_j)$ to explicitly indicate the values of the
other edges. We also write it as $\psi_j[t_{-j},s_{-j}](s_j)$ when we
want to treat it as a function of $s_j$ only. Since the values of most
other edges can be inferred from the context, we only indicate the
values that have changed inside the optional part (the part inside the
square brackets). For example, for an instance $T=(t, s)$, which can
be inferred from the context, if we change $t_1$ from its current
value to $x$, the new boundary function of edge 2 is written as
$\psi_2[t_1=x](s_2)$ or simply $\psi_2[x](s_2)$.

In this section, we use instances that satisfy the \emph{continuity requirement} (Definition~\ref{def:continuity}), although their values do not have to be rational.
We use the following fixed parameters.

\begin{definition}[Parameters] \label{def:parameters} The values of these parameters are assumed to be rational numbers.
  \begin{itemize}
  \item $\xi\in (0,1);$
  \item $\nu\in (0,\frac{\xi}{n^2\cdot 4^n})$, a very small fixed strictly positive value; %
  \item $\nu'=\nu/4$.
  \end{itemize}
\end{definition}

The aim is to prove that a multi-star $(t,s)$ with $t_j=0$ and $s_j\in (\xi, 1)$, for all $j\in[m]$ has a box with parameter $\delta=4^n\nu$. Obviously the range of values $(\xi,1)$ can be replaced by any other interval, so the theorem can be used for the nice multi-stars of the main argument, where the values are in the range $(z, (1+\varepsilon)z)$, where $\varepsilon$ is fixed small parameter.

To apply induction, we need the definition of box for smaller stars, so we repeat the definition of box (Definition~\ref{def:box}) using the
notation of this section and for the specific parameter $\delta$ that we are going to use. Note that $\delta$ depends on the size of the star.

\begin{definition}[Box]
  For a fixed $\nu$, a star $S$ with $k$ leaves, is called a \emph{$\delta$-box or simply box}, if the mechanism allocates all edges to the root, when we set $t_j=\psi_j(s_j)-\delta$ for every leaf $j$ of $S$, where $\delta=4^{k}\nu$.
\end{definition}

The following derived instances are used many times in what follows. %

\begin{definition}[Instance $T_\nu(S)$]\label{def:nu}
  For a given instance $T=(t,s)$, let $T_\nu(S)=(t^\nu, s)$ be the instance that agrees with $T$ everywhere, except for tasks in a star $S$ with $k$ leaves for which $t_j = \psi_j(s_j)-4^k\nu$.  
\end{definition}

We can now restate the main theorem of this section, a restatement of Theorem~\ref{thm:box}.
\begin{theorem}[Box] \label{thm:box-restated} For every $\nu$ and $\xi$
 that satisfy the requirements of
  Definition~\ref{def:parameters},  a multi-star instance $T=(t,\bar
  s)$ with values $t_j=0$, $\bar s_j\in [\xi,1]$ and multiplicity
   $\ell>\left (\frac{5n}{\nu}\right )^{2n}$, that satisfies the continuity requirement contains a
  $\delta$-box with $n-1$ leaves, where $\delta=4^{n-1}\nu$.
\end{theorem}

Usually we fix the values of most tasks and consider the allocation of the mechanism on the remaining tasks. In particular, we do this when we employ the characterization of $2\times 2$ mechanisms. The next definition formalizes this.

\begin{definition}[Slice and slice mechanism] \label{def:slice} Fix an instance $T$ and two tasks $p$ and $p'$. The set of instances that agree with $T$ on all tasks except for the tasks $p$ and $p'$ is called a \emph{$(p,p')$ slice}. The allocation function of these two tasks by the mechanism is called the $(p,p')$ \emph{slice mechanism} for the given values of the other tasks, or simply the $(p,p')$ slice mechanism, when the other values can be inferred by the context. Similarly, we define slice mechanisms for larger sets of tasks.
\end{definition}

\begin{definition}[Trivial leaf]
  A leaf $i$ with set of edges $C_i$ is called \emph{trivial} for a given instance $T$ if $t_j=0$ and $s_j<1$ for every task $j\in C_i.$ 
\end{definition}

We now provide the definition of the main type of instances that we consider throughout.

\begin{definition}[Standard instance] \label{def:standard} %
  For a given truthful mechanism, an instance $T=(t,\bar s)$ is a \emph{standard instance for a set of leaves} $\cal C$ if the following conditions hold
  \begin{itemize}
  \item $t_j=0$, for every $j\in [m]$
  \item $\bar s_j\in (\xi,1)$, for every task $j$ of these leaves, i.e., $j\in \cup_{i\in \cal C}C_i$,
  \item $T$ satisfies the continuity requirement. 
  \end{itemize}
  The leaves in $\cal C$ are then called \emph{standard leaves} for the instance $T.$ 
\end{definition}

Henceforth the notation $\bar s$ will denote some fixed standard instance (clear from the context). We reserve the notation $s$ for the variable, which can take any value. Note that instance $T=(t,\bar s)$ in the statement of Theorem~\ref{thm:box-restated} is standard for the set of all leaves.

The central part of the argument is an induction on the number of
leaves $k$. In the induction step from $(k-1)$ to $k$, the remaining
leaves are trivial with fixed values, and therefore play no role; in
particular they do not affect the approximation ratio.

\begin{definition}[critical values $\alpha_j$ for singletons]\label{def:alphai}
  For standard instances $T=(t, \bar s)$, we will use the shorthand $\alpha_j$ for $\psi_j(\bar s_j)$.
\end{definition}

The following observation about the boundary functions $\psi_j(s_j)$ is a straightforward generalization of Proposition~\ref{prop:bounded slope}. The proof of the lower bound is almost identical to the proof of the upper bound given in Proposition~\ref{prop:bounded slope}.

\begin{lemma} \label{obs:alphas} Let $T=(t,s)$ be an instance so that all leaves are trivial. Assuming that the approximation factor is less than $n$,  for every task $j\in [m]$  it holds that 
\begin{enumerate}
\item[(i)] $s_j/n<\psi_j(s_j)<ns_j;$%
\item[(ii)] $\lim_{s_j\rightarrow 0}\psi_j(s_j)=\psi_j(0)=0;$ 
\end{enumerate}
\end{lemma}

\subsubsection{Region $R_P$}
\label{sec:region-r_p}

The next definition of region %
$R_P$ of a given mechanism and instance $T,$ concerns a $k$-dimensional slice defined by a task set $P.\,$ The other tasks $[m]\setminus P$ have fixed values and for simplicity they are not treated in the definition. In most cases when $R_P$ is considered, $P$ will be a star, or a pair of siblings, and  the tasks $[m]\setminus P$ will be trivial.

\begin{definition}[region $R_P$ and $R_{\emptyset|P}$]\label{def:regionR} Let $T=(t,s)$ be a given instance and $P=\{p_1,\ldots,p_k\}$ be a set of tasks. The allocation region  $R^T_P\subset \mathbb R^k_{0,+}$ consists of all vectors $t'_P$ such that for input $T'=((t'_P, t_{-P}), s),$ all tasks of $P$ are allocated to the root. Similarly, $R^T_{\emptyset|P}$ consists of all the $t'_P$ so that for $((t'_P, t_{-P}), s),$ all tasks of $P$ are allocated to the leaves. $T$ is omitted from the notation when it is clear from the context.
\end{definition}

It follows directly by the Weak Monotonicity property that $R_P $ is a
(possibly degenerate) $k$-dimensional polyhedron defined by linear
constraints of the form $\sum_{i\in I} t_{p_i}\leq c_I,$ for some
$I\subset [k]$, and some $c_I\in \mathbb R_+$. In particular, for
$k=2$, $R_P$ is either a rectangle or a rectangle from which we cut
off a piece by a $-45^o$ cut. For higher dimensions, it is a
hyperrectangle (orthotope) with pieces cut off by specific
hypercuts. See Figure~\ref{fig:box} for examples of $R_P$ for $k=2$
and $k=3$ and see also~\cite{Vid09} for a geometric interpretation of
truthfulness. Note that $R_P\subseteq \times_{i=1}^k[0,\alpha_{p_i}]$,
where $\alpha_{p_i}=\psi_{p_i}(s_{p_i})$.

Two particular shapes of $R_P$ play significant role in the argument,
boxes (complete orthotopes), and \slicedoff\ boxes, which are boxes
with a single corner removed by a diagonal cut of the form $\sum_{i
  \in [k]} t_{p_i} = c_{[k]}$. Strictly speaking, in the definition of
box and \slicedoff\ box, we allow thin pieces of width at most
$\delta$ to be missing from its faces.

More generally, when the facet $\sum_{i\in [k]} t_{p_i}=c_{[k]}$ of
$R_P$ exists, we call it \emph{bundling facet}. When the bundling
boundary facet exists, it separates the regions $R_P$ and
$R_{\emptyset|P}$.

The following obvious lemma will be useful later.

\begin{lemma}\label{obs:Rnonempty} Let $T=(t,s)$ be an instance with
  $t_j=0$ for all $j$, and assume that for some set of tasks
  $P=\{p_1,p_2,\ldots, p_k\}$ the region $R_P$ has a full dimensional
  bundling facet, i.e., it contains a point with strictly positive
  coordinates. Then $\psi_{p_j}(s_{p_j})>0$ and $s_{p_j}>0$ for every
  task $p_j\in P$.
\end{lemma}

\begin{proof} Let $\hat t_P$ be a point on the bundling facet with
  strictly positive coordinates. For some sufficiently small
  $\epsilon>0$, consider the point $(t_{p_j}=\hat
  t_{p_j}-\epsilon\,,\, (t_{p_i}=0)_{i\in[k]\setminus{j}})$. By weak
  monotonicity, this input point is in $R_P,$ so all tasks in $P$
  are given to the root. It follows that $\psi_{p_j}(s_{p_j})\geq \hat
  t_{p_j}-\epsilon>0$. The fact that $s_{p_j}$ is also strictly
  positive follows from Lemma~\ref{obs:alphas}, since $s_{p_j}>
  \psi_{p_j}(s_{p_j})/n$.
\end{proof}

The following is a useful proposition from~\cite{CKK21b}.

\begin{proposition}[\cite{CKK21b} Lemma 21. Claim (i)]\label{prop:lipschitz}
For every truthful mechanism, and arbitrary $T=(t,s),$ the boundary function $\psi_r(t_{-r},s)$ is 1-Lipschitz in $t_{-r}$, i.e.,
    $|\psi_r(t_{-r},s)-\psi_r(t_{-r}',s)|\leq |t_{-r}-t_{-r}'|_1$.
\end{proposition}

\subsubsection{Facts about two tasks.} \label{sec:facts-about-two}

In this section we summarize known concepts and results for two tasks, i.e., two edges that can belong to the same leaf (sibling tasks) or to two different leaves (star of two tasks) (see for example~\cite{CKK21b}). Here we assume that all other tasks are fixed, and omit them from the notation. So let simply $t=(t_1, t_2)$ and $s=(s_1,s_2)$.
First we consider allocation regions for the root depending on his own
bids $(t_1,t_2)$ for fixed values $s=(s_1,s_2)$. Here the $s_1$ and $s_2$ do
not necessarily belong to the same leaf. It is known that in the case
of two tasks, for a fixed $s$ the four possible allocation regions of
the root in a weakly monotone allocation subdivide $\mathbb R_{\geq
  0}^2$ basically in three possible ways summarized in the following definition (see
Figure~\ref{fig:shapesPure} for an illustration).

\begin{definition}\label{def:crossing}
  For given $(s_1,s_2)$ we call the allocation for the root 
  \begin{itemize}
  \item \emph{quasi-bundling}, if there are at least two points $t\neq t'$ on
    the boundary of $R_{\{1,2\}}$ and $R_{\emptyset|\{1,2\}};$ 
  \item \emph{quasi-flipping}, if $R_{\{1,2\}}$ and $R_{\emptyset}$ have no common boundary point;  
  \item \emph{crossing}, otherwise, if there is a unique common boundary point.
  \end{itemize}
  We sometimes refer collectively to both quasi-bundling and quasi-flipping allocations as \emph{non-crossing}.
  \end{definition}

\begin{figure}[h]
  \centering
  \begin{tikzpicture}[scale=0.48]

    \draw[->] (0,0) -- (6,0) node[anchor=north] {$t_{1}$};
    \draw[->] (0,0) -- (0,6) node[anchor=east] {$t_{2}$};
    \draw[very thick, blue] (1.9,6) -- (1.9,3.68) -- (4,1.58) -- (6,1.58) ; \draw[very thick, blue] (0, 3.68) -- (1.9,3.68) ;

    \draw[ultra thick, red, dashed] (2,6) node[anchor = south] {$\psi_1[t_2](s_1)$} -- (2,3.7) -- (4.1,1.6) ; \draw[ultra thick, red, dashed] (4.1,1.6) -- (4.1,0);

    \draw[very thick, blue] (4,0) -- (4, 1.58) -- (1.9, 3.68);

    \draw (1.5,5) node[anchor=center] {}; \draw (1.7,1) node[anchor=center] {$R_{\{1,2\}}$}; \draw (5.5,1) node[anchor=center] {}; \draw (5.5,5) node[anchor=east] {$R_{\emptyset|{\{1,2\}}}$};

    \draw (3,-0.5) node[anchor=north] {(a)};
  \end{tikzpicture}\hspace{20mm}
  \begin{tikzpicture}[scale=0.48]

    \draw[->] (0,0) -- (6,0) node[anchor=north] {$t_{1}$};
    \draw[->] (0,0) -- (0,6) node[anchor=east] {$t_{2}$};
    \draw[very thick, blue] (4.1,6) -- (4.1, 3.68) -- (1.9,1.58) -- (0,1.58) ;

    \draw[ultra thick, red, dashed] (4.2, 6) node[anchor = south] {$\psi_1[t_2](s_1)$} -- (4.2, 3.7) -- (2,1.6) ; \draw[ultra thick, red, dashed] (2,1.6) -- (2,0);

    \draw[very thick, blue] (1.9,0) -- (1.9, 1.58); \draw[very thick, blue] (4.1,3.68) -- (6, 3.68) ;

    \draw (1.5,5) node[anchor=center] {}; \draw (1.7,1) node[anchor=east] {$R_{\{1,2\}}$}; \draw (5.5,1) node[anchor=east] {}; \draw (5.5,5) node[anchor=west] {$R_{\emptyset|\{1,2\}}$};

    \draw (3,-0.5) node[anchor=north] {(b)};
  \end{tikzpicture}
  
  \vspace{6mm}
  
  \begin{tikzpicture}[scale=0.48]
    \draw[->] (0,0) -- (6,0) node[anchor=north] {$t_{1}$};
    \draw[->] (0,0) -- (0,6) node[anchor=east] {$t_{2}$};
    \draw[very thick, blue] (3,0) -- (3,6) ; \draw[very thick, blue] (0,3) -- (6,3) ;

    \draw[ultra thick, red, dashed] (3.1, 0) -- (3.1,6) node[anchor = south] {$\psi_1[t_2](s_1)$} ;

    \draw (1.5,5) node[anchor=center] {}; \draw (1.7,1) node[anchor=center] {$R_{\{1,2\}}$}; \draw (5.5,1) node[anchor=east] {}; \draw (5.5,5) node[anchor=center] {$R_{\emptyset|\{1,2\}}$};

    \draw (3,-0.5) node[anchor=north] {(c)};
  \end{tikzpicture}\hspace{20mm}
  \begin{tikzpicture}[scale=0.48]

    \draw[->] (0,0) -- (6,0) node[anchor=north] {$t_{1}$};
    \draw[->] (0,0) -- (0,6) node[anchor=east] {$t_{2}$};
    \draw[very thick, blue] (1.9,6) -- (1.9,3.68) -- (5.58,0);%
    \draw[very thick, blue] (0, 3.68) -- (1.9,3.68) ;

    \draw (1.5,5) node[anchor=center] {}; \draw (1.7,1) node[anchor=center] {$R_{\{1,2\}}$}; \draw (5.5,1) node[anchor=center] {}; \draw (5.5,5) node[anchor=east] {$R_{\emptyset|{\{1,2\}}}$};

    \draw (3,-0.5) node[anchor=north] {(d)};
  \end{tikzpicture}

  \caption{\small The possible allocation to the root depending on their own bid vector $(t_1,t_2)$ for fixed values $s=(s_1,s_2)$ of the other player: (a) quasi-bundling allocation; (b) quasi-flipping allocation; (c) crossing allocation. The boundary $\psi_1[t_2](s_1)$ for task $1$ is shown by a dashed line. Part (d) depicts a half-bundling allocation at task $1$, a degenerate case of quasi-bundling allocation.}
  \label{fig:shapesPure}
\end{figure}
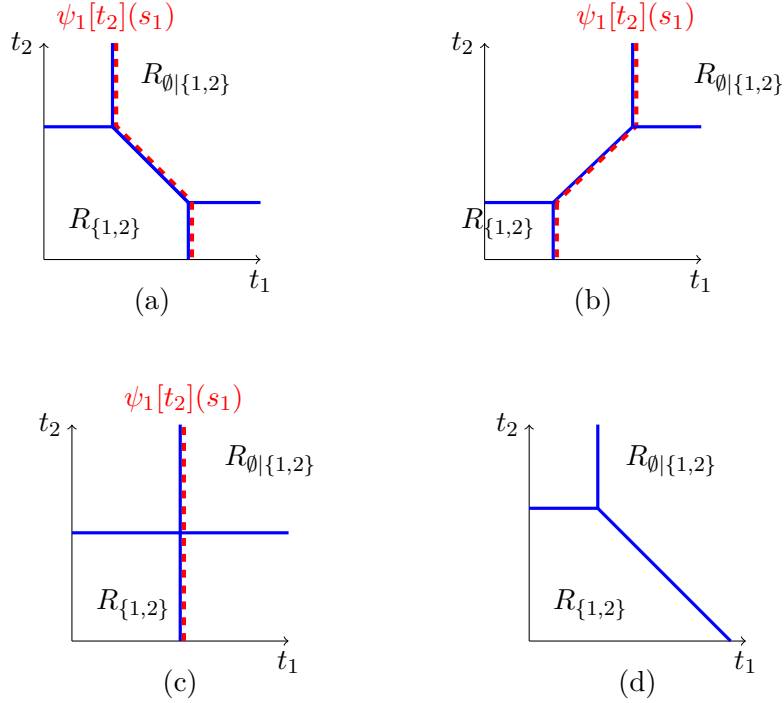

  In case the allocation is quasi-bundling, the boundary between $R_{\{1,2\}}$ and $R_{\emptyset|\{1,2\}}$ is a bundling boundary.  In general the facets defining $R_{\{1,2\}}$ in $\mathbb R^2_{0,+}$ are of the possible forms $t_1=c_1,\, t_{2}=c_2,\,$ and $t_1+t_{2}= c.$ A nontrivial facet $t_1+t_{2}= c$ exists if and only if the allocation is quasi-bundling.

  We need to distinguish and treat separately a special sort of
  quasi-bundling allocation, because in this case the boundary
  function $\psi_1$ can become nonlinear function of $s_1$, even in affine
  minimizers  (more precisely, in their generalizations called \emph{relaxed} affine minimizers, see Theorem~\ref{theo:addchar} and \cite{CKK20}).
  If the facet $t_1+t_2= c$ exists but the facet $t_1=c_1$ is missing
  (i.e., $c_1>c$), then we will call the allocation of the tasks
  half-bundling, as defined next (see Figure~\ref{fig:shapesPure}
  (d)).

  \begin{definition}\label{def:half-bundling} The allocation of two tasks $(1,2)$ (siblings or not) is called \emph{half-bundling at task $1$,} if the region $R_{\{1,2\}}$ is nonempty and is defined by at most two facets $t_{2}=c_2$ and $t_1+t_{2}=c,$ but no nontrivial facet $t_1=c_1$ exists i.e., $c_1> c$.
    The allocation will be called \emph{fully bundling,} if it is half-bundling at both task $1$ and task $2$ i.e. both $c_1,c_2>c$.
\end{definition}

Finally, we include here the characterization result for $2$ machines
and $2$ tasks ($2\times 2$ case).\footnote{For a description of these
  mechanisms see \cite{CKK20, CKK21b}. For related characterization
  results see \cite{ChristodoulouKV08, DS08}.}  That is, from here on
we consider two sibling tasks: the input $(t_1,t_2)$ is of the root,
and $(s_1,s_2)$ belong to the same leaf. The characterization assumes
that for fixed values $s$ and sufficiently high values of the root,
both tasks will be allocated to the leaf; otherwise the approximation
ratio is unbounded (even when there exist other players and tasks with
fixed values). Recall that in our instances $T=(t,s)$ the values for
the root can be arbitrarily high, but for the $s$ players they belong in
$[0,B)$. %
The theorem from \cite{CKK20} that we use, characterizes mechanisms
with values exactly in this domain.

\begin{theorem}[Characterization of $2\times 2$ mechanisms~\cite{CKK20}] \label{theo:addchar} For every $B\in \R_{>0}\cup\{\infty\}$, every weakly monotone allocation for two tasks and two players with values $t\in [0,\infty)\times [0,\infty)$ and $s\in[0,B)\times [0,B)$ -- such that for every $s$ there exists $t$ for which both tasks are allocated to the second player -- belongs to the following classes: (1) relaxed affine minimizers (including the special case of affine minimizers), (2) relaxed task independent mechanisms (including the special case of task independent mechanisms), (3) 1-dimensional mechanisms, (4) constant mechanisms.
\end{theorem}

\begin{remark} Two important remarks are due. %

  (a) Strictly speaking, the allocation of the root in an affine minimizer adheres to the definition over the open region $s\in (0,B)\times(0,B)$. If $s_1=0$ or $s_2=0$ the $\psi$ boundaries of (even) an affine minimizer can have jump discontinuities (and vice versa for the leaf, see \cite{CKK20}). However, this will not affect our proof because (i) no jump in $s_p=0$ in a $2\times 2$ mechanism can occur if every task other than the siblings $(p,p')$ is trivial and the approximation is finite (by Lemma~\ref{obs:alphas}, $\lim_{s_p\rightarrow 0}\psi_p(s_p)=\psi_p(0)=0$); (ii) whenever the other tasks are non trivial, and the characterization for $(p,p')$ is used, we never consider an instance with $s_p=0.$ 

(b) 1-dimensional bundling mechanisms (where the 2 tasks are always bundled) are special relaxed affine minimizers with some additive constants that are $\infty$. 1-dimensional task independent mechanisms (where one of the tasks is always given to the same player) are special relaxed task independent mechanisms. Constant mechanisms (these are independent of $(s_1,s_2)$) are special affine minimizers with a multiplicative constant $0$. In this sense,  the only possible $2\times 2$ mechanisms are \emph{relaxed affine minimizers} and \emph{relaxed task independent mechanisms}. 

Some mechanisms can be of two types; for example the VCG mechanism is both an affine minimizer and task independent, and there exist other task independent affine minimizers. These will be treated as task independent mechanisms.
\end{remark}

We summarize the most relevant properties of these  mechanisms in terms of the possible allocation figures (see \cite{CKK20, CKK21b}). Roughly speaking, when there exists a bundling (or flipping) boundary and three or four allocation regions, the boundary functions are affine.

\begin{observation}\label{obs:possibleFigures}
  The following properties hold for a $2\times 2$ truthful mechanism.
\begin{itemize}
\item[(i)] the allocation of a relaxed task independent mechanism is crossing for every $(s_1,s_2),$ except for countably many points $(s_1',s_2'),$ where both $\psi_1 (s_1)$ has jump discontinuity in $s_1',$ \emph{and} $\psi_2 (s_2)$ has jump discontinuity in $s_2';$

\item[(ii)] the allocation of a (non task independent) relaxed affine minimizer is non crossing for every $(s_1, s_2);$\footnote{unless, because of small $s,$ the allocation has at most two regions} either it is always quasi-bundling, or always quasi-flipping, with the same length of slanted boundary for every high enough $s$ --- as $s_1$ and/or $s_2$ gets smaller, part of the slanted boundary may disappear, thus it may get shorter (moreover, in degenerate relaxed affine minimizers, the possible length of the slanted boundary may be unbounded); %

\item[(iii)] the boundary function $\psi_1[t_{2},s_{2}](s_1)$ of a relaxed affine minimizer is a truncated linear function for every $(t_2, s_1, s_2),$ unless the allocation is fully bundling:\footnote{Therefore, if the allocation is half-bundling at task 1, then for reduced $s_1$ the $\psi_1(s_1)$ can become non-linear when the figure becomes fully bundling.} 
\begin{align*}
\psi_1[t_2,s_2](s_1)=\max(0\,, \,\lambda(t_2,s_2)\, s_1 - \gamma(t_2,s_2)\,).
\end{align*}
      
Symmetric statements hold for $\psi_2.$       
\end{itemize}
\label{ref:obs-characterization}
\end{observation}

The fact that relaxed task independent mechanisms have discontinuities could create complications, but we avoid them by considering instances that satisfy the continuity requirement.

\subsection{Proof of the main Box Theorem~\ref{thm:box-restated}}
\label{sec:box_main_argument}

The proof of the Box Theorem has some similarities to the main lemma in \cite{CKK21b}, but both the general structure and the details of the proof must be handled differently.

\begin{proof}[Proof of Theorem~\ref{thm:box-restated}]
  Fix some standard instance of $n-1$ leaves and sufficiently high multiplicity $\ell$. We show by induction on $k$ that in every subset of $k\in [n-1]$ leaves, \emph{a random star selected uniformly and independently is a box} with a certain positive probability.

  More precisely, let $1-b_k$ be (a lower bound on) the probability that a random star of $k$ leaves is a box. We establish a recurrence on $b_k$ that shows that $b_{n-1}<1$, which proves the existence of at least one box of $n-1$ leaves.

Specifically, Theorem~\ref{cor:b2} establishes the base case of the induction ($k=2$), which shows that $b_2\leq 2/\sqrt{\ell}$. Lemma~\ref{lemma:recurrence} based on the proof of the inductive hypothesis establishes the recurrence $$b_k\leq \left (\frac{5n}{\nu}-1 \right )b_{k-1}+ \frac{2n^3}{\xi \sqrt{\ell} },$$ for $k\geq 3$. It follows that for sufficiently large $\ell$, $b_{n-1}$ can be arbitrarily small (see~Corollary~\ref{cor:bound-on-b} for a more precise bound on $b_k$).
\end{proof}

\subsubsection{Preliminary observations}
\label{sec:preliminary}

The following theorem, essentially a restatement of
Proposition~\ref{prop:2x2-independence}, says that if the mechanism
has bounded approximation ratio, the allocation of the slice mechanism of two siblings must be crossing. It also excludes the degenerate case of constant
mechanisms.

\begin{theorem} \label{thm:sibling-independence} Let $T=(t,\bar s)$ be
  a standard instance. If the mechanism has approximation ratio at
  most $n$, the boundary function $\psi_p[t_{-p},\bar s_{-p}](s_p)$ of
  a task $p$ is independent of the values of its
  siblings. Furthermore, for an arbitrary sibling $p'$ the
  $(p,p')$-slice mechanism cannot be constant, even when we change $t$
  to non-zero values.
\end{theorem}

\begin{proof}
  The theorem is a direct consequence of the following two lemmas. The
  first lemma excludes the degenerate case of constant mechanisms. The
  next one states that slice mechanism of two sibling tasks cannot
  have a quasi-bundling or quasi-flipping boundary, when the
  approximation ratio is bounded. By the characterization, the slice
  mechanism must be a relaxed task independent mechanism. By the
  continuity requirement, there are no discontinuities, so
  $\psi_p[t_{-p},\bar s_{-p}](s_p)$ of a task $p$ is independent of
  the values of its siblings.
\end{proof}

\begin{lemma}\label{prop:noconstant} Assume that the mechanism has approximation ratio at most $n$. Let $T=(t,\bar s)$ be a standard instance and $Q$ any leaf, with sibling tasks $p, p'\in Q.$ Furthermore, for some $k\leq n-1$ let $P=\{p_1, p_2, \ldots , p_k\}$ be an arbitrary set of $k$ other tasks.
 If in $T$ we increase each $t_{p_i}\quad (p_i\in P) $ to some arbitrary value $\hat t_{p_i}\geq 0,$ %
 then in the obtained instance $\hat T$ the slice mechanism $(p, p')$ is not a constant mechanism, unless the approximation ratio of the whole mechanism is at least $n.$%
\end{lemma}

\begin{proof} Consider the $(p, p')$ slice mechanism, and let
  $(t,t')=(t_{p}, t_{p'})$ and $(s,s')=(s_{p}, s_{p'})$ denote an
  input for these two tasks in general. A $2\times 2$ constant
  mechanism is by definition independent of the values of (at least)
  the root or the leaf. Looking at the allocation of the root (for any
  fixed $s$), w.l.o.g. the region $R_{\emptyset|\{p,p'\}}$ is nonempty
  (otherwise we could increase to $(t, t')=(\infty,\infty),$ so that
  the root still gets a task, showing unbounded approximation
  ratio). So, if the mechanism were independent of the root, then for
  \emph{every} $(s, s')$ the allocation would consist only of
  $R_{\emptyset|\{p,p'\}}$ -- the root would never get any task, also
  independently of $s$.
  Thus, {\em in any case the
  allocation must be independent of $s.$} The contribution to $OPT$ of
  every task other than $p, p',$ is at most $\sum_{p_i\in P} \bar
  s_{p_i}\leq (n-1)\cdot 1,$ since $T$ is a standard
  instance. %

Now, set $(t,t')=(n^2, n^2).$ If the root receives no task from $\{p,p'\},$ then increase $(s,s')=(2n^3, 0).$ Then $\mech\geq 2n^3,$ and $\opt< n+n^2< 2n^2,$ proving high approximation. If the root receives at least one of the tasks, then $\mech\geq n^2.$ Setting $(s,s')=(0,0),$ we have $\opt< n,$ so again, we have high approximation ratio.\end{proof}

\begin{figure}
  \centering
  \begin{tikzpicture}[scale=0.35]

    \draw[->] (0,0) -- (8,0) node[anchor=north] {$t_p$};
    \draw[->] (0,0) -- (0,8) node[anchor=east] {$t_{p'}$};

    \draw[very thick, blue] (0, 3) node[anchor=east,black] {$\alpha_{p'}$} -- (4, 3) -- (6, 5) -- (8, 5); 
\draw[very thick, blue ] (4, 0) node[anchor=north,black] {$\alpha_p$}-- (4, 3) -- (6, 5) -- (6, 8);
    \draw[thick, dotted, red ] (0,5) node[anchor=east,black] {$\alpha_{p'}+\alpha_{p,p'}$} -- (6,5);

\draw (1.5,7) node[anchor=west] {}; \draw (2,1) node {$R_{\{p,p'\}}$}; \draw (7.7,2) node[anchor=east] {}; \draw (7.7,7) node[anchor=west] {$R_{\emptyset|\{p,p'\}}$};

    \draw (1,-1) node[anchor=north] {(a)};
  \end{tikzpicture}
\hspace*{1cm}
  \begin{tikzpicture}[scale=0.35]

    \draw[->] (0,0) -- (8,0) node[anchor=north] {$t_p$};
    \draw[->] (0,0) -- (0,8) node[anchor=east] {$t_{p'}$};

    \draw[very thick, blue] (0, 0.5)  -- (1.5, 0.5) -- (3.5, 2.5) -- (8, 2.5); 
\draw[very thick, blue ] (1.5, 0) -- (1.5, 0.5) -- (3.5, 2.5) -- (3.5, 8);

\draw (1.5,7) node[anchor=west] {}; 
\draw (7.7,2) node[anchor=east] {}; 
\draw (7.7,7) node[anchor=west] {$R_{\emptyset|\{p,p'\}}$};
\draw[thick, dotted, red ] (0,1.5) node[anchor=east,black] {$\frac{\alpha_{p,p'}}{2}$} -- (2.3,1.5);  
\draw[thick, dotted, red ] (2.3,0) node[anchor=north,black] {$\frac{\alpha_{p,p'}}{2}$} -- (2.3,1.5);  

    \draw (1,-1) node[anchor=north] {(b)};
  \end{tikzpicture}
\\
\hspace*{1cm}
  \begin{tikzpicture}[scale=0.35]

    \draw[->] (0,0) -- (8,0) node[anchor=north] {$t_p$};
    \draw[->] (0,0) -- (0,8) node[anchor=east] {$t_{p'}$};

    \draw[very thick, blue] (0, 6) node[anchor=east,black] {$\alpha_{p'}$} -- (4, 6) -- (6, 4) -- (8, 4); \draw[very thick, blue ] (6, 0) node[anchor=north,black] {$\alpha_p$}-- (6, 4) -- (4, 6) -- (4, 8);

\draw (1.5,7) node[anchor=west] {}; \draw (1.7,2) node[anchor=west] {$R_{\{p,p'\}}$}; \draw (7.7,2) node[anchor=east] {}; \draw (8.5,7) node[anchor=east] {$R_{\emptyset|\{p,p'\}}$};
\draw[thick, dotted, red ] (0,4) node[anchor=east,black] {$\alpha_{p'}-\alpha_{p,p'}$} -- (6,4);  

    \draw (1,-1) node[anchor=north] {(c)};
  \end{tikzpicture}
\hspace*{1cm}
  \begin{tikzpicture}[scale=0.35]

    \draw[->] (0,0) -- (8,0) node[anchor=north] {$t_p$};
    \draw[->] (0,0) -- (0,8) node[anchor=east] {$t_{p'}$};

    \draw[very thick, blue] (0, 6) -- (0.5, 6) -- (2.5, 4) -- (8, 4); 
\draw[very thick, blue ] (2.5, 0) node[anchor=north,black] {$\alpha_{p,p'}+\frac{\epsilon}{2}$}-- (2.5, 4) -- (0.5, 6) -- (0.5, 8);

\draw (1.5,7) node[anchor=west] {}; \draw (1.7,2) node[anchor=west] {}; \draw (7.7,2) node[anchor=east] {}; \draw (7.7,7) node[anchor=east] {$R_{\emptyset|\{p,p'\}}$};

    \draw (1,-1.5) node[anchor=north] {(d)};
  \end{tikzpicture}
\hspace*{1cm}
  \begin{tikzpicture}[scale=0.35]

    \draw[->] (0,0) -- (8,0) node[anchor=north] {$t_p$};
    \draw[->] (0,0) -- (0,8) node[anchor=east] {$t_{p'}$};

    \draw[very thick, blue] (0, 0.5) -- (0.5, 0.5) -- (1, 0); 
\draw[very thick, blue ] (1, 0) node[anchor=north,black] {}-- (1, 0) -- (0.5, 0.5) -- (0.5, 8);

    \draw (-0.23,0.5) node[anchor=west] {};
    \draw (7,7) node[anchor=east] {$R_{\emptyset|\{p,p'\}}$};
    \draw[very thick, dotted, red] (1.5,0) node[anchor=north,black]
    {$\epsilon$} -- (1.5,1.5) -- (0,1.5) node[anchor=east,black]    {$\epsilon$};
    \draw (1,-1) node[anchor=north] {(e)};
  \end{tikzpicture}
  
  \caption{\small High-level illustration of the proof of Lemma~\ref{prop:quasi} }
  \label{fig:affine-unbounded}
\end{figure}
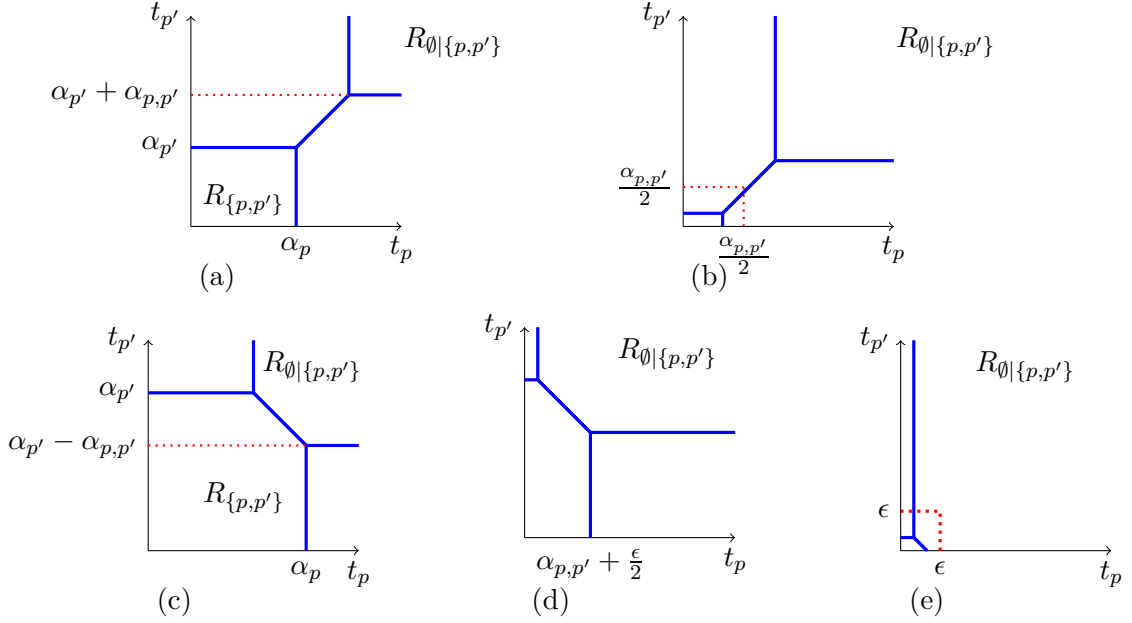

\begin{lemma}\label{prop:quasi} Let $T=(t,\bar s)$ be a standard instance, and $Q$ be any leaf. Assume that two sibling tasks $p, p' \in Q$ exist so that the allocation in the $(p, p')$-slice mechanism is quasi-bundling or quasi-flipping. Then the mechanism has infinite approximation ratio.
\end{lemma}

\begin{proof} By the characterization result  (Theorem~\ref{theo:addchar} and Observation~\ref{ref:obs-characterization}), the $(p,p')$-slice
  mechanism is a (non task independent) relaxed affine minimizer.

Consider first the case that the affine minimizer is quasi-flipping (Fig.~\ref{fig:affine-unbounded} (a)),
with height $\alpha_{p,p'}$ of the flipping boundary (i.e., the height
of the $45^o$ boundary). Lemma~\ref{obs:alphas} implies
$\psi_{p}(s_p)\rightarrow 0$ when $s_p\rightarrow 0,$ and
$\psi_{p'}(s_{p'})\rightarrow 0$ when $s_{p'}\rightarrow 0$, hence for
$(s_p,s_{p'})=(\epsilon,\epsilon)$ the allocation is still
quasi-flipping.  In particular, because $0<\psi_p(\epsilon)\leq
n\epsilon,$ and $0<\psi_{p'}(\epsilon)\leq n\epsilon,$ and the
flipping boundary has height $\alpha_{p,p'},$ exactly one of the tasks
is given to the root for values $(t_p,
t_{p'})=(\alpha_{p,p'}/2,\alpha_{p,p'}/2)$ (Fig.~\ref{fig:affine-unbounded} (b)). The approximation ratio is
$\mech/\opt>(\alpha_{p,p'}/2)/ 2\epsilon\rightarrow \infty$ when
$\epsilon \rightarrow 0.$ Finally, in the quasi-flipping case $\alpha_{pp'}=\infty$ can be excluded by Lemma~\ref{obs:alphas} (i) (to put it simply, both $R_{\{p,p'\}}$ and $R_{\emptyset|\{p,p'\}}$ must be nonempty).

Now consider the case that the relaxed affine minimizer is
quasi-bundling with height of the bundling boundary equal to
$\alpha_{p,p'}$ (Fig.~\ref{fig:affine-unbounded} (c)). The $2\times 2$ characterization implies that $\alpha_{p,p'}$
is constant (i.e., independent of the $(s,t)$ values when $s$ large enough, see Observation~\ref{ref:obs-characterization} ). Set $s_{p}$ so that
$\psi_p(s_p)=\alpha_{p,p'}+\epsilon/2$ (Fig.~\ref{fig:affine-unbounded} (d)).  Now for some large $s_{p'}$,
the allocation is almost half-bundling at $p'$, and by the definition
of relaxed affine minimizers, the allocation remains quasi-bundling
when we decrease $s_{p'}.$ If we set $s_{p'}=\epsilon/(2n)$ and
$t=(\epsilon, \epsilon)$, taking into account that
$\psi_{p'}(s_{p'})<n\, s_{p'}=\epsilon/2$, the existence of the
bundling boundary guarantees that both tasks are given to the leaf (Fig.~\ref{fig:affine-unbounded} (e)). In
turn this gives unbounded approximation ratio because $\mech\geq
s_p>\psi_p(s_p)/n>\alpha_{p,p'}/n$, while $\opt\leq 2\epsilon$.

If in the quasi-bundling case $\alpha_{pp'}=\infty$ (e.g., for one-dimensional bundling mechanisms) the argument is analogous. For $s_p=1,\, s_{p'}=\epsilon/2n$ and  $t=(\epsilon, \epsilon),$ both tasks are allocated to the leaves and we obtain unbounded ratio.
\end{proof}

\begin{corollary}\label{cor:bundlingsiblings} Let $T=(t,\bar s)$ be a
  standard instance, $p,p'$ be sibling tasks and $q$ be an arbitrary
  third task. If the region $R_{\{p,p',q\}}$ has a bundling facet of
  equation $t_p+t_{p'}+t_q=D$ for some $D<\alpha_p+\alpha_{p'},$ then
  $R_{\{p,p'\}}$ is quasi-bundling, so by Lemma~\ref{prop:quasi} the
  mechanism has unbounded ratio.
\end{corollary}

\begin{proof} By definition, in the $(p,p')$ slice the points
  $(\alpha_p,0)$ and $(0, \alpha_{p'})$ are on the boundary of
  $R_{\{p,p'\}}.$ On the other hand, in $R_{\{p,p',q\}}$ the point
  $(\alpha_p,\alpha_{p'},0)$ strictly violates $t_p+t_{p'}+t_q\leq D.$ 
  Therefore, the point $(\alpha_p,\alpha_{p'})$ does not belong in $R_{\{p,p'\}}$ so this allocation region can only be quasi-bundling (See also
  Figure~\ref{fig:LemmaToClaimIV}.).
\end{proof}

\subsection{Base case}
\label{sec:inductionbase}

In this section we prove the special case of the Box Theorem (Theorem~\ref{thm:box-restated}) for two leaves. In this particular case, multiplicity 2 is sufficient.

 \begin{theorem} \label{thm:basecase}
   For a mechanism with approximation ratio at most $n$, every standard instance with 2 leaves and 2 tasks per leaf contain a box of 2 leaves.   
 \end{theorem}

 The proof of the theorem is given in this section. We consider a multi-star with two leaves and two tasks per leaf, as shown in Figure~\ref{fig:2+2-multi-star}. We rename the two tasks of leaf 1 as 1,1' and the two tasks of leaf 2 as 2,2'.

 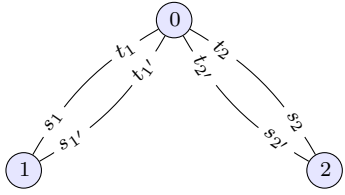
\begin{figure}
   \centering
   
 \begin{tikzpicture}[scale=2, allow upside down,
      vertex/.style = {circle, draw, text centered, text width=0.4em, text
        height=0.5em, fill=blue!10}]

      \scriptsize

      \foreach \k/\x/\y/\l in {1/1/2/0, 3/2/1/2, 2/0/1/1}
        \node[vertex] (\k) at (\x,\y) {$\small \l$};

      \foreach \from/\to/\la/\lb in {2/1/$s_1$/$t_1$, 1/3/$t_2$/$s_2$}
      \draw (\from)  to [bend left=15] node[sloped,fill=white,pos=0.2] {\la} node[sloped,fill=white,pos=0.8] {\lb}  (\to);
      \foreach \from/\to/\la/\lb in {2/1/$s_{1'}$/$t_{1'}$, 1/3/$t_{2'}$/$s_{2'}$}
      \draw (\from)  to [bend right=15] node[sloped,fill=white,pos=0.2] {\la} node[sloped,fill=white,pos=0.8] {\lb}  (\to);
    \end{tikzpicture}

   \caption{The base case instance}
   \label{fig:2+2-multi-star} 
 \end{figure}

 The next claim follows directly from the definition of a box. The factor $32=2\cdot 4^2$ comes from the definition of the box and simple geometric considerations.
 
 \begin{lemma}\label{obs:tech2dim} If the pair of tasks $\{1,2\}$, is not a box, then the $R_{\{1,2\}}$ region of the slice-mechanism for these two tasks has a bundling boundary. In particular, there exists $\alpha_{12}\geq 32\nu$ such that the bundling boundary is defined by the equation $t_1+t_2 = \alpha_1+\alpha_2-\alpha_{12}$. Similarly for the other pairs $\{1,2'\}$, $\{1',2\}$, and $\{1',2'\}$.
 \end{lemma}

 We continue with two technical lemmas. We provide some intuition
 of the base case after the lemmas. Notice that the first lemma, which is a consequence of Corollary~\ref{cor:bundlingsiblings},
 breaks the symmetry between tasks, i.e., the conclusion does not
 apply if we exchange the role of the two leaves.

 \begin{lemma}\label{prop:tech2dim} Assume that $\alpha_{2'}=\max
   \{\alpha_1, \alpha_{1'}, \alpha_2, \alpha_{2'}\}$. If there exists some
   $\bar t_1\in (\alpha_1-\alpha_{12}, \alpha_1)$ such that the allocation in the
   slice $(2,2')$ is half-bundling at task $2$, then the mechanism has unbounded approximation ratio.
   Analogously, if there exists some $\bar t_{1'}\in
   (\alpha_{1'}-\alpha_{1'2}\,,\, \alpha_{1'})$ such that the allocation in
   the slice $(2\,,\,2')$ is half-bundling at task $2$, then the
   mechanism has unbounded approximation ratio.
\end{lemma}

\begin{proof} See Figure~\ref{fig:LemmaToClaimIV} for an illustration
  of the proof. Assume that for $\bar t_1\in (\alpha_1-\alpha_{12},
  \alpha_1)$ the allocation in the slice $(2,2')$ is half-bundling at
  task $2$. Let the endpoint of the half-bundling boundary be $\bar t = (\bar
  t_1, \bar t_2, t_{2'}=0).$ Since this is the boundary point for task
  $2$ (i.e., $\bar t_2$ is a critical value), it also lies  on the bundling boundary of $R_{\{1,2\}},$ and obviously $\bar
  t_1<\alpha_1$ and $\bar t_2<\alpha_2$ hold.

  We will show that the region $R_{\{1,2,2'\}}$ must have a bundling facet
  of equation $t_1+t_2+t_{2'}=D,$ for $D<\alpha_1+\alpha_2\leq
  \alpha_{2'}+\alpha_2.$ By Corollary~\ref{cor:bundlingsiblings} this
  will imply that the mechanism has infinite approximation.

  By the geometry of the half-bundling boundary, for an input point
  $\bar t' = (\bar t_1, \bar t_2-\epsilon, t_{2'}=2\epsilon)$ the tasks $2$ and
  $2'$ are not allocated to the root player. So $\bar t'$ is outside
  the region $R_{\{1,2,2'\}},$ and violates at least one linear
  constraint defining $R_{\{1,2,2'\}}.$ 

  What can these constraints be? We will eliminate all other
  possibilities, by choosing $\epsilon$ to be sufficiently
  small, concluding that there must exist a bundling facet
  of equation $t_1+t_2+t_{2'}=D$. Firstly, $R_{\{1,2,2'\}}$ may have facets of equations
  $t_1=\alpha_1,\,\, t_2=\alpha_2,\,\, t_{2'}=\alpha_{2'},$ but $\bar t'$ does not violate any of
  these. Further, there might be a facet of the form $t_1+t_2=d$. But since $\bar t$ satisfies this constraint, 
  $d\geq \bar t_1+\bar t_2>\bar t_1+\bar t_2-\epsilon$, so this constraint cannot be violated for $\bar t'$ either. Similarly a facet of equation
  $t_2+t_{2'}=e$ is not violated since $e\geq \alpha_2> \bar t_2+\epsilon$, and finally a facet of
  equation $t_1+t_{2'}=f$ cannot be violated since $f\geq \alpha_1>\bar t_1+2\epsilon,$
  for small enough $\epsilon.$  The only
  remaining possibility is a facet $t_1+t_2+t_{2'}=D,$ so that $D\leq
  \bar t_1+(\bar t_2-\epsilon)+2\epsilon=\bar t_1+\bar
  t_2+\epsilon<\alpha_1+\alpha_2,$ which concludes the proof.  The
  proof of the second statement is analogous with task $1'$ instead of
  task $1.$
\end{proof}

\begin{figure}
  \centering

\begin{tikzpicture}[scale=0.4]

\coordinate (AA) at (10,0,0);
\coordinate (BB) at (0,10,0);
\coordinate (CC) at (0,0,10);

\draw[->] (O) -- (AA);
\draw[->] (O) -- (CC);
\draw[->] (O) -- (BB);

\coordinate (a) at (9,0,0);
\coordinate (b) at (0,9,0);
\coordinate (c) at (0,0,9);

\coordinate (d) at (0,5,0);
\coordinate (e) at (0,5,4);
\coordinate (f) at (4,5,0);

\coordinate (g) at (0,0,7.0);
\coordinate (h) at (0.6,0,7.0);
\coordinate (i) at (0.6,1.4,7.0);
\coordinate (j) at (0,2,7.0);

\coordinate (k) at (5.7,0,0);
\coordinate (l) at (5.7,0,1.9);
\coordinate (m) at (5.7,1.4,1.9);
\coordinate (n) at (5.7,3.3,0);

\draw[blue,fill=blue!30] (O) -- (d) -- (f) -- (n) -- (k) -- cycle;%

\draw[blue,fill=yellow!80] (O) -- (k) -- (l) -- (h) -- (g) -- cycle;%

\draw[blue,fill=red!10] (O) -- (d) -- (e) -- (j) -- (g) -- cycle;%

\draw[black,fill=red!20,opacity=0.6] (d) -- (e) -- (f) -- cycle;%

\draw[black,fill=red!20,opacity=0.6] (g) -- (h) -- (i) -- (j) -- cycle;%

\draw[black,fill=red!20,opacity=0.6] (k) -- (l) -- (m) -- (n) -- cycle;%

\draw[black,fill=red!20,opacity=0.6] (i) -- (h) -- (l) -- (m) -- cycle;%

\draw[black,fill=red!20,opacity=0.4] (e) -- (f) -- (n) -- (m) -- (i) -- (j) -- cycle;%

\draw[blue,dotted,fill=blue!20,opacity=0.4] (a) -- (b) -- (c) -- cycle;%

\draw[red,very thick,dashed] (5,4,0) -- (5,1.4,2.6) -- (5,0,2.6);%

\coordinate (O) at (0,0,0);

\draw[red,thick] (5,4,0) circle (3pt) node[anchor=west] {$\bar t =(\bar t_1,\bar t_2,0)$};

\draw (5.7,0,0) node[anchor=north] {$\alpha_1$};
\draw (9,0,0) node[anchor=north] {$D$};
\draw (0,5,0) node[anchor=east] {$\alpha_2$};
\draw (0,0,7) node[anchor=south] {$\alpha_{2'}$};

\draw (10,0,0) node[anchor=south] {$t_1$};
\draw (0,10,0) node[anchor=west] {$t_2$};
\draw (0,0,10) node[anchor=south] {$t_{2'}$};

\end{tikzpicture}

\caption{Illustration of the proof of Lemma~\ref{prop:tech2dim}. The half-bundling boundary in the slice $(t_2,t_{2'})$ for $t_1=\bar t_1$ is drawn by the dashed lines. It can exist  only if $R_{\{1,2,2'\}}$ has a bundling facet on the plane $t_1+t_2+t_{2'}=D$\,\,(where $D=\bar t_1+\bar t_2$). Given that $D< \alpha_2+\alpha_{2'},$ the allocation of $R_{\{2,2'\}}$ (for $t_1=0$) cannot be crossing, even if arbitrary other defining facets for $R_{\{1,2,2'\}}$ exist. }
\label{fig:LemmaToClaimIV}
\end{figure}
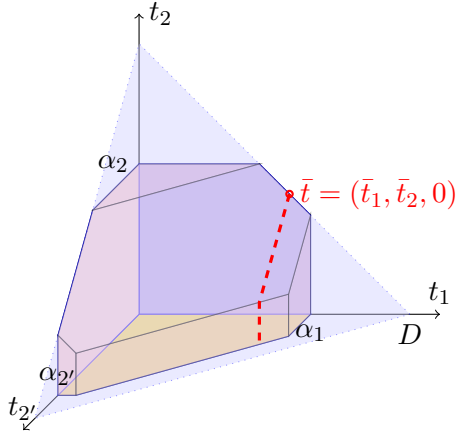

The claim of the next lemma holds under quite general conditions. We will apply it to different triples of the four tasks, and even different instances. 
  
 \begin{lemma} \label{lemma:tech2dim2} 
   If both $R_{\{1,2\}}$ and $R_{\{1,2'\}}$ are quasi-bundling, then the allocation of the slice mechanism $(2, 2')$ is non-crossing when we set $t_1$ to any value in $(\alpha_1-\min\{\alpha_{12}, \alpha_{12'}\}, \alpha_1)$, with the possible exception of a single value in this interval.
   The same holds for task $1'$ instead of 1.
 \end{lemma}

 \begin{proof}
   Towards a contradiction, assume that there are two values $\bar t_1$ and $\bar t_1'$ with $\bar t_1<\bar t_1'$ in this interval, for both of which the allocation in the slice $(2,2')$ is crossing. This means that for $t_1\in \{\bar t_1, \bar t_1'\}$ the slice mechanism for tasks 2 and 2' is task-independent, i.e., the boundary of task 2 is unaffected by changing the value of task 2'. 

   The fact that these two values of $t_1$ are in $(\alpha_1-\min\{\alpha_{12}, \alpha_{12'}\}, \alpha_1) \subseteq (\alpha_1-\alpha_{12}, \alpha_1)$ means that they fall in the bundling boundary between 1 and 2 (Figure~\ref{fig:tech2dim2}). In the interval $[\bar t_1, \bar t_1']$ the bundling boundary is common between tasks 1 and 2. Since the boundary of task 2 is unaffected by changing the value of task 2', the boundary of $t_1$ in $[\bar t_1, \bar t_1']$ also remains unaffected by changing the value of task 2'. In particular, this means that for instances $(t_1=\bar t_1, t_2=0, t_{2'}=x)$, the mechanism allocates both task 1 and 2 to the root, for all $x$.

   But for large $x$, say $x\geq \alpha_{2'}$, the above points are in the region $R_{\emptyset|\{1, 2'\}}$, because $\bar t_1>\alpha_1-\alpha_{12'}$ and the point $(t_1=\bar t_1, t_2=0, t_{2'}=x)$ is above the bundling boundary of 1 and 2'. This contradicts the fact that 1 is allocated to the root. 
 \end{proof}

 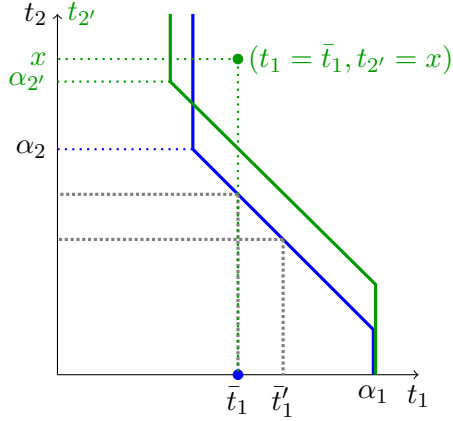
\begin{figure}[h]
  \centering
  \begin{tikzpicture}[scale=0.6]

    \draw[->] (0,0) -- (8,0) node[anchor=north] {$t_1$};
    \draw[->] (0,0) -- (0,8) node[anchor=east] {$t_2$};
    \draw[green!60!black] (0,8) node[anchor=west] {$t_{2'}$};
    \draw[thick, dotted, blue] (0, 5) node[anchor=east,black] {$\alpha_2$} -- (3, 5);
    \draw[very thick, blue ] (7, 0) node[anchor=north,black] {$\alpha_1$}-- (7, 1) -- (3, 5) -- (3, 8);

    \draw[densely dotted, very thick, gray] (4, 4) -- (4,0) node[anchor=north,black] {$\bar t_1$};
    \draw[densely dotted, very thick, gray] (5, 3) -- (5, 0) node[anchor=north,black] {$\bar t_1'$};
    \draw[densely dotted, very thick, gray] (4, 4) -- (0, 4);
    \draw[densely dotted, very thick, gray] (5, 3) -- (0, 3);

    \filldraw[blue] (4,0) circle (3pt);

    \filldraw[green!60!black] (4,7) circle (3pt) node[anchor=west] {$(t_1=\bar t_1, t_{2'}=x)$};

    \draw[green!60!black, very thick] (2.5,8) -- (2.5,6.5) -- (7.05,2) -- (7.05,0);
    \draw[green!60!black, thick, dotted] (2.5, 6.5) -- (0, 6.5) node[anchor=east] {$\alpha_{2'}$};
    \draw[green!60!black, thick, dotted] (4,7) -- (0, 7) node[anchor=east] {$x$};
     \draw[green!60!black, thick, dotted] (4,7) -- (4, 0);
  \end{tikzpicture}
  
  \caption{\small High-level argument of the proof of Lemma~\ref{lemma:tech2dim2}. The blue part depicts the values of tasks 1 and 2 and the green part depicts the values of task 1 and 2'. The blue continuous line shows the boundaries for tasks 1 and 2 when $t_{2'}=0$. Similarly, the green continuous line shows the boundaries for tasks 1 and 2' when $t_2=0$.  
    $t_1=\bar t_1$ and $t_1=\bar t_1'$ are two values on the bundling boundary (i.e., the part with slope $-45^o$) between tasks 1 and 2. The lemma shows that the (2,2') slice mechanism cannot be task independent at both of these values of $t_1$. Indeed 
if tasks 2 and 2' are independent at $t_1\in\{\bar t_1, \bar t_1'\}$, changing $t_{2'}$ does not change the blue line at these values, and therefore the boundary between 1 and 2 remains bundling in the interval $[\bar t_1, \bar t_1']$. This leads to a contradiction, when we consider the allocation of an instance $(t_1=\bar t_1, t_2=0, t_{2'}=x)$, for $x\geq \alpha_{2'}$. This instance, depicted as a blue dot at $(t_1=\bar t_1, t_2=0)$ and as a green dot at $(t_1=\bar t_1, t_{2'}=x)$, is to the west of the blue line (which is the boundary between 1 and 2) and east of the green line (which is the boundary between 1 and 2'); the first says that task 1 is given to the root and the latter that it is not given to the root.}
 \label{fig:tech2dim2} 
\end{figure}

Now we are ready to state the main lemma of the base case. We will show that if none of the four stars from $\{1, 1', 2, 2'\}$ is a box, then the mechanism has high approximation ratio.

Intuitively, we will argue as follows: By reducing $s_2$ to some \emph{positive} value $s_2^*,$ we shift the  2-dimensional bundling boundary facets of both $R_{\{1,2\}}$ and of $R_{\{1',2\}}$ parallel to themselves until both of them reach the $t_2=0$ plane (i.e., bundling boundaries reach the $t_1$-axis, and the $t_{1'}$-axis, respectively, but they do not disappear). Hereby we will use that by the previous lemma, the $(2,2')$ slices are not task-independent, and therefore the boundary positions $\psi_2$ are linear functions of $s_2.$ 

Subsequently, --- applying the lemma in another, orthogonal, direction --- we reduce $s_1,$ in order to shift one of the bundling boundaries (say, of $R_{\{1,2\}}$) in the other direction to reach the $t_2$-axis and further, until the bundling boundary 'disappears'.  This will imply high approximation ratio in the end, because $s_2^*$ remains positive and bounded from below, but the $(t_1,t_2)\approx (0,0)$ input will be allocated to the leaves.

We remark that such a parallel shift of the boundary by linear function of $s_i$ is trivial for an affine minimizer, but far from obvious for tasks $1$ and $2$ that {\em are not siblings}, since the $2\times 2$ characterization does apply to the $(1,2)$-slice mechanism.

We are now ready to prove the base case (Theorem~\ref{thm:basecase}).
\begin{proof}[Proof of Theorem~\ref{thm:basecase}]
  Consider a standard instance $T=(t,\bar s)$, and assume w.l.o.g. that $\alpha_{2'}=\max \{\alpha_1, \alpha_{1'}, \alpha_2, \alpha_{2'}\}$.
  First, we apply Lemma~\ref{lemma:tech2dim2} to the triple $\{1, 2, 2'\}$. According to the lemma, for all but at most one $t_1 \in (\alpha_1-C,\,\alpha_1)$ the corresponding slice mechanism $(2, 2')$ has a non-crossing allocation, where $C=\min\{\alpha_{12}, \alpha_{12'}\}\geq 32\nu$.

  By the continuity requirement, the $(2,2')$-slice mechanism cannot be task independent (for $t_1$ rational) because it has a non-crossing allocation. Also, Theorem~\ref{thm:sibling-independence} establishes that it cannot be a constant mechanism, if the approximation ratio is at most $n$. Thus the $(2,2')$-slice mechanism can only be a relaxed affine minimizer. Lemma~\ref{prop:tech2dim} excludes the possibility that it is half-bundling at task 2, for $t_1 \in (\alpha_1-C,\,\alpha_1)$.
  We conclude that for (all but one) rational $t_1 \in
  (\alpha_1-C,\,\alpha_1)$ the $(2, 2')$ slice mechanism is in the
  linear part of a relaxed affine minimizer, that is, by Observation~\ref{ref:obs-characterization} (iii) it has critical
  values for $t_2$ that are truncated linear functions of $s_2:$
  \begin{align*}
    \psi_2[t_1](s_2)=\max(0, \,\lambda(t_1)\, s_2 - \gamma(t_1)\,).
  \end{align*}
  Recall that this notation omits the parameters that have their
  original values $\bar s$ and $t=0$. In particular, both $\lambda(t_1)$ and
  $\gamma(t_1)$ may depend on values in $s_{-2}$, but this dependency
  will not play any role in the argument.

Analogously, we can apply Proposition~\ref{lemma:tech2dim2} to the triple $\{1', 2, 2'\},$ obtaining that with $C'=\min\{\alpha_{1'2}, \alpha_{1'2'}\}$ for all but at most one, rational $t_{1'}\in (\alpha_{1'}-C',\,\alpha_{1'})$ the $\psi_2[t_{1'}](s_2)$ are truncated linear functions of $s_2.$

The following claim is a variant of (\cite{CKK20}, Lemma 4); we show that for different points on the same bundling boundary, due to linearity, the multiplicative constant of $\psi_2()$ must be the same. 

\begin{figure}
  \centering
  \begin{tikzpicture}[scale=0.35]

    \draw[->] (0,0) -- (8,0) node[anchor=north] {$t_1$};
    \draw[->] (0,0) -- (0,8) node[anchor=east] {$t_2$};

    \draw[very thick, blue] (0, 6) node[anchor=east,black] {$\alpha_2$} -- (4, 6) -- (6, 4) -- (8, 4); \draw[very thick, blue ] (6, 0) node[anchor=north,black] {$\alpha_1$}-- (6, 4) -- (4, 6) -- (4, 8);

\draw (1.5,7) node[anchor=west] {}; \draw (1.7,2) node[anchor=west] {$R_{\{1,2\}}$}; \draw (7.7,2) node[anchor=east] {}; \draw (7.7,7) node[anchor=east] {$R_{\emptyset}$};

    \draw (1,-1) node[anchor=north] {(a)};
  \end{tikzpicture}
\hspace*{1cm}
  \begin{tikzpicture}[scale=0.35]

    \draw[->] (0,0) -- (8,0) node[anchor=north] {$t_1$};
    \draw[->] (0,0) -- (0,8) node[anchor=east] {$t_2$};

    \draw[very thick, blue] (0, 2) node[anchor=east,black] {$\alpha_2$} -- (4, 2) -- (6, 0) ; \draw[very thick, blue ] (6, 0) node[anchor=north,black] {$\alpha_1$} -- (4, 2) -- (4, 8);

\draw (1.5,7) node[anchor=west] {}; \draw (1.7,2) node[anchor=west] {}; \draw (7.7,2) node[anchor=east] {}; \draw (7.7,7) node[anchor=east] {$R_{\emptyset}$};

    \draw (1,-1) node[anchor=north] {(b)};
  \end{tikzpicture}
\hspace*{1cm}
  \begin{tikzpicture}[scale=0.35]

    \draw[->] (0,0) -- (8,0) node[anchor=north] {$t_1$};
    \draw[->] (0,0) -- (0,8) node[anchor=east] {$t_2$};

    \draw[very thick, blue] (1,0) node[anchor=north,black]
    {$\varepsilon$} -- (0,1) node[anchor=east,black]
    {$\varepsilon$};

    \draw (-0.23,0.5) node[anchor=west] {};
    \draw (2,2) node[anchor=west] {$R_{\emptyset}$};

    \draw (1,-1) node[anchor=north] {(c)};
  \end{tikzpicture}
  
  \caption{\small High-level argument of the proof of Theorem~\ref{thm:basecase}. Figure (a) shows the $(1,2)$ slice. Figure (b) shows how the area of the $R_{1,2}$ changes when we reduce the values of $s_2$ to $s_2^*>0$ and Figure (c) when we reduce $s_1$ to $s_1^*$. In Figure (c) none of the tasks is allocated to the root, even for small values of $t_1=t_2=\varepsilon$ which implies an unbounded ratio (see Lemma~\ref{prop:quasi}). Note that this shift of the bundling boundary would be straightforward for an affine-minimizer, because of linearity of the boundaries. But here we need to establish this carefully, by using the sibling tasks $1',2'$. By applying Lemma~\ref{lemma:tech2dim2} first to the triple $\{1, 2, 2'\}$ and then to the triple $\{1', 2, 2'\},$ we establish linearity of the boundary $\psi_2(s_2)$ for an interval of $t_1$ values. Figure (b) illustrates that this allows to shift the bundling boundary by reducing $s_2$ to $s^*_2>0$ reaching a half-bundling position (Claims (i), (ii), (iii)). Then Claim (iv) establishes linearity of $\psi_1(s_1)$, which allows us to reduce $s_1$ to $s_1^*>0$ and obtain an unbounded approximation ratio. }
  \label{fig:high-level-base-case}
\end{figure}
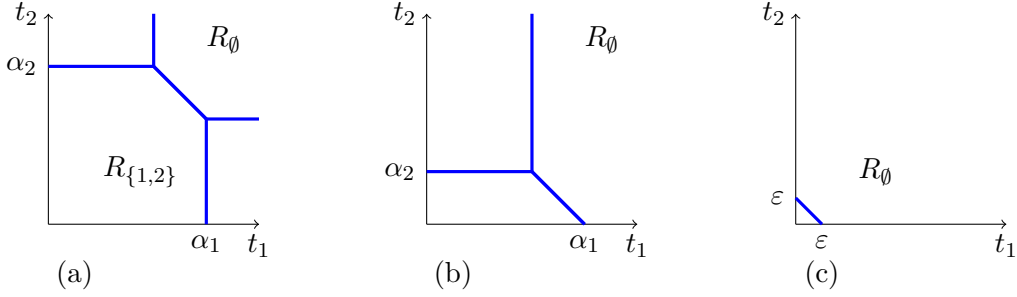

\begin{claim*}[i]\label{cla:magic}
  In the interval $t_1 \in (\alpha_1-C,\,\alpha_1)$, $\lambda(t_1)$ is fixed, i.e., independent of $t_1$. Similarly $\lambda(t_{1'})$ is independent of $t_{1'}\in (\alpha_{1'}-C',\,\alpha_{1'})$.
\end{claim*}

\begin{proof} Fix two values $t_1,\hat t_1 \in (\alpha_1-C,\,\alpha_1)\subset (\alpha_1-\alpha_{12},\,\alpha_1)$ with $t_1 < \hat t_1$. The critical value points for $t_2$ are on a bundling boundary of $R_{\{1,2\}}$, so
  \begin{align*}
    \psi_2[t_1](\bar s_2)-\psi_2[\hat t_1](\bar s_2)=\hat t_1-t_1.
  \end{align*}

  Assume for contradiction that $\lambda(t_1)>\lambda(\hat t_1).$ Then we slightly increase $s_2$ to  $\bar s_2+\epsilon,$ and obtain $\psi_2[t_1](\bar s_2+\epsilon)-\psi_2[\hat t_1](\bar s_2+\epsilon)=\hat t_1-t_1+\epsilon(\lambda(t_1)-\lambda(\hat t_1))>\hat t_1-t_1$, which contradicts the Lipschitz property (Proposition~\ref{prop:lipschitz}).

For the other direction, $\lambda(t_1)<\lambda(\hat t_1),$ we slightly decrease the value to $s_2-\epsilon,$ and again violate the Lipschitz property analogously.
 \end{proof}

 From Claim (i) it follows that $\psi_{2}[ t_1](s_2)-\psi_2[t_1'](s_2)=t_1'-t_1$ for \emph{every} $s_2$ for which these $\psi_2$ values are positive. Thus for every $s_2$ the following set of boundary points of $R_{\{1,2\}},$
 $$\left\{(t_1\,, \,\psi_2[t_1](s_2)) \colon t_1\in (\alpha_1-C,\,\alpha_1)\,\,\land\,\, \psi_{2}[t_1](s_2)>0 \,\right\},$$
 remain part of a bundling boundary. An analogous statement holds for the boundary of $R_{\{1'2\}}$ for $t_{1'}\in (\alpha_1-C',\,\alpha_1]$. Next we give simple bounds for $\lambda$ and $\gamma.$
  
\begin{claim*}[ii]\label{cla:lambdaGamma} Let $t_1\in (\alpha_1-C,\,\alpha_1],$ and   \begin{align*}\psi_2[t_1](s_2)=\max(0\,,\,\lambda(t_1)\cdot s_2 - \gamma( t_1)\,);
      \end{align*} 
  Then   $0< \gamma(t_1)<\lambda(t_1)<2n$ or the approximation ratio is at least $n.$ 
  Similarly for every $ t_1'\in (\alpha_1-C',\,\alpha_1].$  
\end{claim*}

\begin{proof} Fix a $t_1\in (\alpha_1-C,\,\alpha_1],$ and let $\lambda=\lambda(t_1)$ and $\gamma=\gamma(t_1).$

  First we prove $0< \gamma.$ For every $s_2,$ for which
  $R_{\{1,2\}}$ (or $R_{\{1',2\}}$) has a bundling boundary, it
  obviously holds that $\psi_2(s_2)>0$ (Lemma~\ref{obs:Rnonempty}). On
  the other hand, by Lemma~\ref{obs:alphas}, $\psi_2(s_2)\rightarrow
  0$ as $s_2\rightarrow 0.$ Therefore, for $s_2=0,$ no piece of
  bundling boundary can remain for $R_{\{1,2\}}.$ Thus for $t_1$ it
  holds that $\psi_2[t_1](0)=0,$ and therefore $\lambda\cdot
  0-\gamma\leq 0,$ which implies $0\leq \gamma.$ Finally, since by the same argument $\gamma(t_1-\epsilon)\geq 0,$ and due to the bundling boundary $\gamma(t_1)-\epsilon=\gamma(t_1-\epsilon)\geq 0,$ we even obtain strict inequality $\gamma(t_1)>0$ for every $t_1\in (\alpha_1-C,\,\alpha_1].$
 
  Second, we prove $\gamma < \lambda.$ By $R_{\{1,2\}}$ being
  quasi-bundling, and by the position of $t_1,$ we obviously have
  $\psi_2[t_1](\bar s_2)>0$. Since $\bar s_2\leq 1$, we get $\gamma<
  \lambda$. Then, $\lambda < 2n $ follows from the fact that
  $\psi_2[t_1](s_2)<ns_2$, for all values of $s_2$
  (Lemma~\ref{obs:alphas}).
\end{proof}

Let $\hat s_2=\gamma(\alpha_1)/\lambda(\alpha_1)>0.$ Note that $\hat s_2$ is the largest $s_2$ value such that
$\psi_2[t_1=\alpha_1](s_2)=0$ (that is, for $\hat s_2$ the
bundling boundary of $R_{\{1,2\}}$ touches the $t_1$ axis). %
Analogously, define $\hat s'_{2}>0$
to be the largest $s_2$ value such that
$\psi_2[t_{1'}=\alpha_{1'}](s_{2})=0.$ Assume w.l.o.g that $\hat
s_2\leq \hat s'_{2},$ that is
  $$\,\psi_2[t_1=\alpha_1](\hat s_2)=0 \quad\Rightarrow\quad \psi_2[t_{1'}=\alpha_{1'}](\hat s_2)=0.$$
 We would like to reduce $s_2=\bar s_2$ to some appropriate $s_2=s_2^*>0,$ where $\psi_2[t_1=\alpha_1](s_2^*)=0$ and $\psi_2[t_{1'}=\alpha_{1'}](s_2^*)=0,$ but for $(t_1, t_{1'})=(0,0)$ it still holds that $\psi_2[t_1=0,t_{1'}=0](s_2^*)>0.$
 The next claim provides such an $s_2^*:$
 
\begin{claim*}[iii]
  There exists $q^*\in \{0,1,\ldots,4n/\nu\}$ such that for $s_2^*:=\bar s_{2}- q^*  \nu/(4n)$ the following hold:
  \begin{enumerate}
  \item[(a)]  $\psi_2[t_1=\alpha_1](s_2^*)\,\,=\,\,\psi_2[t_{1'}=\alpha_{1'}](s_2^*)\,\,=\,\,0;$ 
  \item[(b)]  $\psi_2[t_1=0,t_{1'}=0](s_2^*)>30\nu$
  \item[(c)]  $s_2^*>30\nu/n.$
  \end{enumerate}
\end{claim*}

\begin{proof} For any $t_1\in [\alpha_1-C,\,\alpha_1],$ consider the sequence of values $\psi_2[t_1](\bar s_{2}-q \frac{\nu}{4n})$, for $q=0,1,2, \ldots .$ By Claim (ii), $\lambda(t_1)\leq 2n$, so successive values in this sequence differ by at most $\lambda (\nu/4n)\leq 2n\nu/4n=\nu/2.$ 

As above, let $\hat s_2$ be the largest value such that $\psi_2[t_1=\alpha_1](\hat s_2)=0.$  Then for this $\hat s_2,$ over the whole interval $t_1\in [\alpha_1-C,\,\alpha_1]$ there is still a bundling boundary of $R_{\{1,2\}}.$ It follows that $$\psi_2[t_1=\alpha_1-C](\hat s_2)-0=\psi_2[t_1=\alpha_1-C](\hat s_2)-\psi_2[t_1=\alpha_1](\hat s_2)=\alpha_1-(\alpha_1-C)=C\geq 32\nu.$$ Now we further reduce $\hat s_2$ to the next possible value of the form $\bar s_{2}- q  \nu/(4n),$ and define this to be $s_2^*.$ We have $\psi_2[t_1=0,t_{1'}=0](s_2^*)\geq \psi_2[t_1=\alpha_1-C,t_{1'}=0](s^*_2)\geq 32\nu-\nu/2>30\nu,$ so (b) holds. Also, (a) holds by the assumption that $\,\psi_2[t_1=\alpha_1](s_2)=0 \quad\Rightarrow\quad \psi_2[t_{1'}=\alpha_{1'}](s_2)=0.$ Finally, (c) holds, since by Observation~\ref{obs:alphas} $s_2^*\cdot n>\psi_2(s_2^*)=\psi_2[t_1=0,t_{1'}=0](s_2^*)> 30\nu.$ 
\end{proof}

Let $T^*$ denote the instance that we get by changing  $\,\bar s_2\,$ to $\,s_2^*\,$ in $\,T=(t,\bar s).$ We claim that in $T^*$ the region $R^{T^*}_{\{1,2\}}$ is half-bundling at $1,$ and the region $R^{T^*}_{\{1',2\}}$ is half-bundling at $1'$. This follows immediately from Claim (iii) (a) and (b). Indeed, if $\psi_2[t_1=0](s_2^*)>\psi_2[t_1=t_1'](s_2^*)=0$ for some $t_1'>0,$ then geometry implies that $R_{\{1,2\}}$ can only be half-bundling (and similarly for $R_{\{1',2\}}$).

Observe, that $T^*$ is a standard instance for leaf set $\{1\}$. We apply Lemma~\ref{prop:quasi} to $(1, 1'),$ and obtain that $R_{\{1,1'\}}$ (for $t_2=0$) must be crossing, otherwise the approximation ratio is $\infty.$

Let $c=\psi_2[t_1=0,t_{1'}=0](s_2^*)>30\nu$ (see also Figure~\ref{fig:claimIV}). 
For the critical value of $t_1,$ we will use the notation  $\alpha_1^*=\psi_1[T^*](\bar s_1).$ Analogously, let $\alpha_{1'}^*=\psi_{1'}[T^*](\bar s_{1'}).$ By $R_{\{1',2\}}$ being half-bundling at $1'$, we have $\alpha_1^*\geq c$ and $\alpha_{1'}^*\geq c.$

Consider now the $(1,1')$ slice allocation for $T^*$ for different positive rational values of $t_2\in (0,c).$ What kind of allocations can these be? By Proposition~\ref{lemma:tech2dim2} applied to $\{2, 1, 1'\}$, for almost all (in fact, for all) $t_2\in (0,c),$ it cannot be a crossing allocation. %
The next claim excludes that it is half-bundling for any $t_2\in (0,c).$ The only remaining possibility is, that for every rational $t_2\in (0,c),$ it is (the linear part of) a relaxed affine minimizer. This will imply that for such a $t_2,$ the critical value fuction $\psi_1[t_2](s_1)$ is truncated linear in the variable $s_1.$

\begin{figure}
  \centering

\begin{tikzpicture}[scale=1.2]
\coordinate (O) at (0,0,0);
\coordinate (A) at (0,1,0);
\coordinate (B) at (0,1,1);
\coordinate (C) at (0,0,2);
\coordinate (D) at (3,0,0);
\coordinate (E) at (2,1,0);
\coordinate (F) at (1,1,1);
\coordinate (G) at (1,0,2);

\draw[blue,fill=yellow!80] (O) -- (C) -- (G) -- (D) -- cycle;%
\draw[blue,fill=blue!30] (O) -- (A) -- (E) -- (D) -- cycle;%
\draw[blue,fill=red!10] (O) -- (A) -- (B) -- (C) -- cycle;%
\draw[blue,fill=red!20,opacity=0.8] (D) -- (E) -- (F)  -- (G) -- cycle;%
\draw[blue,fill=red!20,opacity=0.6] (C) -- (B) -- (F) -- (G) -- cycle;%
\draw[blue,fill=red!20,opacity=0.8] (A) -- (B) -- (F) -- (E) -- cycle;%

\coordinate (ED) at (2.3,0.7,0);
\coordinate (FG) at (1,0.7, 1.3);
\coordinate (BC) at (0,0.7,1.3);

\coordinate (AA) at (4.5,0,0);
\coordinate (BB) at (0,1.5,0);
\coordinate (CC) at (0,0,3);

\draw[-, thick,red,dashed] (ED) -- (FG);
\draw[-, thick,red,dashed] (BC) -- (FG);
\draw[->] (O) -- (AA);
\draw[->] (O) -- (CC);
\draw[->] (O) -- (BB);

\draw[red,thick] (2.3,0.7,0) circle (1pt) node[anchor=west] {$(\hat t_1,\hat t_2,0)$};

\draw (3,0,0) node[anchor=north] {$\alpha_1^*$};
\draw (0,1,0) node[anchor=east] {$c$};
\draw (0,0,2.0) node[anchor=east] {$\alpha_{1'}^*$};

\draw (4.5,0,0) node[anchor=north] {$t_1$};
\draw (0,1.5,0) node[anchor=west] {$t_2$};
\draw (0,0,3) node[anchor=north] {$t_{1'}$};

\end{tikzpicture}

\caption{Illustration to the proof of Claim iv. The dashed line is a half-bundling boundary in the slice $(t_1,t_{1'}).$ It implies the existence of a bundling facet of $R_{\{1,1',2\}}.$ Such a facet excludes that $R_{\{1,1'\}}$ (for $t_2=0$) is crossing, even if other facets of $R_{\{1,1',2\}}$ exist.}
\label{fig:claimIV}
\end{figure}
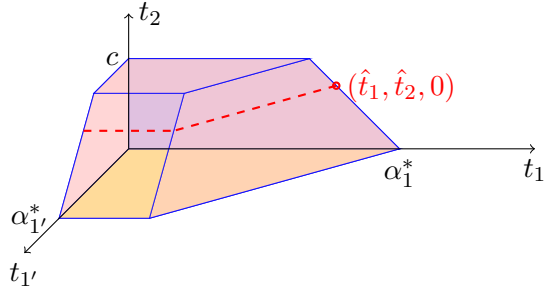

\begin{claim*}[iv] If the allocation of the slice $(1, 1')$ is half bundling for some $t_2\in (0,c),$ then the approximation ratio of the mechanism is unbounded.
\end{claim*}

\begin{proof}  The proof is similar to that of Lemma~\ref{prop:tech2dim}.
Assume that for some $\hat t_2\in(0,c)$ the $(1,1')$ allocation is half-bundling at task $1.$  Let $(\hat t_1, t_{1'}=0, \hat t_2)$ denote the endpoint of the half-bundling boundary (see Figure~\ref{fig:claimIV}). By definition, this point is on the half-bundling boundary of the slice $(1,1')$ (for $t_2=\hat t_2$ fixed). On the other hand, (since the critical value $\hat t_1$ for task 1 is unique) it is also on the half-bundling boundary of the slice $(1,2)$ (for $t_{1'}=0$ fixed), because $R_{\{1,2\}}$ is half-bundling at task 1 and $\hat t_2<c.$ 
From the latter it follows that $\hat t_1+\hat t_2= \alpha_1^*.$

We show that the region $R_{\{1,1',2\}}$ has a bundling facet of equation $t_1+t_{1'}+t_{2}=\alpha_1^*<\alpha_1^*+\alpha_{1'}^*.$  By Corollary~\ref{cor:bundlingsiblings} this will imply that the mechanism has unbounded approximation ratio.

For some small enough $\epsilon>0,$ the point $(\hat t_1-\epsilon, t_{1'}=2\epsilon, \hat t_2)$ is not in $R_{\{1,1', 2\}},$ since (at least) the tasks $1$ and $1'$ are not allocated to the root  by the definition of the half-bundling boundary of the slice $(1,1').$ So this point violates at least one linear constraint defining  $R_{\{1,1',2\}}.$
 
$R_{\{1,1',2\}}$ may have facets of equations $t_1=\alpha_1^*,\,$ $t_{1'}=\alpha_{1'}^*\,$ and $t_2=c,$ respectively. 
The point $(\hat t_1-\epsilon, t_{1'}=2\epsilon, \hat t_2)$ is within these boundaries for $2\epsilon<\alpha_{1'}^*.$ Further, there might be boundaries of equations $t_1+t_2=d$ (for some $d\geq \alpha_1^*$), or $t_1+t_{1'}=e$ (for $e\geq \alpha_1^*$), or $t_{1'}+t_2=f$ (for $f\geq c $). However,  for small $\epsilon$ these would not exclude the point $(\hat t_1-\epsilon, t_{1'}=2\epsilon, \hat t_2)$ from $R_{\{1,1',2\}}$ (as follows from $\hat t_1+\hat t_2=\alpha_1^*\leq d,$ and $\hat t_1<\alpha_1^*\leq e,$ and $\hat t_2< c\leq f,$ respectively).

The only remaining possibility is a facet  $t_1+t_{1'}+t_2=D$ such that $\hat t_1+\epsilon+\hat t_2\geq D\geq \alpha_1^*=\hat t_1+\hat t_2. $ Taking $\epsilon\rightarrow 0,$ we obtain $D=\alpha_1^*,$ which concludes the proof.
\end{proof}

We therefore conclude that for every (rational) $t_2\in (0,c)$ the $(1,1')$ slice is (in the linear part of) an affine minimizer and so the critical value function (at $t_{1'}=0$) for $t_1$ is linear in $s_1,$ that is $\psi_1[t_2](s_1)=\max(0\,,\,\lambda(t_2)\cdot s_1 - \gamma( t_2)\,).$ 
 With role change between tasks $1$ and $2,$ with precisely the same argument as in Claim (i), we obtain that $\lambda(t_2)$ is equal for all rational $t_2\in (0,c),$ and therefore, for all $t_2\in (0,c).$ 
 
 We concentrate only on the $(1, 2)$ allocation now, while every other $t$-value is $0.$ Let $\psi_1(s_1)=\psi_1[t_2=0](s_1).$ We reduce $s_1$ so that $s_1\rightarrow 0.$  By Observation~\ref{obs:alphas} (ii) applied to $T^*,$ we have $\psi_1(s_1)\rightarrow 0$ (that is, $\gamma(t_2=0)=0$). Because of linearity of $\psi_1$ with unique $\lambda$ for every $t_2\in (0,c),$ the $R_{\{1,2\}}$ has a half-bundling boundary, as long as $\psi_1(s_1)>0,$ that is, as long as $s_1>0.$ This implies that for every $s_1>0,$ the point $(t_1,t_2)=(\psi_1(s_1), \psi_1(s_1))$ is in the region $R_{\emptyset|\{1,2\}},$ i.e., tasks $1$ and $2$ are allocated to the leaves.
 As $s_1$ tends to $0,$ the point $(\psi_1(s_1), \psi_1(s_1))$ tends to $(0,0).$  Clearly, $\opt \rightarrow 0,$ and $\mech\geq s_2^*>30 \nu.$ Taking $s_1\rightarrow 0,$ we obtain unbounded approximation ratio.
\end{proof}      

Theorem~\ref{thm:basecase} established that every pair of leaves with two tasks per leaf contains a box of size two. This immediately suggests that there are many boxes of size two. The following statement makes it precise. 

\begin{theorem}\label{cor:b2} Let $T$ be a standard instance. If the approximation ratio is at most $n$, a random star of two leaves is not a box with probability at most $2/\sqrt{\ell}$. 
\end{theorem}
\begin{proof}
  The proof is an immediate consequence of the following extremal graph theory result~\cite{Ram90}: Every $\ell\times \ell$ bipartite graph with at least $2\ell\sqrt{\ell}$ edges, contains a cycle $C_4$.

  Take an $\ell\times \ell$ bipartite graph with edges indicating that the associated star is not a box. Theorem~\ref{thm:basecase} guarantees that there exists no $C_4$, which implies that there are at most $2\ell\sqrt{\ell}$ stars that are not boxes.
\end{proof}

\subsection{Induction step}
\label{sec:induction_step}

For the induction step, we consider a star of $k\geq 3$ tasks, which
we call \emph{\slicedoff\ box}, such that all its subsets of $k-1$
tasks are boxes. The precise structure of a \slicedoff\ box is
detailed in the following
definition. %

\begin{definition}[\Slicedoff\ box] \label{def:potentially-good} Fix a
  mechanism. A star $P=\{p_1,\ldots,p_k\}$ from a set of leaves $\cal
  C,$ for $|\mathcal C|=k\geq 3$ is called a \emph{\slicedoff\ box for
    an instance $T$} if $T=(t,\bar s)$ is standard for $P$ and the
  following conditions hold
  \begin{itemize}
  \item for every $i\in [k]$, $P_{-i}=P\setminus \{p_i\}$ is a box for $T$;
  \item for every $q= 1,\ldots,4n/\nu,$ such that $\bar s_{p_k}>q \nu/(4n),$ the instance that results from $T$ when we replace the values of task $p_k$ with $[t_{p_k}=0, \, \bar s_{p_k}-q \nu/(4n)]$, the  $P_{-k}$ is a box. 
  \end{itemize}
  If a \slicedoff\ box is not a box itself, it will be called \slicedoffx\ box.
\end{definition}

The definition of a box is about the allocation area $R_P$, in which
all tasks are allocated to the root. Recall that a box is a set of
tasks for which $R_P$ is (almost) an orthotope (a hyperrectangle). On
the other hand, the first condition in the definition of a \slicedoff\
box essentially says that it is a box from which a simplex was cut off
by a diagonal cut (i.e., by a hyperplane of the form $\sum_{i \in [k]}
t_{p_i} = c_{[k]}$). While the definition allows the removed simplex
to be empty (in which case it is simply a box), when we want to emphasize
that it is not empty, we call it a \slicedoffx\ box (see
Figure~\ref{fig:box} (c) and (e)).

In the rest of this section we establish that the probability that a
\slicedoff\ box is not a box, is small. We consider a standard
instance $T=(t,s)$ and a star $P=\{p_1,\ldots, p_k\}$ from a set of
leaves $\mathcal C=(Q_1,\ldots,Q_k)$, and we assume that $P$ is a
\slicedoffx\ box. Let $p_k'$ be a random sibling of $p_k$, i.e.,
another task of leaf $Q_k.$ The main result of this section
establishes that {\em almost all other sets $(P_{-k},p_k')$ are
  boxes}.  In particular, we will show via a probabilistic argument
that $(P_{-k},p_k')$ is a box with probability at least
$1-2n^3/\ell\xi$ (Lemma~\ref{lemma:many-good-sets}).  In order to show
this result, we consider the $(p_k,p_k')$-slice mechanism which is a
mechanism between two players and two tasks, utilizing the $2\times 2$
characterization (Theorem~\ref{theo:addchar} and
Observation~\ref{obs:possibleFigures}).

We will focus on input points that lie on a specific area around the
allocation region $R_P$, for which we can establish useful properties
when we use the $2\times 2$ characterization. In what follows, two important input
points are $T_\nu(P)$ and $T_{\nu'}(P)$. Recall that
by Definition~\ref{def:nu}, $T_\nu(P)$ is an instance $(t^\nu, s)$
that agrees with $T$ everywhere, except for tasks in $P,$ where
$t^{\nu}_{p_i}=\alpha_{p_i}-4^{k}\nu.$ We define
$T_{\nu'}(P)=(t^{\nu'}, s)$ similarly, but for $\nu'=\nu/4$. Then
$t^{\nu'}_P$ is a point between $t^\nu_P$ and $(\alpha_{p_i})_{p_i\in
  P}.$ Moreover, $T_\nu(P_{-k})$ is the projection in direction of the
$p_k$-axis of $T_{\nu'}(P),$ i.e., all tasks $p_i$ in $P\setminus p_k.$
 have $t$-values equal to $\alpha_{p_i}-4^{k-1}\nu$. We will
use the short notation $\,T_\nu=T_\nu(P)$ and $T_{\nu'}=T_{\nu'}(P).$

We proceed with some useful properties of \slicedoff\ boxes
(Section~\ref{sec:ind-general-observations}), and then we examine the
different possibilities for a $(p_k,p_k')$-slice mechanism
(Section~\ref{sec:ind-p_k-p'k-slice}). Then we conclude with the main
result of this section, that shows that the probability that a \slicedoff\ box is not a
box, is small (Section~\ref{sec:ind-main-lemma}).

\subsubsection{General observations.}
\label{sec:ind-general-observations}
We will need a few simple observations about weakly monotone mechanisms. 
First of all, the following proposition and its corollaries provide some intuition about \slicedoff\ box sets.

\begin{proposition}\label{prop:tech1} Suppose that
  $P=\{p_1,p_2,\ldots,p_k\}$ is a \slicedoffx\ box for $T=(t,\bar s).$
  Consider an arbitrary instance $T'=[(t'_P, t_{[m]\setminus P}) ,\bar
  s]$ such that $t^\nu_{P}< t'_P < t^{\nu'}_P$ coordinate-wise -- and
  so, $T_\nu\leq T'\leq T_{\nu'}.$ Then the point $t'_P$ obeys every
  linear constraint defining the region $R^T_P,$ except for a single
  violated non-redundant constraint $\sum_{i=1}^k t_{p_i}=c$ for some
  $c\leq \sum_{i=1}^k \alpha_{p_i}-k\cdot 4^k\cdot\nu,$ which defines
  a bundling boundary facet of $R_P.$
\end{proposition}

\begin{proof} 
  For input $t^\nu_P$ at least one task $p_i\in P$ is not allocated to
  the root, otherwise $P$ would be a box. In other words,
  $(\alpha_{p_i}-4^k\nu)_{i=1}^k =t^{\nu }_P\not\in R_P,$, and
  therefore a linear constraint of the form $\sum_{i\in I} t_{p_i}
  \leq c_I$ for $R_P$ is either violated or satisfied with equality by
  $t^{\nu }_P$. Clearly, then all points $t'_P>t^\nu_P$ violate this
  constraint as well.

  We need to show that $I=[k].$ Assume for contradiction that $I\neq
  [k],$ and let $j\in[k]\setminus I.$ Then in the $(k-1)$-dimensional
  space of the tasks $P_{-j},$ the point
  $(t^{\nu}_{p_i})_{i\in [k]\setminus \{j\}} $ violates (at
  least) the same constraint $\sum_{i\in I} t_{p_i}\leq c_I$. But then the
  higher point $(t^{\nu'}_{p_i})_{i\in [k]\setminus \{j\}}=(\alpha_{p_i}-4^{(k-1)}\nu)_{i\in [k]\setminus \{j\}}
  $ strictly violates the same constraint contradicting the fact, that
  $P_{-j}$ is a box. By the same argument, no point $t'_P\leq
  t_P^{\nu'}$ can violate any other linear constraint of $R_P,$ but
  this single constraint with $I=[k].$

\begin{remark} We used the fact that a restriction of a constraint for
  $R_P$ to the (hyper-)plane $t_{p_j}=0$ is a constraint for
  $R_{P\setminus \{p_j\}}$ for the same instance $T.$ One short
  explanation for this is that the constant terms in the linear
  constraints for allocation regions correspond to (differences of)
  payments which are unique for a given mechanism and given $\bar s.$
\end{remark}
\end{proof}

\begin{corollary}\label{prop:tech2} Let the conditions be like in Proposition~\ref{prop:tech1}. Then
\begin{itemize}
\item[(i)] the region $R_P$ has a non-redundant bundling facet of equation $\sum_{i=1}^k t_{p_i}=c$ for some $c;$ 

\item[(ii)] for $T'$ every task $p_i\in P$ is allocated to the leaf;

\item[(iii)] for each  $j\in [k],$ if in $T'$ we change the $t$-value of only task $p_j$ to $0,$ then all tasks $i\in P$ are allocated to the root.

\item[(iv)] in $T'$ for the task $p_k$ and its  critical value $\psi_{p_k}=\psi_{p_k}(t'_{-p_k},\bar s),$ it holds that $\psi_{p_k}>0,$ and the point $(t'_{P\setminus\{p_k\}}, \psi_{p_k})\in \mathbb R^k$ is a point of the bundling boundary facet of $R_P$. 

\end{itemize}
\end{corollary}

\begin{proof} (i) follows trivially by
  Proposition~\ref{prop:tech1}.

(ii) For $T_\nu(P)=(t^\nu,\bar s)$ it holds that $t^{\nu}_P\not\in
   R_P$, since $P$ is not a box. 
 Therefore, by
   Proposition~\ref{prop:tech1}, it must be either a point on the
   bundling facet between $R_P$ and $R_{\emptyset|P},$ or a point in
   $R_{\emptyset|P},$ since it cannot violate any other type of linear
   constraint (i.e., with $I\neq [k]$) for $R_P.$ By weak
   monotonicity, $t'_P>t^\nu_P$ implies that $t'_P\in R_{\emptyset|P}$
   for $T'.$

(iii) By definition, since $P$ is a  {\slicedoffx } box, $P_{-j}$ must be a box, and therefore for input $(t_{p_j}=0, t^{\nu'}_{-p_j})$  all tasks  $i\in P\setminus \{p_j\}$ are allocated to the root.
Further, by weak monotonicity the same holds for every $t'_P\leq t^{\nu'}_P.$

Next we argue that task $p_j$ is also allocated to the root, when we
change $t'_{p_j}$ to $0.$ The only constraint that $t^{\nu'}_P$
violates, is defined by the bundling facet between $R_P$ and
$R_{\emptyset|P}.$ Thus, $t^{\nu'}_P$ obeys, and $t'_P$ strictly obeys
every other constraint defining $R_P.$ Obviously, this holds also
after the change to $t'_{p_j}=0.$ Since, by
Proposition~\ref{prop:tech1}, no other type of constraint is violated
(i.e., with $I\neq [k]$), after this change to $t'_{p_j}=0,$ the point
can be either in $R_P,$ or in $R_{\emptyset|P}.$ However, the fact
that the tasks in $P\setminus \{p_j\}$ are allocated to the root,
implies that the point cannot be in $R_{\emptyset|P}$, and
therefore it must be in $R_P$.

(iv) $\psi_{p_k}>0$ follows by (iii). The boundary point
$(t'_{P\setminus\{p_k\}}, \psi_{p_k})\in \mathbb R^k$ lies on at least
one facet defined by a linear constraint for a set $I$ of tasks such
that $p_k\in I.$ By Proposition~\ref{prop:tech1}, this can only be the
constraint with $I=[k].$
\end{proof}

The next propositions are about points on the bundling boundary facet
of $R_P.$ We consider inputs $T_1=[(t^1_P, t_{[m]\setminus P}) ,\bar
s]$ and $T_2=[(t^2_P, t_{[m]\setminus P}) ,\bar s]$ that differ from
$T_{\nu'}$ and $T_{\nu}$ only in their coordinates $t_{p_i}$ for tasks
$p_i\in P,$ such that $t^\nu_{p_i}\leq t^1_{p_i}< t^2_{p_i}\leq
t^{\nu'}_{p_i}$. All other $t$-coordinates $j\in [m]\setminus P$ stay
the same, i.e.,  $t^\nu_{j}= t^1_{j}= t^2_{j}= t^{\nu'}_{j},$ as well
as the $s$-coordinates.

Due to the relative position of $T_1$ and $T_2,$ and since the
critical value points for task $p_k$ are on a bundling boundary of
$R_P,$ (Corollary~\ref{prop:tech2} (iv)), it is intuitively clear that
the Lipschitz property (Proposition~\ref{prop:lipschitz}) for any two
such points is fulfilled with equality.

\begin{proposition}[\cite{CKK21b} Lemma 21. Claim (ii)]\label{prop:psiAreFar} Let $T_\nu \leq T_1\leq T_2 \leq T_{\nu'}$  coordinate-wise, and $t^1_{p_i}<t^2_{p_i}$ hold with strict inequality for every $p_i\in P.$ Then  $$\psi_{p_k}(t_{-p_k}^1,\bar s)-\psi_{p_k}(t_{-p_k}^2,\bar s)=|t_{-p_k}^1-t_{-p_k}^2|_1=\sum_{i\in [k-1]} (t_{p_i}^2-t_{p_i}^1).$$
  \end{proposition}

  \begin{proof}  
    Fix $\bar s.$ We know by Corollary~\ref{prop:tech2} (ii) that in
    the allocation of instance $T_1$ all tasks in $P$ are given to the
    leaves. We create another instance $\hat T_1=[(\hat t^1_P,
    t_{[m]\setminus P}) ,\bar s]$ which is close to $T_1$ that has the
    same allocation for the tasks in $P$: we  set $\hat
    t_{p_i}^1=t_{p_i}^1+\epsilon/n$, for $i\in[k-1]$ and we set $\hat
    t_{p_k}^1=\psi_{p_k}(t_{-p_k}^1)+\epsilon$, for some small
    $\epsilon>0$. The allocation of $P$ remains the same, because (1)
    the boundary $\psi_{p_k}(t_{-p_k}^1)$ changes by less than
    $(n-1)\epsilon/n<\epsilon$ (due to the Lipschitz property), so the new
    $t$-value of $p_k$, $\hat t_{p_k}^1$ is still above the boundary, and (2) by weak
    monotonicity, the allocation of the other tasks remains the
    same.

    Similarly, we create an instance $\hat T_2$ that is very close to
    $T_2$: we set $\hat t_{p_i}^2= t_{p_i}^2-\epsilon/n$, for $i\in
    [k-1]$ and $\hat t_{p_k}^2=\psi_{p_k}(t_{-p_k}^2)-\epsilon$. We
    claim that in $\hat T_2$ all tasks are allocated to the root. To
    see this, observe that for $(t_{-p_k}^2,t_{p_k}=0)$ all tasks are
    allocated to the root (by Corollary~\ref{prop:tech2}), and by the
    Lipschitz property $\psi_{p_k}(t_{-p_k}^2)$ can decrease by less
    than $\epsilon$, so $p_k$ will be again allocated to the root. Again
    the allocation will not change.

  By weak monotonicity we get
  \begin{align*}
    0\geq \sum_{i\in [k]} (1-0)(\hat t_{p_i}^2-\hat t_{p_i}^1)\geq   \sum_{i\in [k-1]} (t_{p_i}^2-t_{p_i}^1) +  \psi_{p_k}( t_{-p_k}^2)-\psi_{p_k}( t_{-p_k}^1)-\left(2+\frac{k-1}{n}\right)\epsilon.\\    
  \end{align*}  

On the other hand, the Lipschitz property implies 
\begin{align*}
\psi_{p_k}( t_{-p_k}^2)-\psi_{p_k}( t_{-p_k}^1)\leq \sum_{i\in [k-1]} (t_{p_i}^2-t_{p_i}^1),
  \end{align*}  
since $T_2\geq T_1$, hence all terms inside the sum on the right hand side are nonnegative. By letting $\epsilon$ tend to 0 we get the equality of the claim.

\end{proof}

The (proof of the) previous proposition shows that
$\psi_{p_k}(t^1)-\psi_{p_k}(t^2)=|t^1-t^2|_1,$ essentially uses the
fact that for $i=1,2$ the points $(t^i_{-p_k},\psi_{p_k}(t^i))\in
\mathbb R^k$ are on the bundling boundary of $R_P$ and
$R_{\emptyset|P}.$ The next proposition and its corollary show that
the reverse also holds: If
$\psi_{p_k}(t_{-p_k}^1)-\psi_{p_k}(t_{-p_k}^2)=|t_{-p_k}^2-t_{-p_k}^1|_1,$
then for all points $t'$ with $t^1\leq t'\leq t^2$ their critical
inputs for task $p_k$ are on a bundling boundary of $R_P$.

\begin{proposition}\label{prop:FarPsiAreBundling} Let $s$ be fixed, let
  $t^1_{p_i}< t^2_{p_i}$ for every $p_i\in P,$ and $t^1_{j}=t^2_j$ for
  $j\in [m]\setminus P.$ If
  $\psi_{p_k}(t_{-p_k}^1)-\psi_{p_k}(t_{-p_k}^2)=\sum_{i\in [k-1]}
  (t_{p_i}^2-t_{p_i}^1),$ then
\begin{enumerate}
\item[(i)]  for  $(t^1_{-p_k},\psi_{p_k}(t^1)-\epsilon)$ all tasks $p_i\in P$ are given to the root for $0<\epsilon< \psi_{p_k}(t^1),$ and 
\item[(ii)] for  $(t^2_{-p_k},\psi_{p_k}(t^2)+\epsilon)$ all tasks $p_i\in P$ are given to the leaves  for $\epsilon>0.$
\end{enumerate}
    
  \end{proposition}

\begin{proof} The proof is almost the same as of Proposition~\ref{prop:lipschitz}. We can ignore the coordinates $j\in [m]\setminus P,$ since here the values of these tasks stay fixed.

 Consider the instances $(t^1_{-p_k},\psi_{p_k}(t^1_{-p_k})-\epsilon)$ and $(t^2_{-p_k},\psi_{p_k}(t^2_{-p_k})+\epsilon)$. By definition of $\psi_{p_k},$ task $p_k$ is allocated to the root in the first instance, but not in the second one.  Let $a^1$ and $a^2$ denote the allocation for these two instances. By weak monotonicity, we have
\begin{align*}
  \sum_{i\neq k} (a^1_{p_i}-a^2_{p_i}) (t^1_{p_i}-t^2_{p_i})+(1-0)((\psi_{p_k}(t^1)-\epsilon)-(\psi_{p_k}(t^2)+\epsilon)) \leq 0 \\
  \psi_{p_k}(t^1)-\psi_{p_k}(t^2) \leq \sum_{i\neq k} (a^1_{p_i}-a^2_{p_i}) (t^2_{p_i}-t^1_{p_i}) + 2\epsilon.%
\end{align*}
Since $\psi_{p_k}(t_{-p_k}^1)-\psi_{p_k}(t_{-p_k}^2)=\sum_{i\in [k-1]} (t_{p_i}^2-t_{p_i}^1),$ by letting $\epsilon$ tend to 0, we get $a^1_{p_i}=1$ and $a^2_{p_i}=0$ for all $i\in[k-1]$ as well. 
\end{proof}

\begin{corollary}\label{cor:FarPsiAreBundling} Assume that $t^1_{p_i}< t^2_{p_i}$ for every $p_i\in P,$ and $t^1_{j}=t^2_j$ for $j\in [m]\setminus P,$ and $\psi_{p_k}(t_{-p_k}^1)-\psi_{p_k}(t_{-p_k}^2)=\sum_{i\in [k-1]} (t_{p_i}^2-t_{p_i}^1).$ 
Let $t^1\leq t'\leq t^2$  be so that $t^1_{p_i}< t'_{p_i}< t^2_{p_i}\, (i\in [k]).$ Then
\begin{enumerate}
\item[(i)] for  $(t'_{-p_k},\psi_{p_k}(t')-\epsilon)$ all tasks $p_i\in P$ are given to the root, and 
\item[(ii)] for  $(t'_{-p_k},\psi_{p_k}(t')+\epsilon)$ all tasks $p_i\in P$ are given to the leaves. 
\end{enumerate}

\end{corollary}

\begin{proof} 
  Let $\psi^1=\psi_{p_k}(t^1),\,\psi^2=\psi_{p_k}(t^2),$ and
  $\psi'=\psi_{p_k}(t').$
  Clearly, $$\psi^1-\psi^2=(\psi^1-\psi')+(\psi'-\psi^2)\leq
  \sum_{i\neq k}(t'_{p_i}-t^1_{p_i})+\sum_{i\neq
    k}(t^2_{p_i}-t'_{p_i})=\sum_{i\neq
    k}(t^2_{p_i}-t^1_{p_i})=\psi^1-\psi^2.$$ The inequality follows
  from the Lipschitz property, and the last equality by the
  assumption. We conclude that equality must hold everywhere, in
  particular $\psi^1-\psi'= \sum_{i\neq k}(t'_{p_i}-t^1_{p_i})$ and
  $(\psi'-\psi^2)= \sum_{i\neq k}(t^2_{p_i}-t'_{p_i}).$ So, $t'$ can
  play the role of $t^1$ in Proposition~\ref{prop:FarPsiAreBundling}, implying (i).
  Symmetrically, $t'$ can take the role of $t^2,$ that implies (ii).
\end{proof}
  
\subsubsection{The $p_k$-$p_k'$ slice}
\label{sec:ind-p_k-p'k-slice}

We now consider a \slicedoffx\ box $P=\{p_1,\ldots, p_k\}$ and choose
a random sibling $p_k'$ of $p_k$.  We will consider all different
possibilities for the ($p_k, p_k'$)-slice mechanism based on the
$2\times 2$ characterization. First, we consider the case that the
mechanism is task-independent, establishing that then $(P_{-k}, p'_k)$
must be a box (Lemma~\ref{lem:crossing}). Then we treat the non
task-independent case, excluding almost surely affine minimizers
(Lemma~\ref{lemma:linearBoundary}) and bounding from above the number of
$p_k'$ siblings for which the mechanism can be relaxed-affine
minimizer (Lemma~\ref{lem:noHalfBundling}).

\paragraph{The case of task-independent (crossing) allocations.}

In this paragraph we treat the case when in at least two points $T_1<T_2$ between $T_\nu$ and $T_{\nu'}$ the allocation of  the $(p_k, p_k')$-slice is crossing. We prove that in this case the set $(P_{-k}, p_k')$ is a box.

\begin{lemma}\label{lem:crossing} Suppose that $P=\{p_1,p_2,\ldots,p_k\}$ is a \slicedoffx\ box for $T=(t,\bar s)$. Fix a sibling $p'_k$ of $p_k.$ Assume that there exist two distinct instances $T_1,\,T_2,$ such that 
\begin{itemize}

\item $T_\nu \leq T_1\leq T_2 \leq T_{\nu'}$  coordinate-wise, and $t^\nu_{p_i}<t^1_{p_i}<t^2_{p_i}< t^{\nu'}_{p_i}$ holds with strict inequality for every $p_i\in P;$
 \item for both $T_1$ and $T_2$, in the $(p_k, p_k')$ slice mechanism  the allocation of the root is crossing.

\end{itemize}
Assuming that the approximation ratio is less than $n,$ then $(P_{-k}, p'_k)$ is a box for $T.$

\end{lemma}
\begin{proof}
  The proof is analogous to the proof of (\cite{CKK21b} Lemma 27), and
  a generalization of
  Proposition~\ref{lemma:tech2dim2}. %
  We use the notation $t^1=t^1_P$ and $t^2=t^2_P$ for the respective
  $k$-dimensional $t$-vectors over the task set $P.$ All other $t$ values, as well as $\bar s$ are fixed  for all
  considered inputs, except for $t_{p_k'},$ that we will change and
  treat explicitly. We also use the notation $\alpha'=\alpha_{p'_k}.$

  In both $T_1$ and $T_2$ let us now increase the value of task $p'_k$
  from $t_{p'_k}=0$ to $\bar t_{p'_k}=\alpha'-\epsilon$ for some small
  $\epsilon < 4^k\nu,$ to obtain the new instances $\bar T_1$ and $\bar T_2,$
  respectively. Since by assumption of the lemma, in both $T_1$ and
  $T_2$ the $(p_k, p'_k)$ mechanism has crossing (task independent)
  allocation, this change of the $t$-values does not affect the respective critical values for
  task $p_k,$ i.e., $\psi_{p_k}(\bar T_1)=\psi_{p_k}(T_1)$ and
  $\psi_{p_k}(\bar T_2)=\psi_{p_k}(T_2).$

  By Corollary~\ref{prop:tech2} $(iv)$ for $T_1$ and $T_2$ the points
  $(t^1_{-p_k},\psi_{p_k}(t^1))$ and $(t^2_{-p_k},\psi_{p_k}(t^2))$
  are on a bundling boundary facet of $R_P;$ moreover by
  Proposition~\ref{prop:psiAreFar}
  $(\psi_{p_k}(T_1))-\psi_{p_k}(T_2))= |t^2_{-p_k}-t^1_{-p_k}|_1$
  holds. Since for $\bar T_1$ and $\bar T_2$ the critical values do
  not change, $(\psi_{p_k}(\bar T_1))-\psi_{p_k}(\bar T_2))=
  |t^2_{-p_k}- t^1_{-p_k}|_1$ remains true for $\bar
  t_{p'_k}=\alpha'-\epsilon.$

We apply now Proposition~\ref{prop:FarPsiAreBundling} (i) for $\bar T_1\leq \bar T_2$ (for the region $R^{\bar T_1}_P$).
Define  the instance $\hat T$ by reducing $t_{p_k}$ to zero in $\bar T_1,$ i.e., let $\hat t:=(t^1_{-p_k}, t_{p_k}=0)$ and $\hat t_{p_k'}=\alpha'-\epsilon.$ 
By Proposition~\ref{prop:FarPsiAreBundling} (i), we obtain that for input $\hat T$ all tasks in $P$ are given to the root. 

We show that also $p_k'$ is given to the root for $\hat T.$ Assume for
contradiction that this is not the case. Then we increase $\hat
t_{p'_k}$ to $t^*_{p'_k}=\alpha'-\epsilon/2,$ and decrease all other
values from $\hat t_{p_i}$ to $t^*_{p_i}=0$ and obtain $t^*.$ Omitting the
coordinate $t_{p_k}$ (which is $0$), by weak monotonicity we get

\begin{align*}0 &\geq \sum_{i\in [k-1]} (\hat a_{p_i}-a_{p_i}^*) (\hat t_{p_i}^{ }-t_{p_i}^*)+ (\hat a_{p_k'}-a_{p_k'}^*) (\hat t_{p_k'}^{ }-t_{p_k'}^*)\\
 & = \sum_{i\in [k-1]} (1-a_{p_i}^*) (\hat t_{p_i}^{ })+ (0-1)
  (-\epsilon/2).\end{align*}

Here the $t$-values follow from the
definition of $t^*;$ $\hat a_{p_i}=1$ by the previous paragraph; $\hat
a_{p_k'}=0$ by the indirect assumption, and $a^*_{p_k'}=1$ by the
definition of $\alpha'=\alpha_{p_k'},$ since $t^*_{p_k'}<\alpha'$ and
the other $t$-values are $0$ in $t^*.$ Since $\hat t_{p_i}>0$ this is
a contradiction.  All in all, this proves that for input $\hat T$ all
tasks in $P_{-k}\cup \{p'_k\}$ are allocated to the root.

Finally, observe that by definition $\hat t_{p_i}=t^1_{p_i}>\alpha_{p_i}-4^k\nu, \,\, (i\in [k-1]),$ and $\hat t_{p_k'}=\alpha_{p'_{k}}-\epsilon> \alpha_{p'_k}-4^k\nu.$ This implies by weak monotonicity that also for $T_\nu(P_{-k}\cup \{p'_k\})$   the task set $P_{-k}\cup \{p'_k\}$ is given to the root, and therefore $(P_{-k}, p'_k)$ is a box for $T.$   \end{proof}

\paragraph{The case of non-crossing allocations.}

In the remaining part, we need to exclude the case that between
$T_\nu$ and $T_{\nu'}$ there are `many' points where the allocation of
the slice mechanism $(p_k, p_k')$ is non-crossing.  Allocations of
$2\times 2$-mechanisms that are non-crossing for the root, occur (1)
in the linear part (linear as function of $s_{p_k}$) of relaxed affine
minimizers\footnote{Constant mechanisms are special affine minimizers;
  or we can exclude them immediately using
  Lemma~\ref{prop:noconstant}}, or (2) they must be half-bundling at
$p_k$ (including completely bundling)\footnote{The case of being at some
  singular point $(s_{p_k}, s_{p_k'})$ of a relaxed task independent
  mechanism is excluded since the inputs $T$ that we consider satisfy
  the continuity requirement almost surely.}.

The following paragraphs deal with both of these cases; for any given
point $T,$ case (1) -- linear boundary functions -- will be excluded
to occur in $k$ different points, positioned appropriately, hence
linear functions are excluded almost surely; case (2) is excluded with
high probability over the choice of $p_k'$.

\paragraph{Case 1: Linear boundary functions.}

Here we exclude the case that in $k$ carefully selected points $T_0,
\, T_1, \ldots , T_{k-1}$ the $(p_k, p_k')$-slice mechanism is a
(relaxed) affine minimizer, and so that in each of these points the
critical value function $\psi_k(s_{p_k})$ is truncated linear. This
implies that linear boundary functions are excluded almost surely. The
next lemma is a variant of Lemma 21 in \cite{CKK21b}.

\begin{lemma} \label{lemma:linearBoundary} Assume that the
  approximation ratio is at most $n.$ Suppose that
  $P=(p_1,\ldots,p_k)$ is a \slicedoffx\ box for standard instance
  $T=(t,\bar s).$ Fix a sibling $p_k'$ of $p_k.$ Finally let
  $h=(5/8)\cdot 4^{k}\nu,$ and let $0<\eta< \nu/4n$ be some small
  number. There exist no distinct instances $T_0=(t^0,\bar s)$,
  $T_1=(t^1,\bar s),\ldots, T_{k-1}=(t^{k-1},\bar s)$ such that
  \begin{itemize}
  \item $T_\nu\leq T_j\leq T_{\nu'}$ for every $j = 0, \ldots, k-1;$ 
  \item $t^0_{p_j}+h<t^j_{p_j}<t^0_{p_j}+h+\eta$ for all $j\in [k-1];$
  \item $t^0_{p_i}+\eta/2<t^j_{p_i}<t^0_{p_i}+\eta$ for all $j\in [k-1],$ and all $i\neq j;$
  \item the boundary function $\psi_{p_k}$ is truncated linear in $s_{p_k},$ that is,
    \begin{align*}
      \psi_{p_k}(t^j_{-p_k},\bar s)=\max(0\,,\, \lambda(t^j_{-p_k},\bar s_{-p_k}) s_{p_k} - \gamma(t^j_{-p_k}, \bar s_{-p_k})),
    \end{align*}
    for  each of the instances $T_j,\quad$ $j = 0, 1,\ldots, k-1.$ 
  \end{itemize}
\end{lemma}

\begin{proof}
Towards a contradiction, suppose that such instances $T_0, T_1,\ldots , T_{k-1}$ exist. For simplicity of notation, throughout the proof we write $\psi_k^j (s_{p_k})$ for the critical value function $\psi[T_j]_{p_k}(s_{p_k}),$ defined for the fixed values $t_{-p_k}^j$ and  $\bar s_{-p_k}.$ 

The high-level idea is the following. Let us ignore all other
dimensions, except for those $k$ coordinates of tasks $P.$ Due to
Corollary~\ref{prop:tech2}((ii),(iii)and (iv) respectively), we know
that in the allocation of $T_j,$ $j=0,\,1,\ldots, k-1$, all tasks
are given to the leaf, when we change the $t_{p_k}$-value to $0,$ all
tasks $p_i\in P$ are given to the root, and the respective boundary
points $(t^j_{-p_k},\psi_k^j(\bar s_{p_k}))\in \mathbb R^k$ are on a
bundling boundary facet of $R_P.$ Now, as we reduce $s_{p_k},$ due to
the linearity of each boundary function $\psi^j_k(s_{p_k}),$ these $k$
boundary points move by the same speed $\lambda$ towards the
coordinate-plane $t_{p_k}=0,$ so that all the time they remain on a
bundling boundary.  For some positive $s^*_{p_k}$ this bundling
boundary facet will reach the plane $t_{p_k}=0,$\footnote{Since we
  discretize (in order to have finitely many reduction steps for
  $s_{p_k}$), we have to 'stop' the considered boundary points
  slightly above the plane $t_{p_k}=0.$} and there it will `form' a
bundling boundary of dimension $k-1$ for the task-set $P_{-k}.$ But
this will prove that for this $s^*_{p_k}$ the set $P_{k}$ cannot be
a box, contradicting the fact that $P$ is \slicedoff\ box.

We start with some observations about the critical values $\psi_k^j
(\bar s_{p_k}),$ and the respective boundary points of $R_P$ when
$s_{p_k}=\bar s_{p_k}$.  Note that for every $j\in [k-1]$ by
definition $ t^0_{p_i}< t^j_{p_i}$ holds for all $p_i\in P.$
Therefore, by Proposition~\ref{prop:psiAreFar} $$\psi^0_k (\bar
s_{p_k})-\psi_k^j (\bar s_{p_k})=|t_{-p_k}^0-t_{-p_k}^j|_1=\sum_{i\in
  [k-1]}(t^j_{p_i}-t^0_{p_i}).$$ Moreover, for every $j\in [k-1]$ this
difference roughly equals $h,$ in particular
$$\, h<\sum_{i\in [k-1]}(t^j_{p_i}-t^0_{p_i})<h+n\cdot \eta.$$ 

Therefore, all these critical values are almost the same, that is, for
any pair of indices $j,j'\in[k-1]$
it holds $$|\psi_k^j (\bar s_{p_k})- \psi_k^{j'} (\bar s_{p_k})|<
n\cdot\eta.$$ Next we show that these relative positions remain the
same {\em for every $s_{p_k}.$} The following is a generalization of Claim
(i) in the proof of Theorem~\ref{thm:basecase}.

\begin{claim*}[i]

The boundary functions for $t^j$, $j=0,\,1,\ldots , k-1$,
  \begin{align*}
      \psi_k^j (s_{p_k}) & =\lambda(t_{-p_k}^j) s_{p_k} - \gamma(t_{-p_k}^j)
  \end{align*}
  have the same linear coefficient, i.e., $\lambda(t_{-p_k}^0)=\lambda(t_{-p_k}^1)= \ldots =\lambda(t_{-p_k}^{k-1})=\lambda$. 
\end{claim*}
\begin{proof} We prove that all linear coefficients are equal to $\lambda(t_{-p_k}^0).$ 
  Suppose not, and assume that $\lambda(t_{-p_k}^0)>\lambda(t_{-p_k}^j)$ for some $j\in [k-1].$ For $s_{p_k}=\bar s_{p_k}+\epsilon$, for some $\epsilon>0$, we get
  \begin{align*}
    \psi_k^0 (\bar s_{p_k}+\epsilon)-\psi_k^j (\bar s_{p_k}+\epsilon)&=\psi_k^0 (\bar s_{p_k})-\psi_k^j (\bar s_{p_k})+\left(\lambda(t_{-p_k}^0)-\lambda(t_{-p_k}^j)\right) \epsilon >|t_{-p_k}^0-t_{-p_k}^j|_1,
  \end{align*}
  which violates the Lipschitz property (Proposition~\ref{prop:lipschitz}). The case of $\lambda(t_{-p_k}^0)<\lambda(t_{-p_k}^j)$ is similar, only that we now consider $s_{p_k}=\bar s_{p_k}-\epsilon.$
\end{proof}
As a corollary we obtain that the difference between (positive) critical values $\psi_k^0$ and $\psi_k^j$ remains the same \emph{for every} $s_{p_k}$.

  \begin{claim*}[ii] Let $j\in [k-1].$ \emph{For every} $s_{p_k}$, for which $\psi_k^j ( s_{p_k})>0,$ the following hold:  
\begin{enumerate}
\item[(a)]  $ \psi_k^0 (s_{p_k})-\psi_k^j (s_{p_k})=|t_{-p_k}^0-t_{-p_k}^j|_1;$
  
\item[(b)]  $|\psi_k^j (s_{p_k})- \psi_k^{j'} ( s_{p_k})|\leq n\cdot \eta$ for every $j'\in [k-1];$ 
\item[(c)] the point $(t^j_{-p_k}, \psi_k^j(s_{p_k}))\in \mathbb R^k$ is on a bundling boundary, in particular on a boundary of $R^{T'}_P,$ where $T'$ is obtained from $T$ by replacing $\bar s_{p_k}$ by $s_{p_k},$ for $j=0, 1, \ldots, k-1;$
\end{enumerate}
\end{claim*}

\begin{proof} (a) and (b) follow from Claim (i); (c)  follows from (a) and from Proposition~\ref{prop:FarPsiAreBundling} and Corollary~\ref{cor:FarPsiAreBundling}.
\end{proof}

Let $\psi^1_k(s_{p_k}) = \max(0, \lambda\, s_{p_k} - \gamma).$ Next we prove lower and upper bounds for $\gamma$ and  $\lambda.$

\begin{claim*}[iii] Let $\psi^1_k(s_{p_k})=\max(0, \lambda\, s_{p_k} -
  \gamma)$ for some $\lambda$ and $\gamma$ that do not depend on
  $t_{p_k}$
  and $s_{p_k}$. Suppose further that the approximation
  ratio is less than $n,$ then 
  \begin{enumerate}
\item[(a)] $0\leq  \gamma < \lambda; $
\item[(b)] $\frac{1}{2n}<\lambda<2n.$ 
\end{enumerate}
\end{claim*}

\begin{proof} (a) By Corollary~\ref{prop:tech2} (iv) holds that $\psi^1_k (\bar s_{p_k})>0$. Now notice that $\lambda\cdot 1-\gamma=\psi^1_k(1)\geq \psi^1_k(\bar s_{p_k}) > 0.$ This implies $\gamma< \lambda$. 

  Next we show $\gamma\geq 0.$ We claim that $\psi_k^1(s_{p_k}=0)=0.$
  Assume the contrary, that $\psi_k^1(s_{p_k}=0)>0.$ Then, by Claim
  (ii) (c), the $(t_{-p_k}^1\,, \, \psi_k^1(0))$ is a point with
  positive coordinates on a bundling boundary of $R_P^{T'}$ where $T'$
  is obtained from $T$ by setting $s_{p_k}=0.$ On the other hand, then
  by Observation~\ref{obs:Rnonempty} $s_{p_k}>0$ must hold,
  contradiction. So we have $\psi_k^1(s_{p_k}=0)=0.$ Consequently,
  $\lambda\cdot 0 - \gamma\leq 0,$ so $0\leq \gamma.$

(b) We prove $\lambda<2n$ first. 
Assume that $\lambda\geq 2n.$ We set $s_{p_k}=n+1,\, t_{-p_k}=t^1_{-p_k}$ and $t_{p_k}=\psi_k^1(s_{p_k})-\epsilon=\lambda s_{p_k}-\gamma-\epsilon.$ Then task $p_k$ is allocated to the root, and the makespan is $\mech\geq \lambda s_{p_k}-\gamma-\epsilon \geq \lambda s_{p_k}-\lambda-\epsilon\geq \lambda (s_{p_k}-1)-\epsilon\geq 2n\cdot n-\epsilon.$

The contribution of the tasks $[m]\setminus P$ to the optimal makespan is $0,$ since the other leaves are trivial. We have $\opt\leq s_{p_k}+\sum_{i\neq k}\bar s_{p_i}\leq n+1+n-1=2n.$ We obtain $\mech/\opt\rightarrow n$ when $\epsilon \rightarrow 0,$ which concludes the proof of the upper bound.

In order to show the lower bound for $\lambda,$ we  assume for contradiction $\lambda\leq  1/2n,$ and set $s_{p_k}=2n^2,$ and $t_{p_k}=\lambda s_{p_k}-\gamma+\epsilon,$ so that the leaf gets task $p_k.$

Then $\mech\geq 2 n^2,$ and because the contribution of the other tasks is at most $n-1,$ using $\gamma\geq 0,$ we  obtain $\opt\leq t_{p_k}+n-1\leq (1/2n)\cdot 2n^2-\gamma+\epsilon+n-1\leq 2n.$ The approximation factor is at least $n,$ a contradiction.
\end{proof}

\begin{claim*}[iv]
  There exists $q^*\in \{0,1,\ldots,4n/\nu\}$ such that $\bar s_{p_k}\geq q^*  \nu/(4n),$ and $\psi^j_k(\bar s_{p_k}-q^*  \nu/(4n))$ is in the interval $(0, \nu)$, for every $j\in [k-1].$ 
\end{claim*}
\begin{proof} By Corollary~\ref{prop:tech2} (iv) it holds for every
  $j\in[k]$ that $\psi^j_k(\bar s_{p_k})>0.$ If also $\psi^j_k(\bar
  s_{p_k})<\nu$ holds for every $j\in[k-1],$ then the claim holds
  trivially with $q^*=0.$
  
  Otherwise let $\bar q\leq 4n/\nu$ be the largest integer such that
  $\bar s_{p_k}\geq \bar q \nu/(4n)$ (recall that $\bar s_{p_k}\leq
  1$). Consider the sequence of values $\psi^1_k(\bar s_{p_k}-q
  \frac{\nu}{4n}),$ for $q=0,1,\ldots,\bar q.$ By Claim (iii),
  $\lambda< 2n$, so successive values in this sequence differ by at
  most $\lambda (\nu/4n)\leq 2n\nu/4n=\nu/2.$ Moreover, because
  $\gamma\geq 0,$ we have $\psi^1_k(0)=0.$ Therefore for some $0\leq
  q^*\leq \bar q,$ the $\psi^1_k(\bar s_{p_k}-q^*\cdot\frac{\nu}{4n})$
  falls between $\nu/4$ and $3\nu/4.$ Finally, by Claim (ii) (b)
  it follows that $\psi^j_k(\bar s_{p_k}-q^*\cdot\frac{\nu}{4n})\in (0,
  \nu)$ for \emph{every} $j\in[k-1],$ given that $n\eta< \nu/4.$
\end{proof}

Let $s^*_{p_k}:=\bar s_{p_k}-q^* \nu/(4n)$ be the value from the
previous claim for which $\psi^j_k(s^*_{p_k}) \in (0, \nu)$, for every
$j\in[n-1].$ For the rest of the proof we fix $(s^*_{p_k}, \bar
s_{-p_k}).$ We know from Claim (ii) (c) that for every $j\in [k-1],$
in an arbitrarily small neighborhood of the point $(t^j_{-p_k},
\psi^j_k(s^*_{p_k}))\in \mathbb R^k$ there are points $t_P$ so that
all tasks $p_i\in P$ are allocated to the root, and also there are
points so that none of the tasks $p_i\in P$ is allocated to the root.

Recall that  $t_{p_i}=0$ for every $p_i\in P$ in the instance $T=(t,\bar s).$ Let $T^*=(t, (s^*_{p_k}, \bar s_{-p_k}))$ denote the instance that we obtain from $T$ by replacing $\bar s_{p_k}$ by $s^*_{p_k}.$  Notice that $T^*$ is standard instance for the star $P_{-k}=P\setminus {p_k}.$ We need to be careful, because all we know about the allocation of instance $T^*$ are the $k$ shifted points $(t^j_{-p_k},
\psi^j_k(s^*_{p_k}))$ of a bundling boundary (in particular, the $\alpha_{p_j}$ values may have changed).

Here is how we finish the proof. For every $j\in [k-1]$ we will define
two nearby instances to $(t^j_{-p_k}, \psi^j_k(s^*_{p_k})),$ called
$\underline T_j$ and $\overline T_j$ so that for each of them the
$t_{p_k}$ entry is zero, the $s_{p_k}$ entry equals $s^*_{p_k},$ and
for $\underline T_j$ all tasks $p_i\in P_{-k}$ are given to the root,
and for $\overline T_j$ all tasks $p_i\in P_{-k}$ are given to the
leaves. This will prove in the end, that for $t_{p_k}=0$ and
$s_{p_k}=s^*_{p_k}$ the task set $P_{-k}$ is not a box, and thus
contradict the assumption that $P$ is a \slicedoff\ box (the second
bullet of Definition~\ref{def:potentially-good}).

 We will use the notation $\alpha^*_{p_j}%
 $ for the (new) $\alpha$ values of the tasks in the instance $T^*$
 (see Definition~\ref{def:alphai}).  In the next claim we define and
 use the $\underline T_j$ instances to lower bound
the $\alpha^*_{p_j}$ values for tasks in $P_{-k}$ and for the instance
 $T^*.$
 
\begin{claim*}[v]
  $\alpha^*_{p_j}> t^0_{p_j}+h$  for every $j\in [k-1].$ 
\end{claim*}
\begin{proof} Fix an arbitrary $j\in [k-1].$ For simplicity, we ignore
  the coordinates $[m]\setminus P.$ Define $\underline T_j=(\underline
  t^{j}, (s^*_{p_k}, \bar s_{-p_k})),$ so that $ \underline
  t^{j}_{p_k}=0,\quad$ and $\underline
  t^{j}_{p_i}=t^j_{p_i}-\epsilon,\,$ for some small $\epsilon$ for
  every $p_i\in P_{-p_k}.$ Since $\underline t^{j}_{p_i}<t^j_{p_i}$
  for each $p_i\in P_{-k},$ and $\underline t^{j}_{p_k}=0<
  \psi^j_k(s^*_{p_k}),$ every $p_i\in P$ must be allocated to the root
  for $\underline T_j,$ (by Claim (ii) (c), and weak
  monotonicity). This implies $\underline t^{j}_{-p_k}\in
  R^{T^*}_{P_{-k}},$ and thus $\alpha^*_{p_j}\geq \underline  t^{j}_{p_j}.$ Given that for $\epsilon$ small enough $\underline  t^{j}_{p_j}=t^j_{p_j}-\epsilon> t^0_{p_j}+h,$ we obtain $\alpha^*_{p_j}\geq \underline  t^{j}_{p_j}=t^j_{p_j}-\epsilon> t^0_{p_j}+h.$   
\end{proof}

In the last claim we define the $\overline T_j$ instances and use them to prove that for $T^*_\nu(P_{-k})$ the jobs of $P_{-k}$ are all allocated to the leaves (see Definition~\ref{def:nu} for $T^*_\nu$). This contradicts the assumption that for $s^*_{p_k}$ the set $P_{-k}$ is a box, and thus also that $P$ is a \slicedoff\ box, which will conclude the proof.

\begin{claim*}[vi]
For the instance $T^*_\nu(P_{-k})$ none of the tasks in  $P_{-k}$ is allocated to the root.
\end{claim*}
\begin{proof} We ignore the coordinates $[m]\setminus P,$ and
  $(s^*_{p_k}, \bar s_{-p_k}))$ is fixed.  Take first an arbitrary
  $j\in[k-1].$ Let $t'=(t^j_{-p_k}, \psi_k^j(s_{p_k}^*))_{P},$ denote
  the point on the bundling boundary (by Claim (ii) (c)) and $a'$ the
  respective allocation for $t'.$ Assume w.l.o.g. that $a'$ assigns
  all tasks in $P$ to the leaves (otherwise we slightly increase all
  $t$-values in $P$ by an infinitesimally small amount).

We define $\overline T_j=(\overline t^{j}, (s^*_{p_k}, \bar
s_{-p_k}))\,$ as follows: let $ \overline t^{j }_{p_k}=0>
\psi^j_k(s^*_{p_k})-\nu,$ and $\overline t^{j
}_{p_i}=t^j_{p_i}+\nu,\,$ for each $p_i\in P_{-p_k}.$ We show first
that for $\overline t^{j }$ every $p_i\in P_{-k}$ is allocated to the
leaves, i.e., $\overline t^{j }\in R^{T^*}_{\emptyset|P_{-k}}.$ Let
$\overline t=\overline t^{j }_P$ and $\overline a$ be the allocation
for $\overline T_j.$ Then by weak monotonicity for the root

\begin{align*}0\geq \sum_{i\in [k]} (\overline a_i-a_i') (\overline t_i^{ }-t_i')
= \sum_{i\in [k]} (\overline a_i-0) (\overline t_i^{ }-t_i')=\sum_{i\in [k-1]} \overline a_i \cdot \nu+\overline a_k (0-t_k').\end{align*} 

The first equality holds because $a_i'=0$ by assumption that all tasks
go to the leaves in $t';$ the second equality holds by the definition
of $\overline t_i=\overline t^{j }_{p_i},$ and $t'_{i}=t^j_{p_i}.$

Since $t'_k=\psi^j_k(s^*_{p_k})<\nu$ it must hold

\begin{align*}
\sum_{i\in [k-1]} \overline a_i \cdot \nu < \overline a_k\cdot \nu.
\end{align*}

But since  $\overline a_i\in\{0,1\},$ the above implies $\overline a_i=0$ for every $i\in[k-1],$ concluding that
$\overline t^{j }\in R^{T^*}_{\emptyset|P_{-k}}.$

Now let $t^c$ be the mean of the points $\{\overline t^{1 }, \overline t^{2 }, \ldots ,\overline  t^{(k-1) }\}.$ That is, $$t^c_{p_i}=\frac{\sum_{j=1}^{k-1} \overline t^{j }_{p_i}}{k-1}\qquad \forall p_i\in P.$$ Then  $t^c_{p_k}=0,$ and for all $p_i\in P_{-k}$
\begin{equation}\label{eq:five} t^c_{p_i}<t^0_{p_i}+h/(k-1)+\eta+\nu<t^0_{p_i}+h/2+2\nu.\end{equation} 
Moreover, since $\overline t^{j }\in R^{T^*}_{\emptyset|P_{-k}}$ for every $j,$ the same holds for their convex combination $t^c.$ We use the notation $T^*_\nu(P_{-k})=(t^{*\nu}, (s_{p_k}^*, \bar s_{-p_k}));$ by Definition~\ref{def:nu} $t^{*\nu}_{p_i}=\alpha^*_{p_i}-4^{k-1}\nu$ for $p_i\in P_{-k}.$ 

We conclude the proof by showing that $t^c_{p_i}< t^{*\nu}_{p_i}$ for every $p_i\in P_{-k}.$
This will prove (by weak monotonicity) that also $t^{*\nu}\in
R^{T^*}_{\emptyset|P_{-k}},$ and therefore $P_{-k}$ is not a box for
$T^*,$ a contradiction.  We have 
\begin{align*}t^{*\nu}_{p_i}=\alpha^*_{p_i}-4^{k-1}\nu>t^0_{p_i}+h-4^{k-1}\nu>
  t^0_{p_i}+h/2+2\nu>t^c_{p_i}.\end{align*} The first inequality holds
by Claim (v); the second inequality holds if $h=(5/8)4^{k}\nu$ and
$k\geq 3;$ the last inequality holds by inequality~\ref{eq:five}. This
concludes the proof.
\end{proof}
\end{proof}

\begin{figure}
  \centering
  
\newcommand{\Depth}{2}
\newcommand{\Height}{2}
\newcommand{\Width}{2}
\newcommand{\Fraction}{0.5}

\begin{tikzpicture}[scale=1.5]
\coordinate (O) at (0,0,0);
\coordinate (A) at (0,\Width,0);
\coordinate (B) at (0,\Width,\Height);
\coordinate (C) at (0,0,\Height);
\coordinate (D) at (\Depth,0,0);
\coordinate (E) at (\Depth,\Width,0);
\coordinate (F) at (\Depth,\Width,\Height);
\coordinate (G) at (\Depth,0,\Height);

\coordinate (BF) at (\Fraction*\Depth,\Width,\Height);
\coordinate (EF) at (\Depth,\Width,\Fraction*\Height);
\coordinate (GF) at (\Depth,\Fraction*\Width,\Height);

\draw[blue,fill=yellow!80] (O) -- (C) -- (G) -- (D) -- cycle;%
\draw[blue,fill=blue!30] (O) -- (A) -- (E) -- (D) -- cycle;%
\draw[blue,fill=red!10] (O) -- (A) -- (B) -- (C) -- cycle;%
\draw[blue,fill=red!20,opacity=0.8] (D) -- (E) -- (EF) -- (GF) -- (G) -- cycle;%
\draw[blue,fill=red!20,opacity=0.6] (C) -- (B) -- (BF) -- (GF) -- (G) -- cycle;%
\draw[blue,fill=red!20,opacity=0.8] (A) -- (B) -- (BF) -- (EF) -- (E) -- cycle;%

\draw[blue,fill=blue!40,opacity=1] (BF) -- (EF) -- (GF);

\coordinate (AA) at (0,1.4*\Width,0);
\coordinate (CC) at (0,0,1.4*\Height);
\coordinate (DD) at (1.4*\Depth,0,0);

\draw[->] (O) -- (AA);
\draw[->] (O) -- (CC);
\draw[->] (O) -- (DD);

\draw (2,0,0) node[anchor=north west] {$\alpha_1$};
\draw (0,2,0) node[anchor=south west] {$\alpha_3$};
\draw (0,0,2) node[anchor=north west] {$\alpha_2$};

\end{tikzpicture}
\begin{tikzpicture}[scale=1.5]
\coordinate (O) at (0,0,0);
\coordinate (A) at (0,\Width,0);
\coordinate (B) at (0,\Width,\Height);
\coordinate (C) at (0,0,\Height);
\coordinate (D) at (\Depth,0,0);
\coordinate (E) at (\Depth,\Width,0);
\coordinate (F) at (\Depth,\Width,\Height);
\coordinate (G) at (\Depth,0,\Height);

\coordinate (BF) at (\Fraction*\Depth,\Width,\Height);
\coordinate (EF) at (\Depth,\Width,\Fraction*\Height);
\coordinate (GF) at (\Depth,\Fraction*\Width,\Height);

\coordinate (T0) at (1.5,1.5,2);

\coordinate (BGF) at (1.5,1.5,2);
\coordinate (EGF) at (2,1.5,1.5);

\coordinate (Aa) at (0,\Width-1.7,0);
\coordinate (Ba) at (0,\Width-1.7,\Height);
\coordinate (Ea) at (\Depth,\Width-1.7,0);
\coordinate (BFa) at (\Fraction*\Depth,\Width-1.7,\Height);
\coordinate (EFa) at (\Depth,\Width-1.7,\Fraction*\Height);
\coordinate (GFa) at (\Depth,\Fraction*\Width-1.7,\Height);

\coordinate (BGFa) at (1.3,0,2);
\coordinate (EGFa) at (2,0,1.3);

\draw[blue,fill=yellow!80] (O) -- (C) -- (BGFa) -- (EGFa) -- (D) -- cycle;%
\draw[blue,fill=blue!30] (O) -- (Aa) -- (Ea) -- (D) -- cycle;%
\draw[blue,fill=red!10] (O) -- (Aa) -- (Ba) -- (C) -- cycle;%
\draw[blue,fill=red!20,opacity=0.8] (Aa) -- (Ba) -- (BFa) -- (EFa) -- (Ea) -- cycle;%

\draw[blue,fill=blue!40,opacity=1] (BFa) -- (EFa) -- (EGFa) -- (BGFa);

\draw[blue] (O) -- (C) -- (G) -- (D) -- cycle;%
\draw[blue] (O) -- (A) -- (E) -- (D) -- cycle;%
\draw[blue] (O) -- (A) -- (B) -- (C) -- cycle;%
\draw[blue,opacity=0.8] (D) -- (E) -- (EF) -- (GF) -- (G) -- cycle;%
\draw[blue,opacity=0.6] (C) -- (B) -- (BF) -- (GF) -- (G) -- cycle;%
\draw[blue,opacity=0.8] (A) -- (B) -- (BF) -- (EF) -- (E) -- cycle;%

\coordinate (AA) at (0,1.4*\Width,0);
\coordinate (CC) at (0,0,1.4*\Height);
\coordinate (DD) at (1.4*\Depth,0,0);

\draw[->] (O) -- (AA);
\draw[->] (O) -- (CC);
\draw[->] (O) -- (DD);

\draw (2,0,0) node[anchor=north west] {$\alpha^*_1$};
\draw (0,0.3,0) node[anchor=south west] {$\alpha^*_3$};
\draw (0,0,2) node[anchor=north west] {$\alpha^*_2$};

\draw[red,thick,dotted] (0,1.8,1.8) -- (1.8,1.8,1.8) node [anchor=south ]{${t}^{\nu'}$} -- (1.8,1.8,0) ;
\draw[red,thick,dotted] (1.8,0,1.8) -- (1.8,1.8,1.8) ;

\draw[red,thick,dotted] (0,0,1.8) -- (1.8,0,1.8) node [anchor=north]{${t^*}^\nu$} -- (1.8,0,0) ;

\end{tikzpicture}

\caption{The figure is a high-level illustration of the main argument
  of the proof of Lemma~\ref{lemma:linearBoundary} for $k=3$. The left
  figure corresponds to a \slicedoffx\ box, for $\bar s$ while the right figure
  illustrates the shift of the boundaries when $s_3$ is decreased to $s^*_3$,
  which leads to a contradiction, as the last figure is not a box for
  tasks $1,2$.}
\label{fig:3dshift-induction-step}
\end{figure}
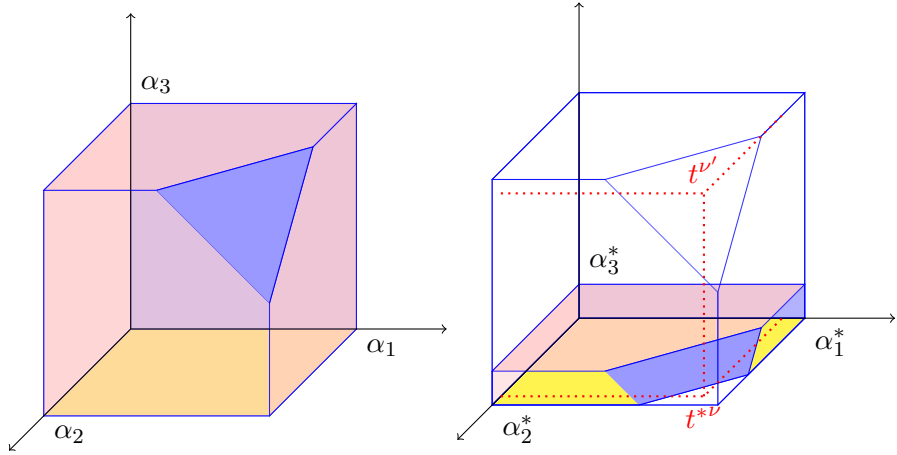

\paragraph{Case 2: Half-bundling boundary functions.}

We consider the case that for some $T_\nu(P)\leq \hat T \leq
T_{\nu'}(P),$ the slice mechanism $(p_k,p_k')$ is non-crossing, and so
that it has a boundary half-bundling at $p_k$ (or even fully bundling
at both $p_k$ and $p_k'$). Recall that this includes one-dimensional
bundling mechanisms as well as relaxed affine minimizers that might
become nonlinear for small $s_{p_k}.$ Using a more technically
involved argument, it is possible to show that even in this case, the
boundary function $\psi_k(s_{p_k})$ must be linear, or else the
(whole) mechanism has high approximation factor.\footnote{Exploiting
  the fact that $R_{\{p_k,p_k'\}}$ must be crossing, and $T$ is in
  general position, one can argue similarly as in
  Lemma~\ref{prop:tech2dim}. } However, for the sake of simpler
presentation, instead we argue the same way as in \cite{CKK21b}: For
arbitrary \emph{fixed} $T_\nu(P)\leq \hat T \leq T_{\nu'}(P),$ we
exclude that the slice mechanism $(p_k,p_k')$ is half-bundling at
$p_k,$ \emph{for many } different siblings $p_k'\in Q_k.$ The next
lemma is a variant of (\cite{CKK21b} Lemma 23).

\begin{lemma}\label{lem:noHalfBundling} Suppose that $P=\{p_1,p_2,\ldots ,p_k\}$ is a \slicedoffx\ box for standard instance $T=(t, \bar s),$ and corresponding leaf set $\mathcal C=\{Q_1, Q_2, \ldots Q_k\},$ and the approximation ratio is less than $n.$ Let 
  $T_\nu(P)\leq \hat T \leq T_{\nu'}(P),$ be an arbitrary input. Then
  fewer than $n^2/\xi$ siblings $p_k'\in Q_k$ of $p_k$ have a $2\times
  2$ slice mechanism $(p_k, p_k')$ in $\hat T$ that is half-bundling
  at $p_k.$
\end{lemma}

\begin{proof} Let $\hat T=((\hat t_P, t_{-P}), \bar s),$ (since by
  definition $t_{-P}=\hat t_{-P}$). We fix and omit $\bar s$ from the
  notation in the rest of the proof.  Let $\psi_{p_k}=\psi_{p_k}[\hat
  T]$ be the critical value of $t_{p_k}$ in $\hat T.$

  Let $p_k'$ be an arbitrary sibling of $p_k,$ such that the slice
  $(p_k,p_k')$ in $\hat T$ is half-bundling at $p_k$ (see
  Figure~\ref{fig:case2-half-bundling}). Then, for the input point
  slightly above the boundary $t^+=(\hat t_{P_{-k}},
  t_{p_k}=\psi_{p_k}+\epsilon, t_{-P})$ neither of the tasks $p_k$ or
  $p_k'$ is allocated to the root. This follows from the definition of
  $\psi_{p_k},$ the fact that $t_{p_k'}=0$ (because $T$ is standard
  instance for $P$), and that the boundary for $p_k$ in the slice
  $(p_k,p_k')$ is locally bundling at $t_{p_k'}=0$ by the definition of
  half-bundling allocations (Definition~\ref{def:half-bundling}).

  Assuming by contradiction that $p_k$ has at least $n^2/\xi$
  different siblings $p_k'$ with which it is half-bundling, for input
  $t^+$ we obtain that all these $p'_k$ tasks are allocated to the
  same leaf player, incurring a cost of at least $\xi$ for each of
  them (because $Q_k$ is standard leaf, hence $\bar
  s_{p_k'}>\xi$). Therefore the makespan achieved by the mechanism is
  at least $(n^2/\xi)\cdot \xi \geq n^2$.  On the other hand, since
  tasks in $[m]\setminus P$ are trivial, we have that the makespan of the
  optimal allocation is at most $\max_{p_i\in P}\bar s_{p_i}\leq 1.$
  This would imply approximation ratio higher than $n,$
  contradiction. \end{proof}

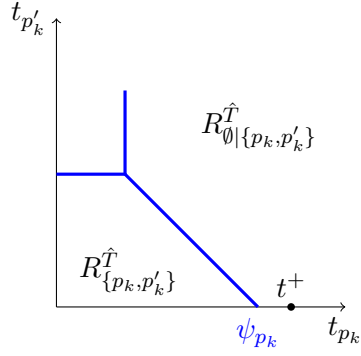
\begin{figure}
\centering
  \begin{tikzpicture}[scale=0.48]

    \draw[->] (0,0) -- (8,0) node[anchor=north] {$t_{p_k}$};
    \draw[->] (0,0) -- (0,8) node[anchor=east] {$t_{{p'_k}}$};
    \draw[very thick, blue] (1.9,6) -- (1.9,3.68) -- (5.58,0) node[anchor=north]{$\psi_{p_k}$};%
    \draw[very thick, blue] (0, 3.68) -- (1.9,3.68) ;

\draw (2,1) node[anchor=center] {$R^{\hat T}_{\{p_k,p'_k\}}$}; 
\draw (7.5,5) node[anchor=east] {$R^{\hat T}_{\emptyset|\{p_k,p'_k\}}$};
\draw (6.5,0) node[anchor=south] {$t^+$};
 \fill (6.5cm,0cm)    circle (3pt);

  \end{tikzpicture}

  \caption{Proof of Lemma~\ref{lem:noHalfBundling}: The slice $(p_k,p_k')$ in $\hat T$ is half-bundling at
    $p_k$. For input $(t^+,\bar s)$ neither $p_k$ nor $p'_k$ is allocated to the root.}
\label{fig:case2-half-bundling}
\end{figure}

\subsubsection{The main inductive lemma.}
\label{sec:ind-main-lemma}
The results of Section~\ref{sec:ind-p_k-p'k-slice} culminate in the following lemma.

\begin{lemma} \label{lemma:many-good-sets} Suppose that
  $P=(p_1,\ldots,p_k)$ is a \slicedoffx\ box for $T.$ Then either $P$ is a box for $T$ or %
  the number of different siblings $p_k'$ of $p_k$ for which 
$(P_{-k},p_k')$ is a box for $T$,  is at least $\ell-2n^3/\xi.$
\end{lemma}

\begin{proof}
  Suppose that $P$ is not a box, and take a random sibling $p_k'$ of $p_k.$
  For $j=0, 1, \ldots k-1,$ we fix  instances $T_j, T_j^* $  such that 
  \begin{itemize}
  
  \item[(i)] $T_\nu \leq T_j\leq T_j^*\leq T_{\nu'}$ and all $t$-coordinates of $T_j, \,T_j^*$ are rational;

  \item[(ii)] $t^{j*}_{p_i}=t^{j}_{p_i}+\epsilon$ for some small $\epsilon$ for every $i\in[k];$
  \item[(iii)] the relative position of $T_0, T_1, \ldots T_{k-1}$ is as defined in the statement of Lemma~\ref{lemma:linearBoundary}, moreover this property holds even if we replace any subset of the $T_j$ by the respective $T_j^*$ instances.
  \end{itemize}
  Observe first that such instances exist:  in  Lemma~\ref{lemma:linearBoundary} $h=(5/8)\cdot 4^{k}\nu,$ and $\eta$ is arbitrarily small; on the other hand, $t^{\nu'}_{p_i}-t^{\nu}_{p_i}=(4^k-4^{k-1})\nu=(3/4)4^k\nu>(5/8)4^k\nu.$ So there is enough space for instances in the prescribed relative position.

  By Lemma~\ref{lem:noHalfBundling} and using the union bound, the
  number of siblings $p_k'\in Q_k\setminus \{p_k\}$ for which the
  $(p_k,p_k')$ slice mechanism is half-bundling in at least one of
  these $2k$ instances, is less than $2kn^2/\xi\leq 2n^3/\xi-1.$
  Therefore, most of the siblings $p'_k$ (at least $\ell-1-(2n^3/\xi-1)=\ell-2n^3/\xi$ of them)
  do not have a half-bundling boundary with $p_k$ in any of the
  $2k$ fixed instances.  We focus now on such a sibling $p_k'.$ By
  Lemma~\ref{lemma:linearBoundary}, there must exist at least one
  $j\in\{0,1,2,\ldots, k-1\}$ such that in both $T_j$ and $T_j^*$ the
  $\psi[T_j]_{p_k}(s_{p_k})$ and $\psi[T_j^*]_{p_k}(s_{p_k})$
  functions are not truncated linear. Altogether, this excludes
  (non-task independent) relaxed affine minimizer mechanisms,
  constant, and one-dimensional bundling mechanisms as well as
  discontinuities in these two points. The remaining possibility is
  that the allocation of the slice $(p_k,p_k')$ is crossing for both
  $T_j$ and $T_j^*.$ Then, by Lemma~\ref{lem:crossing} the star
  $(P_{-k}, p_k')$ is a box for $T.$\end{proof}

\subsection{Existence of a box of size $n-1.$}
\label{sub:existence}
We are now ready to show the main result of this section
(Corollary~\ref{cor:bound-on-b}) that proves existence of a box of
size $n-1$. The proof is by induction on the size of the box, and the
next definition is handy to define `bad' events on $k$ leaves.

\begin{definition}[Probability $b_k$] Fix a mechanism, and let $1\leq k\leq n-1.$ Suppose that $T$ is a standard instance for a set $\mathcal C$ of $k$ leaves. We denote by $b(T,\mathcal C)$ the probability that a random star $P$ of $k$ tasks from $\mathcal C$ is not a box. Let $b_k$ denote the supremum of all possible $b(T,\mathcal C)$ for all choices of $T$ and $\mathcal C$ with $|\mathcal C|=k$ for the given mechanism.
\end{definition}
Next we upper bound the probability that a random star $P$ of $k$ tasks is not a \slicedoff\ box, in terms of  $b_{k-1}.$

\begin{lemma} \label{lemma:potentially-good} Fix any instance $T$
  which is standard for a set $\cal C$ of $k\geq 3$ leaves, and let
  $P$ be a random star of $k$ tasks from $\cal C$. The
  probability that $P$ is not a \slicedoff\ box for $T$ is at most $(5n/\nu-1)b_{k-1}.$
 
\end{lemma}

\begin{proof} Take a random star $P=\{p_1, p_2,\ldots , p_k\}$ of $k$ tasks from $\mathcal C.$ For each $i\in [k]$ the probability that $P_{-i}$ is not a box for $T$ is at most $b_{k-1}.$ Also the probability that $P_{-k}$ is not a box when we set task $p_k$ to $(t_{p_k}=0, s_{p_k}=\bar s_{p_k}- q\nu/(4n)),$ is at most $b_{k-1}$ for (at most) each of  $q= 1, \ldots , 4n/\nu. $ By the union bound, the probability that some of these bad events happens is at most $(k+4n/\nu)b_{k-1}< (n+4n/\nu)b_{k-1}<(5n/\nu-1)b_{k-1}.$  \end{proof}

Finally, the next lemma combines the obtained probabilities.

\begin{lemma} \label{lemma:recurrence}
  Fix a mechanism with approximation ratio less than $n,$ then for $3\leq k\leq n-1$ $$b_k\leq \left (\frac{5n}{\nu}-1 \right )b_{k-1}+ \frac{2n^3}{\xi \sqrt{\ell} }.$$ 
\end{lemma}

\begin{proof} Let $T$ be a standard instance for some set of leaves $\mathcal C=\{Q_1,\ldots , Q_k\},$ and let $\mathcal C_{-k}=\{Q_1, \ldots Q_{k-1}\}.$ Consider a star $P^*=\{p_1,\ldots, p_{k-1}\}$ from $\mathcal C_{-k}.$ %
  Let's call the sets $P^*\cup\{p_k\}$, $p_k\in Q_k$, extensions of $P^*$.
  
  Let $0\leq y\leq \ell$ be the number of the extensions of $P^*$ that are
  \slicedoff\ boxes. How many of these extensions are boxes? Either all $y$
  extensions are boxes, in which case the number of box extensions is at least
  $y$ or some extension is not a box and $\ell-2n^3/\xi$ other extensions are boxes
  (Lemma~\ref{lemma:many-good-sets}). Therefore the number of extensions that
  are boxes is at least  $\min\{y, \ell-2n^3/\xi\}\geq y (\ell-2n^3/\xi)/\ell= y (1-2n^3/\ell\xi).$ Since this holds for every star $P^*$, the probability that a random star of size $k$ is a box is at
  least $1-2n^3/\xi\ell$ times the probability that a random star of size
  $k$ is a \slicedoff\ box.
  
  For a rigorous proof, let ${\mathcal R}$ be the set of different stars of $k-1$ tasks from $\mathcal C_{-k}.$ 
We select a random star $P$ of $k$ tasks from $\mathcal C,$ and define the events:
$\quad Y:\,$ "$P$ is a \slicedoff\ box"; $\quad X:\,$  "$P$ is a box"; $\quad E_{P^*}:\,$  "$P_{-k}=P^*\,$".  Note that  in general $X\not\subset Y.$
By the above argument $\mathbb P(X|E_{P^*})\geq (1-2n^3/\xi\ell)\cdot \mathbb P(Y|E_{P^*}).$ Therefore we obtain

\begin{align*}\mathbb P(X)=\sum_{P^*\in \mathcal R} \mathbb P(X| E_{P^*})\mathbb P(E_{P^*})\\
\geq (1-2n^3/\xi\ell) \sum_{P^*\in \mathcal R} \mathbb P(Y|E_{P^*}) \mathbb P(E_{P^*})\\
=(1-2n^3/\xi\ell)\cdot \mathbb P(Y).\end{align*}
By Lemma~\ref{lemma:potentially-good}, the
  probability that a random star of size $k$ is a \slicedoff\ box is at
  least $1-(5n/\nu - 1) b_{k-1}; $ so we get $$\mathbb P(X)\geq \left (1-\frac{2n^3}{\xi\ell}\right )\left [1-\left (\frac{5n}{\nu} - 1\right ) b_{k-1}\right ]> 1-\frac{2n^3}{\xi\ell}-\left (\frac{5n}{\nu} - 1\right ) b_{k-1}.$$  
This further yields $$b_k=1-\mathbb P(X)< \left (\frac{5n}{\nu} - 1\right ) b_{k-1}+ \frac{2n^3}{\xi\ell} < \left (\frac{5n}{\nu} - 1\right ) b_{k-1}+ \frac{2n^3}{\xi\sqrt{\ell}}.$$
\end{proof}

\begin{corollary} \label{cor:bound-on-b}
  Fix a mechanism with approximation ratio less than $n,$ then for $2\leq k\leq n-1$ $$b_k\leq \left (\frac{5n}{\nu}\right )^{k-2}\cdot \frac{2n^3}{\xi \sqrt{\ell} }.$$ 
\end{corollary}
\begin{proof} For $k=2,$ using Theorem~\ref{cor:b2}  $b_2\leq 3/\sqrt{\ell}< 2n^3/(\xi\sqrt{\ell}),$ so the base case holds. Now for $k\geq 3$ by induction we obtain
\begin{align*}
b_k & \leq \left (\frac{5n}{\nu}-1 \right )b_{k-1}\,+ \,\frac{2n^3}{\xi \sqrt{\ell} }\\
& \leq 
\left (\frac{5n}{\nu}\,-\, 1 \right )\left (\frac{5n}{\nu}\right )^{k-3}\cdot \frac{2n^3}{\xi \sqrt{\ell} }\,+ \,\frac{2n^3}{\xi \sqrt{\ell} }\\
& = 
\left (\left (\frac{5n}{\nu}\right )^{k-2}\,-\,\left (\frac{5n}{\nu}\right )^{k-3}\,+\,1 \right )\cdot \frac{2n^3}{\xi \sqrt{\ell} }\\
& \leq  
\left (\frac{5n}{\nu}\right )^{k-2}\cdot \frac{2n^3}{\xi \sqrt{\ell} }.
\end{align*}
\end{proof}

Setting $k=n-1,$ and using $\nu<\xi,$ we obtain the following rough lower bound for $\ell$.

\begin{corollary} If $\ell >\left (\frac{5n}{\nu}\right )^{2n},$ we get $b_{n-1}<1,$ and there exists at least one box of size $n-1$ for every instance $T=(t,\bar s)$ standard for $\{C_i\}_{i\in [n-1]}.$ 
\end{corollary}